\documentclass[journal,nofonttune]{IEEEtran}         
\RequirePackage[l2tabu,orthodox]{nag}		

\usepackage{url}

\usepackage{amsmath} 
\usepackage{amssymb}
\usepackage{mathtools}
\usepackage{bbm,bm}
\usepackage{amsthm,amsfonts}
\usepackage[shortlabels,inline]{enumitem}
\usepackage[ruled,lined,noend]{algorithm2e}
\usepackage{empheq}
\usepackage{booktabs}

\theoremstyle{plain}
\newtheorem{theorem}{Theorem}		
\newtheorem{corollary}{Corollary}		
\newtheorem{lemma}{Lemma}		

\newtheorem{conjecture}{Conjecture}		

\newtheorem*{corollary*}{Corollary}		

\theoremstyle{definition}

\newtheorem*{definition*}{Definition}		
\newtheorem*{assumption*}{Assumptions}		

\theoremstyle{remark}
\newtheorem{remark}{Remark}		

\newtheorem*{remark*}{Remark}		
\newtheorem*{example*}{Example}		

\makeatletter
\newcommand{\mainnl}{%
  \refstepcounter{AlgoLine}%
  \nlset{\arabic{AlgoLine}:}%
  }
\newcommand{\routinenl}{%
  \refstepcounter{AlgoLine}%
  \nlset{\arabic{AlgoLine}:}
  }
\makeatother

\usepackage[varg]{txfonts}
\let\mathbb=\varmathbb

\usepackage{dsfont}		

\usepackage[scaled]{inconsolata}

\usepackage[
cal=cm,
]
{mathalfa}

\usepackage{acronym}		
\renewcommand*{\aclabelfont}[1]{\acsfont{#1}}		
\newcommand{\acdef}[1]{\textit{\acl{#1}} \textup{(\acs{#1})}\acused{#1}}		

\usepackage{newfloat}
\usepackage[labelfont={bf,small},labelsep=colon,font=small]{caption}	

\usepackage[dvipsnames,svgnames]{xcolor}		
\colorlet{MyBlue}{MediumBlue}
\colorlet{MyGreen}{DarkGreen!85!Black}
\definecolor{Powder}{HTML}{FEFEFA}
\definecolor{Cloud}{HTML}{F5F5F5}
\definecolor{charcoallight}{HTML}{707070}
\definecolor{tealblue}{HTML}{275666}
\definecolor{oldlavender}{rgb}{0.47, 0.41, 0.47}
\definecolor{mediumpersianblue}{rgb}{0.0, 0.4, 0.65}
\definecolor{mediumelectricblue}{rgb}{0.01, 0.31, 0.59}
\definecolor{egyptianblue}{rgb}{0.0, 0.3, 0.6}

\colorlet{linkcolor}{black}

\newcommand{\para}[1]{\smallskip\paragraph*{\textbf{#1}}}

\usepackage{subcaption}		

\usepackage{latexsym}		
\usepackage{xspace}		

\usepackage[numbers,sort&compress]{natbib}		

\bibpunct[, ]{[}{]}{,}{n}{}{,}

\usepackage{hyperref} 
\hypersetup{
colorlinks=true,
linktocpage=true,
pdfstartview=FitH,
breaklinks=true,
pdfpagemode=UseNone,
pageanchor=true,
pdfpagemode=UseOutlines,
plainpages=false,
bookmarksnumbered,
bookmarksopen=false,
bookmarksopenlevel=1,
hypertexnames=true,
pdfhighlight=/O,
urlcolor=linkcolor,
linkcolor=linkcolor,
citecolor=linkcolor,
pdftitle={},
pdfauthor={},
pdfsubject={},
pdfkeywords={},
pdfcreator={pdfLaTeX},
pdfproducer={LaTeX with hyperref}
}

\usepackage[sort&compress,capitalize,nameinlink]{cleveref}		
\crefname{app}{Appendix}{Appendices}
\crefname{algorithm}{Alg.}{Algs.}
\crefname{algorithmenv}{Alg.}{Algs.}
\crefrangeformat{equation}{\upshape(#3#1#4)\textendash(#5#2#6)}
\creflabelformat{part}{(#2#1#3)}

\usepackage[textwidth=15mm]{todonotes}		

\newcommand{\revise}[1]{#1}		

\def\beginrev{}		
\def\endedit{\color{black}}		

\newcounter{proofpart}

\DeclarePairedDelimiter{\bracks}{[}{]}		

\DeclarePairedDelimiter{\abs}{\lvert}{\rvert}		

\DeclarePairedDelimiterX{\setdef}[2]{\{}{\}}{#1:#2}		
\DeclarePairedDelimiterXPP{\exclude}[1]{\mathopen{}\setminus}{\{}{\}}{}{#1}

\newcommand{\cf}{cf.\xspace}		
\newcommand{\eg}{e.g.,\xspace}		
\newcommand{\ie}{i.e.,\xspace}		
\newcommand{\textpar}[1]{\textup(#1\textup)}		

\newcommand{\txs}{\textstyle}		

\DeclareMathOperator{\bigoh}{\mathcal O}		
\DeclareMathOperator{\poly}{poly}		

\newcommand{\eps}{\varepsilon}		

\newcommand{\mgeq}{\succcurlyeq}

{
\renewcommand{\Lipschitz}{L_\genericfunction}%
\renewcommand{\gradlipschitz}{L}%
}%
{}

\usepackage{verbatim}   
\newenvironment{biographies}{\expandafter\comment}{\expandafter\endcomment}
 
\newcommand{\obsolete}[1]{}
\newcommand{\nocolsep}{\arraycolsep=1.4pt\def\arraystretch{1}}
\makeatletter
\newcommand{\vast}{\bBigg@{4}}
\newcommand{\Vast}{\bBigg@{5}}
\makeatother
\newcommand{\Real}{\mathbb{R}}
 
\newcommand{\Realplus}{\Real_{>0}}

\newcommand{\Complex}{\mathbb{C}}

\newcommand{\hermfunction}{\textup{Herm}}

\newcommand{\Herm}[1]{\hermfunction\left(#1\right)}
\newcommand{\setst}{:}
\newcommand{\lebesguemeasure}{\mu}
\newcommand{\lebesgue}[1]{\lebesguemeasure(#1)}

\newcommand{\grad}{\nabla}

\newcommand{\gradk}[1]{\grad_{#1}}

\newcommand{\projgrad}{\scalebox{1}[1.22]{\rotatebox[origin=c]{180}{$\triangle$}}}

\newcommand{\projgradk}[1]{\projgrad_{#1}}

\newcommand{\inner}[2]{\langle{#1},{#2}\rangle}

\newcommand{\subdifffunction}[1]{\delta{#1}}
\newcommand{\subdiff}[2]{\subdifffunction{#1}(#2)}

\newcommand{\fenchelcoupling}{F}
\newcommand{\fenchel}[2]{\fenchelcoupling({#1},#2)}

\newcommand{\expectationfunction}{\mathbb{E}}
\newcommand{\expectation}[1]{\expectationfunction[#1]}
\newcommand{\Expectation}[1]{\expectationfunction\left[#1\right]}
\newcommand{\bigexpectation}[1]{\expectationfunction\big[#1\big]}

\newcommand{\probafunction}{\mathbb{P}}
\newcommand{\proba}[1]{\probafunction(#1)}
\newcommand{\Proba}[1]{\probafunction\left(#1\right)}
\newcommand{\bigproba}[1]{\probafunction\big(#1\big)}
\newcommand{\Bigproba}[1]{\probafunction\Big(#1\Big)}

\newcommand{\vect}[1]{(#1)}

\newcommand{\factorial}[1]{{#1}!}

\usepackage{scalerel}
\newlength\bshft
\def\fakebold#1{\ThisStyle{\ooalign{$\SavedStyle#1$\cr%
  \kern-\bshft$\SavedStyle#1$\cr%
  \kern\bshft$\SavedStyle#1$}}}

\newcommand{\trmsphere}{\maketr{\mathbb{S}}}
\newcommand{\trmspherei}[1]{\trmsphere_{#1}}

\newcommand{\trmball}{\maketr{\mathbb{B}}}
\newcommand{\trmballi}[1]{\trmball_{#1}}

\newcommand{\volume}[1]{\textup{vol}(#1)}
\newcommand{\Volume}[1]{\textup{vol}\left(#1\right)}
\newcommand{\zero}[1]{o(#1)}

\newcommand{\magnitude}[1]{\mathcal{O}(#1)}
\newcommand{\Magnitude}[1]{\mathcal{O}\left(#1\right)}

\newcommand{\bigtheta}[1]{\Theta(#1)}
\newcommand{\Bigtheta}[1]{\Theta\left(#1\right)}

\newcommand{\cond}{|}

\newcommand{\as}{\textup{a.s.}}
\newcommand{\As}{\textup{A.s.}}
\newcommand{\indicatorfunction}[1]{\mathbbm{1}_{#1}}
\newcommand{\indicator}[2]{\indicatorfunction{#1}(#2)}

\newcommand{\ceil}[1]{\lceil{#1}\rceil}

\newcommand{\maketr}[1]{\bar{#1}}
\newcommand{\twitchoperatorfunction}{T}

\newcommand{\twitchkoperator}[1]{\makeuserk{\twitchoperatorfunction}{#1}}
\newcommand{\twitchk}[3]{\twitchkoperator{#1}(#2,#3)}

\newcommand{\twitchFlipschitz}{\Lambda_{\text{F}}}

\newcommand{\supermartfunction}{S}
\newcommand{\supermartn}[1]{\makeiteraten{\supermartfunction}{#1}}
\newcommand{\lyapunovfunction}{\mathcal{L}}
\newcommand{\lyapunovprobfunction}[1]{\lyapunovfunction_{#1}}
\newcommand{\lyapunov}[1]{\lyapunovfunction({#1})}
\newcommand{\lyapunovprob}[2]{\lyapunovprobfunction{#1}({#2})}
\newcommand{\Lyapunov}[1]{\lyapunovfunction\left({#1}\right)}

\newcommand{\defeq}{:=}

\newcommand{\tracefunction}{\textup{tr}}
\newcommand{\trace}[1]{\tracefunction(#1)}

\newcommand{\exponentialfunction}{\exp}
\newcommand{\exponential}[1]{\exponentialfunction(#1)}
\newcommand{\Exponential}[1]{\exponentialfunction\left(#1\right)}

\newcommand{\Bigexponential}[1]{\exponentialfunction\Big(#1\Big)}

\newcommand{\brackexponential}[1]{\exponentialfunction[#1]}

\newcommand{\logarithm}[1]{\textup{log}(#1)}

\newcommand{\powlogarithm}[2]{\textup{log}^{#2}(#1)}

\newcommand{\norm}[1]{\|#1\|}
\newcommand{\Norm}[1]{\left\|#1\right\|}

\newcommand{\dualsymbol}{*}
\newcommand{\dualnorm}[1]{\norm{#1}_{\dualsymbol}}
\newcommand{\dualNorm}[1]{\Norm{#1}_{\dualsymbol}}

\newcommand{\frobeniuslabel}{\textup{F}}
\newcommand{\frobeniusnorm}[1]{\|#1\|_{\frobeniuslabel}}
\newcommand{\frobeniusNorm}[1]{\left\|#1\right\|_{\frobeniuslabel}}
\newcommand{\modulus}[1]{|#1|}

\newcommand{\makeuserk}[2]{#1_{#2}}
\newcommand{\makeiteraten}[2]{#1_{#2}} 
\newcommand{\makeuserkiteraten}[3]{{#1}_{#2,#3}}

\newcommand{\msg}{\mathbf{x}}
\newcommand{\msgk}[1]{\makeuserk{\msg}{#1}}
\newcommand{\signal}{\mathbf{y}}

\newcommand{\noise}{\mathbf{z}}

\newcommand{\chmat}{\mathbf{H}}
\newcommand{\chmatk}[1]{\makeuserk{\chmat}{#1}}
\newcommand{\users}{\mathcal{K}}
\newcommand{\kk}{K}
\newcommand{\nbk}{k}
\newcommand{\altaltnbk}{j}
\newcommand{\altnbk}{\nbk^{\prime}}
\newcommand{\nn}{N}

\newcommand{\mm}{M}

\newcommand{\mmk}[1]{\mm_{#1}}

\newcommand{\identity}{\mathbf{I}}
\newcommand{\identityd}[1]{\identity_{#1}}
\newcommand{\identityk}[1]{\makeuserk{\identity}{#1}}

\newcommand{\channel}{\mathbf{H}}
\newcommand{\channelk}[1]{\makeuserk{\channel}{#1}}
\newcommand{\Mat}{\mathcal{Q}}
\newcommand{\trMat}{\maketr\Mat}

\newcommand{\matsolfunction}{Q}
\newcommand{\matsol}[1]{\matsolfunction({#1})}
\newcommand{\matsolkfunction}[1]{\matsolfunction_{#1}}
\newcommand{\matsolk}[2]{\matsolkfunction{#1}({#2})}

\newcommand{\aggregatematsolfunction}{\matsolfunction}
\newcommand{\aggregatematsol}[1]{\aggregatematsolfunction{#1}}

\newcommand{\multiplier}{\gamma}

\newcommand{\mat}{\mathbf{Q}}

\newcommand{\matn}[1]{\makeiteraten{\mat}{#1}}
\newcommand{\matk}[1]{\makeuserk{\mat}{#1}}
\newcommand{\matkn}[2]{\makeuserkiteraten{\mat}{#1}{#2}}

\newcommand{\Dualmat}{\mathcal{Y}}
\newcommand{\Dualmatk}[1]{\makeuserk{\Dualmat}{#1}}

\newcommand{\dualmat}{\mathbf{Y}}
\newcommand{\dualmatn}[1]{\makeiteraten{\dualmat}{#1}}
\newcommand{\dualmatk}[1]{\makeuserk{\dualmat}{#1}}
\newcommand{\dualmatkn}[2]{\makeuserkiteraten{\dualmat}{#1}{#2}}

\newcommand{\altdualmat}{\dualmat^\prime}

\newcommand{\altmat}{\mat^\prime}

\newcommand{\altmatk}[1]{\makeuserk{\altmat}{#1}}

\newcommand{\meanmat}{\bar{\mathbf{Q}}}
\newcommand{\meanmatn}[1]{\makeiteraten{\meanmat}{#1}}

\newcommand{\trZmat}{\maketr{\mathcal{Z}}}
\newcommand{\trZmatk}[1]{\makeuserk{\trZmat}{#1}}

\newcommand{\zmat}{\mathbf{Z}}

\newcommand{\zmatn}[1]{\makeiteraten{\zmat}{#1}}
\newcommand{\zmatkn}[2]{\makeuserkiteraten{\zmat}{#1}{#2}}

\newcommand{\trzmat}{\maketr\zmat}

\newcommand{\trzmatk}[1]{\makeuserk{\trzmat}{#1}}

\newcommand{\trzmatkn}[2]{\makeuserkiteraten{\trzmat}{#1}{#2}}
\newcommand{\cmat}{\mathbf{C}}

\newcommand{\cmatk}[1]{\makeuserk{\cmat}{#1}}

\newcommand{\stepsize}{\gamma}
\newcommand{\stepsizen}[1]{\stepsize_{#1}}

\newcommand{\probas}{\Pi}
\newcommand{\probasU}[1]{\makeuserk{\probas}{#1}}
\newcommand{\prob}{\pi}
\newcommand{\probk}[1]{\makeuserk{\prob}{#1}}

\newcommand{\Updating}{U}
\newcommand{\Updatingn}[1]{\makeiteraten{\Updating}{#1}}

\newcommand{\radius}{\rho}
\newcommand{\radiusk}[1]{\makeuserk{\radius}{#1}}

\newcommand{\deltan}[1]{\delta_{#1}}
\newcommand{\bias}{\mathbf{B}}

\newcommand{\biaskn}[2]{\makeuserkiteraten{\bias}{#1}{#2}}
\newcommand{\devmat}{\mathbf{U}}

\newcommand{\devmatkn}[2]{\makeuserkiteraten{\devmat}{#1}{#2}}

\newcommand{\momega}{\mathbf{\omega}}

\newcommand{\zmatk}[1]{\makeuserk{\zmat}{#1}}

\newcommand{\optMat}{\Mat^\star}
\newcommand{\optmat}{\mat^\star}

\newcommand{\optrate}{\ratefunction^{\ast}}
\newcommand{\optmatk}[1]{\makeuserk{\optmat}{#1}}

\newcommand{\ratevalue}{\hat{R}}
\newcommand{\ratefunction}{R}

\newcommand{\rate}[1]{\ratefunction({#1})}

\newcommand{\ratekfunction}[1]{\makeuserk{\ratefunction}{#1}}
\newcommand{\ratek}[2]{\ratekfunction{#1}({#2})}

\newcommand{\DAfunction}{\textup{\texttt{MXL}}}
\newcommand{\DAkfunction}[1]{\makeuserk{\DAfunction}{#1}}

\newcommand{\DAk}[2]{\DAkfunction{#1}({#2})}
\newcommand{\DAplusfunction}{\textup{\texttt{MXL+}}}
\newcommand{\DApluskfunction}[1]{\makeuserk{\DAplusfunction}{#1}}

\newcommand{\DAplusk}[2]{\DApluskfunction{#1}({#2})}

\newcommand{\MDSfunction}{Z}
\newcommand{\MDSn}[1]{\MDSfunction_{#1}}
\newcommand{\randomMDSfunction}{\altMDSfunction}
\newcommand{\randomMDSn}[1]{\randomMDSfunction_{#1}}

\newcommand{\altMDSfunction}{X}

\newcommand{\lipschitzkl}[2]{\lipschitz_{#1#2}}

\newcommand{\compactset}{C}
\newcommand{\compactsetk}[1]{\compactset_{#1}}

\newcommand{\estgradratefunction}{\hat{V}} 
\newcommand{\estgrad}{\hat{\mathbf{V}}} 
\newcommand{\estgradn}[1]{\makeiteraten{\estgrad}{#1}}
\newcommand{\estgradk}[1]{\makeuserk{\estgrad}{#1}}
\newcommand{\estgradkn}[2]{\makeuserkiteraten{\estgrad}{#1}{#2}}
\newcommand{\partialestgrad}{\hat{\mathbf{V}}} 
\newcommand{\partialestgradn}[1]{\makeiteraten{\partialestgrad}{#1}}

\newcommand{\partialestgradkn}[2]{\makeuserkiteraten{\partialestgrad}{#1}{#2}}

\newcommand{\estgradratekfunction}[1]{\makeuserk{\estgradratefunction}{#1}}
\newcommand{\estgradratek}[2]{\estgradratekfunction{#1}({#2})}

\newcommand{\offset}{\rho}
\newcommand{\offsetk}[1]{\makeuserk{\offset}{#1}}
\newcommand{\offsetn}[1]{\makeiteraten{\offset}{#1}}

\newcommand{\genericfunction}{f}
\newcommand{\generic}[1]{\genericfunction({#1})}
\newcommand{\dgenericfunction}{\genericfunction^\prime}
\newcommand{\dgeneric}[1]{\dgenericfunction({#1})}

\newcommand{\regfunction}{h}
\newcommand{\reg}[1]{\regfunction({#1})}

\newcommand{\conjregfunction}{\regfunction^*}
\newcommand{\conjreg}[1]{\conjregfunction({#1})}

\newcommand{\power}{P}
\newcommand{\powerk}[1]{\power_{#1}}

\newcommand{\contrmatfunction}{\tilde\mat}
\newcommand{\contrmatkfunction}[1]{\contrmatfunction_{#1}}
\newcommand{\contrmatkdeltaZ}[3]{\contrmatkfunction{#1}^{#2}(#3)}
\newcommand{\contrmatkdeltaZk}[4]{\contrmatkfunction{#1}^{#2}_{#4}(#3)}%

\newcommand{\compactmat}{\mathbf{C}}

\newcommand{\compactmatn}[1]{\makeiteraten{\compactmat}{#1}}
\newcommand{\testmat}{\hat\mat}
\newcommand{\testmatk}[1]{\testmat_{#1}}
\newcommand{\testmatn}[1]{\makeiteraten{\testmat}{#1}}
\newcommand{\testmatkn}[2]{\makeuserkiteraten{\testmat}{#1}{#2}}
\newcommand{\contrzmat}{\mathbf{W}}
\newcommand{\contrzmatk}[1]{\contrzmat_{#1}}

\newcommand{\filtration}{\mathcal{F}}
\newcommand{\filtrationn}[1]{\makeiteraten{\filtration}{#1}}

\newcommand{\refereq}[2]{\overset{\text{\tiny{#1}}}{#2}}

\newcommand{\hidestoryboard}[1]{}
 
\newcommand{\SPSAplus}{SPSA+{}}

\newcommand{\randomMXLzeroplus}{AMXL0$^{+}${}}
\newcommand{\anyMXLzeroplus}{(UCD-)MXL0$^{+}${}}
\newcommand{\randomMXL}{AMXL0$^{+}$}
\newcommand{\CDrandomMXL}{UCD-MXL0$^{+}${}}

\renewcommand{\DAfunction}{\textup{\tt{MXL0}}} 
\renewcommand{\DAplusfunction}{\textup{\tt{MXL0}}^{\textup{\tt{+}}}}
\newcommand{\passfunction}{\textup{\tt{Pass}}}
\newcommand{\passkfunction}[1]{\passfunction_{#1}}

\def\shadezero{0.7} 
\def\shaderandom{0.9} 
\def\shadezeroplus{0.4} 
\definecolor{MXLzeroshade}{rgb}{\shadezero,\shadezero,\shadezero}
\definecolor{MXLrandomshade}{rgb}{\shaderandom,\shaderandom,\shaderandom}
\definecolor{MXLzeroplusshade}{rgb}{\shadezeroplus,\shadezeroplus,\shadezeroplus}

\def\lightness{0.7} 
\definecolor{claret}{rgb}{0.6,0.09,0.2}
\definecolor{DodgerB}{rgb}{0.118,0.565,1}
\definecolor{turquoise}{rgb}{0,\lightness,\lightness}

\definecolor{MXLonecol}{rgb}{0,0,0}
 \colorlet{MXLzerocol}{claret!70!white}
 \colorlet{MXLzeropluscol}{DodgerB}
 \colorlet{randomMXLzeropluscol}{turquoise}

\definecolor{IWFcolor}{rgb}{0,0.5,1.0}
\definecolor{SWFcolor}{rgb}{0,0,1.0}
\definecolor{MXLcolor}{rgb}{0,0,0}
\definecolor{MXL0color}{rgb}{0.8,0.1,0.3}
\definecolor{MXL+color}{rgb}{0,0.5,0.2}
\definecolor{AMXL+color}{rgb}{0,0.8,0.5}

\definecolor{darkcol}{rgb}{0.20,0.20,0.20}
\definecolor{lightcol}{rgb}{0.95,0.95,0.95}
\definecolor{envelopecol}{rgb}{0.50,0.50,0.50}

\def\prd{1.0}
\def\palaon{7pt/\prd}
\def\palaoff{1pt/\prd}
\def\palbon{6pt/\prd}
\def\palboff{2pt/\prd}
\def\palcon{5pt/\prd}
\def\palcoff{3pt/\prd}
\def\paldon{4pt/\prd}
\def\paldoff{4pt/\prd}
\def\paleon{3pt/\prd}
\def\paleoff{5pt/\prd}
\def\palfon{2pt/\prd}
\def\palfoff{6pt/\prd}
\def\palgon{1pt/\prd}
\def\palgoff{7pt/\prd}
\tikzset{dashdot/.style={
postaction={draw,randomMXLzeropluscol!100!MXLzeropluscol,dash pattern= on \palaon off \palaoff,dash phase=\palaon/2},
postaction={draw,randomMXLzeropluscol!80!MXLzeropluscol,dash pattern= on \palbon off \palboff,dash phase=\palbon/2},
postaction={draw,randomMXLzeropluscol!60!MXLzeropluscol,dash pattern= on \palcon off \palcoff,dash phase=\palcon/2},
postaction={draw,randomMXLzeropluscol!50!MXLzeropluscol,dash pattern= on \paldon off \paldoff,dash phase=\paldon/2},
postaction={draw,randomMXLzeropluscol!40!MXLzeropluscol,dash pattern= on \paleon off \paleoff,dash phase=\paleon/2},
postaction={draw,randomMXLzeropluscol!30!MXLzeropluscol,dash pattern= on \palfon off \palfoff,dash phase=\palfon/2},
postaction={draw,randomMXLzeropluscol!20!MXLzeropluscol,dash pattern= on \palgon off \palgoff,dash phase=\palgon/2},
}}

\DeclareFloatingEnvironment[fileext=frm,placement={!ht},name=Algorithm,within=none]{algorithmenv}  
\makeatletter
\renewcommand{\@algocf@capt@boxed}{above}
\makeatother
\makeatletter
\newcommand{\removelatexerror}{\let\@latex@error\@gobble}
\makeatother

    \SetKwInOut{Parameters}{\texttt{\textbf{Parameters}}}
    \SetKwInOut{Input}{Input}
    \SetKwInOut{Output}{Output}
    \SetKwFor{Repeat}{\texttt{\textbf{Repeat}}}{}{}
    \SetKwFor{Foreach}{\texttt{\textbf{For each}}}{\texttt{\textbf{do}}}{}
    \SetKwFor{Ateach}{\Ae}{\texttt{\textbf{do in parallel}}}{}
    \SetKwFor{Fctn}{\texttt{\textbf{Routine}}}{\texttt{\textbf{:}}}{end}
\def\Ae{\texttt{\textbf{For }}}
\def\Fe{\texttt{\textbf{For }}}

\def\Do{\texttt{\textbf{do }}}
\def\Get{\texttt{\textbf{Get }}}
\def\InParallel{\texttt{\textbf{in parallel}}}
\def\inlineIf{\texttt{\textbf{If }}}
\def\inlineThen{\texttt{\textbf{ then }}}
\def\inlineElse{\texttt{\textbf{ else }}}

\newcommand{\inlineAteachDIP}[2]{\Fe{#1}{ \Do{#2}\InParallel}}
\newcommand{\inlineIfElse}[3]{\inlineIf{#1}{\inlineThen{#2}}{\inlineElse{#3}}}

\def\Play{\texttt{\textbf{Transmit with }}}

\def\Set{\texttt{\textbf{Set }}}

\def\Sample{\texttt{\textbf{Sample }}}
\newcommand{\Sampleuniformly}[2]{\Sample{#1}{ \texttt{\textbf{uniformly over }}{#2}}}

\newcommand{\Drawfrom}[2]{\Draw{#1}{ \texttt{\textbf{according to }}{#2}}}
\def\Draw{\texttt{\textbf{Draw set of active users }}}
\SetKwInput{KwInit}{\texttt{\textbf{Initialization}}}
\SetKwInput{KwInitshort}{Init.}
\DontPrintSemicolon

\newcounter{eq}

\newcounter{hypothesis}
\newcommand{\hyplabel}[1]{{\refstepcounter{hypothesis}\label{#1}\tag{H\thehypothesis}}}
\newcounter{subhypothesis}

\newcommand{\trconditionstepsize}{\textup{(\refeq{nsss}a)}}
\newcommand{\trconditionsquared}{\textup{(\refeq{nsss}b)}}

\renewcommand{\makeuserk}[2]{#1_{#2}}
\renewcommand{\makeiteraten}[2]{{#1}_{#2}} 
\renewcommand{\makeuserkiteraten}[3]{{#1}_{#2,#3}}

\renewcommand{\optMat}{\Mat^*}
\renewcommand{\optmat}{\mat^*}

\renewcommand{\lyapunov}[2]{\lyapunovfunction(#2;{#1})}
\renewcommand{\lyapunovprob}[3]{\lyapunovprobfunction{#1}(#3;{#2})}
\renewcommand{\Lyapunov}[2]{\lyapunovfunction\left(#2;{#1}\right)}

\renewcommand{\aggregatematsol}[1]{\aggregatematsolfunction({#1})}

\newcommand{\limitpoints}{\mathcal{S}}
\renewcommand{\compactset}{\trMat}

\newcommand{\itr}{t}
\newcommand{\altitr}{s}
\newcommand{\altaltitr}{u}
 
\renewcommand{\frobeniuslabel}{2}
  
\newcommand{\withdagger}[1]{#1^\dagger}
\renewcommand{\matsolfunction}{\bm{\Lambda}}
 
\newcommand{\vbound}{\bar{v}}
\newcommand{\vboundk}[1]{\vbound_{#1}}
 
\renewcommand{\maketr}[1]{#1}

\newcommand{\normalize}[2]{{#2}}

\newcommand{\dimension}{d}
\newcommand{\dimensionk}[1]{\makeuserk{\dimension}{#1}}
\newcommand{\Lipschitz}{L}
\newcommand{\gradlipschitz}{\meanlipschitz}

\renewcommand{\contrmatkfunction}[1]{\contrmatfunction}

\newcommand{\strongly}{\kappa}
\renewcommand{\multiplier}{\nu}

\newcommand{\mome}{\mu}
\newcommand{\momfunction}{\tilde\mome}
\newcommand{\mom}[2]{\momfunction({#1},{#2})}

\renewcommand{\momega}{\bm{\zeta}}
\renewcommand{\lipschitzkl}[2]{\lambda_{#1#2}}
\newcommand{\meanlipschitzk}[1]{\lambda_{#1}}
\newcommand{\meanlipschitz}{\lambda}

\newcommand{\meansqdimension}{\tilde\dimension}
\newcommand{\meandimension}{\bar\dimension}

\newcommand{\maxdimension}{\hat\dimension}

\newcommand{\maxmm}{\hat\mm}

\renewcommand{\estgrad}{\mathbf{V}} %
\renewcommand{\estgradratefunction}{\estgrad} %

\newcommand{\basicestgradratek}[3]{\estgradratekfunction{#1}({#2})}
\renewcommand{\estgradratek}[4]{\estgradratekfunction{#1}({#2},#3;#4)}

\renewcommand{\maxdimension}{\dimension}
\renewcommand{\maxmm}{\mm}

\newcommand{\probdistfunction}{\upsilon}
\newcommand{\probdist}[1]{\probdistfunction_{#1}}

\newcommand{\errortolerance}{\varepsilon}
\newcommand{\errorproba}{\alpha}
\newcommand{\zmargin}{\eps}

\newcommand{\Acoefficientfunction}{A}
\newcommand{\Acoefficient}[2]{\Acoefficientfunction(#1,#2)}
\newcommand{\Bcoefficientfunction}{B}
\newcommand{\Bcoefficient}[2]{\Bcoefficientfunction(#1,#2)}
\newcommand{\Bcoefficientprobfunction}[1]{\Bcoefficientfunction_{#1}}
\newcommand{\Bcoefficientprob}[3]{\Bcoefficientprobfunction{#1}(#2,#3)}
\newcommand{\Ccoefficientfunction}{C}
\newcommand{\Ccoefficient}[2]{\Ccoefficientfunction(#1,#2)}
\newcommand{\Ccoefficientprobfunction}[1]{\Ccoefficientfunction_{#1}}
\newcommand{\Ccoefficientprob}[3]{\Ccoefficientprobfunction{#1}(#2,#3)}

\renewcommand{\trZmat}{\trMat^*}

\renewcommand{\twitchFlipschitz}{2}

\newcommand{\squaredhalftwitchFlipschitz}{}
\newcommand{\sqrthalftwitchFlipschitz}{}
\newcommand{\sqrtsqrthalftwitchFlipschitz}{}
\newcommand{\halftwitchFlipschitz}{}

\newcommand{\twopowersevensquaredtwitchFlipschitz}{2^9}
\newcommand{\twosquaredtwitchFlipschitz}{8}
\newcommand{\twotwitchFlipschitz}{4}
\newcommand{\fourtwitchFlipschitz}{8}

\newcommand{\mutehalftwitchFlipschitz}{}

\newcommand{\forcesqrthalftwitchFlipschitzoneoverstrongly}{1}

\newcommand{\renrmlzdcsteps}[1]{\phi(#1)}\newcommand{\repownrmlzdcsteps}[2]{\phi^{#2}(#1)}

\newcommand{\renrmlzdstuff}[1]{\hat\chi(#1)}
\newcommand{\renrmlzdcstepsprobfunction}[1]{\hat\phi_{#1}}
\newcommand{\renrmlzdcstepsprob}[2]{\renrmlzdcstepsprobfunction{#1}(#2)}
\newcommand{\repownrmlzdcstepsprob}[3]{\renrmlzdcstepsprobfunction{#1}^{#3}(#2)}

\newcommand{\basicvtpfunction}{\hat v_{2}}
\newcommand{\basicvtp}[2]{\basicvtpfunction(#1)}
\newcommand{\basicvopfunction}{\hat v_{1}}
\newcommand{\basicvop}[2]{\basicvopfunction(#1)}
\newcommand{\basicvopOfunction}{\hat v_{1,\offset}}
\newcommand{\basicvopO}[2]{\basicvopOfunction(#1)}
\renewcommand{\basicvtpfunction}{\hat{v}}
\renewcommand{\basicvopfunction}{v}
\renewcommand{\basicvopOfunction}{v_{\offset}}

\renewcommand{\offset}{\rho}

\newcommand{\stepsizecoefficient}[1]{\stepsize}
\newcommand{\deltacoefficient}[1]{\delta}

\newcommand{\cststepsize}{\tilde\stepsize}
\newcommand{\cstdelta}{\tilde\delta}

\renewcommand{\trmspherei}[1]{\trmsphere^{#1}}
\renewcommand{\trmballi}[1]{\trmball^{#1}}

\renewcommand{\altnbk}{\ell}

\newcommand{\globalnorm}{\eqref{globalnorm}}%
\newcommand{\globaldualnorm}{\eqref{globalnorm}}%
\newcommand{\oldfirstcondition}{(\refeq{oldcondition}b)}
\newcommand{\oldsecondcondition}{(\refeq{oldcondition}a)}

\newcommand{\cstnormbound}{\vbound}
\newcommand{\normboundsymbol}{\bar{v}}
\newcommand{\normboundfunction}[2]{\normboundsymbol_{#1,#2}}%
\newcommand{\normbound}[4]{\normboundfunction{#1}{#2}(#3,#4)}

\newcommand{\coef}[1]{\tau_{#1}}
\providecommand{\dimension}{d}
\renewcommand{\radius}{r}

\renewcommand{\Dualmat}{\trZmat}
\renewcommand{\trZmat}{\mathcal{Y}}%

\newcommand{\indicatorsign}{\mathbf{1}}
\renewcommand{\indicatorfunction}[1]{\indicatorsign_{#1}}

\begin{document}


\title{Fast Optimization \revise{with Zeroth-Order Feedback}\\in Distributed, Multi-User MIMO Systems}

\author{%
Olivier Bilenne,
Panayotis Mertikopoulos,
	\IEEEmembership{Member,~IEEE}
and 
E. Veronica Belmega,
	\IEEEmembership{Senior Member,~IEEE}
\thanks{%
O.~Bilenne and P.~Mertikopoulos are with Univ. Grenoble Alpes, CNRS, Inria, Grenoble INP, LIG, 38000 Grenoble, France;
P.~Mertikopoulos is also with Criteo AI Lab, Grenoble, France.
E.~V.~Belmega is with ETIS, CY Cergy Paris University, ENSEA, CNRS, UMR 8051, F-95000, Cergy, France.}%
\thanks{The authors are grateful for financial support from the French National Research Agency (ANR) projects ORACLESS (ANR\textendash 16\textendash CE33\textendash 0004\textendash 01) and ELIOT (ANR\textendash 18\textendash CE40\textendash 0030 and FAPESP 2018/12579\textendash 7).
This research has also received financial support from the COST Action CA 16228 `European Network for Game Theory' (GAMENET).}%
}


%



\maketitle

\newacro{NE}{Nash equilibrium}
\newacroplural{NE}{Nash equilibria}
\newacro{BS}{base station}
\newacro{5G}{fifth generation}
\newacro{6G}{sixth generation}
\newacro{CSI}{channel state information}
\newacro{CSIT}{channel state information at the transmitter}
\newacro{MAC}{multiple access channel}
\newacro{MU-MIMO}{multi-user, multiple-input and multiple-output}
\newacro{MIMO}{multiple-input and multiple-output}
\newacro{MUI}{multi-user interference-plus-noise}
\newacro{MXL}{matrix exponential learning}
\newacro{MXL0}{gradient-free matrix exponential learning}
\newacro{MXL+}[MXL0$^{+}$]{gradient-free \acs{MXL} with callbacks}
\newacro{AMXL+}[AMXL0$^{+}$]{asynchronous \acs{MXL+}}
\newacro{CD}[UCD-MXL0$^{+}$]{(uniform) ``coordinate descent''}
\newacro{SINR}{signal to interference-plus-noise ratio}
\newacro{SPSA}{simultaneous perturbation stochastic approximation}
\newacro{SPSA+}[SPSA+]{enhanced simultaneous perturbation stochastic approximation}
\newacro{SIC}{successive interference cancellation}
\newacro{SUD}{single user decoding}
\newacro{TDD}{time-division duplexing}
\newacro{WF}{water-filling}
\newacro{IWF}{iterative water-filling}
\newacro{SWF}{simultaneous water-filling}
\newacro{IWMMSE}[IW-MMSE]{iterative weighted MMSE}
\newacro{NOMA}{non-orthogonal multiple access}
\newacro{BC}{broadcast channel}

\begin{abstract}
%
%
In this paper, we develop a gradient-free optimization methodology for efficient resource allocation in Gaussian \acs{MIMO} \aclp{MAC}.
Our approach combines two main ingredients:
\begin{enumerate*}
[(\itshape i\hspace*{1pt}\upshape)]
\item
an entropic semidefinite optimization based on \ac{MXL};
and
\item
a one-shot gradient estimator which achieves low variance through the reuse of past information.
\end{enumerate*}
This novel algorithm, which we call  \acdef{MXL+}, retains the convergence speed of gradient-based methods while requiring minimal feedback per iteration\textemdash a \emph{single} scalar.
In more detail, in a \acs{MIMO} \acl{MAC} with $\kk$ users and $\mm$ transmit antennas per user, the \ac{MXL+} algorithm achieves $\eps$-optimality within $\poly(\kk,\mm)/\eps^{2}$ iterations (on average and with high probability), even when implemented in a fully distributed, asynchronous manner.
For cross-validation, we also perform a series of numerical experiments in medium- to large-scale \acs{MIMO} networks under realistic channel conditions.
Throughout our experiments, the performance of \ac{MXL+} matches\textemdash and sometimes exceeds\textemdash that of gradient-based \ac{MXL} methods, all the while operating with a vastly reduced communication overhead.
In view of these findings, the \ac{MXL+} algorithm appears to be uniquely suited for distributed massive \acs{MIMO} systems where gradient calculations can become prohibitively expensive.

\end{abstract}

\begin{IEEEkeywords}
Gradient-free optimization;
matrix exponential learning;
multi-user MIMO networks;
throughput maximization.
\end{IEEEkeywords}


  

\section{Introduction}
\label{section:introduction}

\IEEEPARstart{T}{he} deployment of \ac{MIMO} terminals at a massive scale  has been identified as one of the key enabling technologies for \ac{5G} wireless networks, and for good reason:
massive-\ac{MIMO} arrays can increase throughput by a factor of $10\times$ to $100\times$ (or more), they improve the system's robustness to ambient noise and channel fluctuations, and they bring about significant latency reductions over the air interface \cite{LETM14,ABC+14}.
Moreover, ongoing discussions for the evolution of \ac{5G} envision the deployment of advanced \ac{MIMO} technologies at an even larger scale in order to reach the throughput and spectral efficiency required for ``speed of thought'' connectivity \cite{B5G15,6G19}.

In view of this, there have been intense efforts to meet the complex technological requirements that the massive-\ac{MIMO} paradigm entails.
At the hardware level, this requires scaling up existing multiple-antenna transceivers through the use of inexpensive service antennas and/or \ac{TDD} \cite{HtBD13,RPL+13,LETM14}.
At the same time however, given the vast amount of resources involved in upgrading an ageing infrastructure, a brute-force approach based solely on the evolution of wireless hardware technology cannot suffice.
Instead, unleashing the full potential of massive-\ac{MIMO} arrays requires a principled approach with the aim of minimizing computational overhead and related expenditures as the network scales up to accommodate more and more users.

In this general multi-user \ac{MIMO} context, it is crucial to optimize the input signal covariance matrix of each user, especially in the moderate (or low) \ac{SINR} regime \cite{CV93,YRBC04,SPB08i-sp,SPB08ii-sp,BLD11}.
The conventional approach to this problem involves the use of \ac{WF} solution methods, either \emph{iterative} (\acs{IWF}) \cite{YRBC04,SPB08-it} or \emph{simultaneous} (\acs{SWF}) \cite{SPB06}.
\acused{IWF}\acused{SWF}
In the \ac{IWF} algorithm only one transmitter updates its input covariance matrix per iteration (selected in a round-robin fashion);
instead, in \ac{SWF} all transmitters update their transmission characteristics simultaneously.
Owing to this ``parallelizability'', \ac{SWF} can be deployed in a distributed and decentralized fashion;
on the other hand, because of potential clashes in the users' concurrent updates, the \ac{SWF} algorithm may fail to converge \cite{SPB06}.
By comparison, \ac{IWF} \emph{always} converges to an optimal state \cite{YRBC04}, but this comes at the cost of centralization (to orchestrate the updating transmitters at each iteration) and a greatly reduced convergence speed (which is inversely proportional to the number of users in the system).%
\footnote{\beginrev
As suggested by one of the referees, it is worth pointing out here that \acl{WF} has also been applied to a broad range of distributed network paradigms;
see \eg \cite{HF13} for an application to cognitive radio OFDM networks.}

\revise{In addition} to the above, the authors of \cite{CADC08} proposed the so-called \acdef{IWMMSE} algorithm to solve the (non-convex) throughput maximization problem in the broadcast channel (downlink).
This work was subsequently extended in \cite{SRLH11} to broadcasting in multi-cell interference channels.
This formulation includes as a special case the uplink \ac{MAC} under the assumption that
\begin{enumerate*}
[(\itshape a\upshape)]
\item
all receivers are co-located and act as a single entity;
and
\item
this amalgamated entity employs \ac{SIC} to decode incoming messages.
\end{enumerate*}
In this context, \ac{IWMMSE} was shown to converge to an optimal solution in a distributed fashion, without suffering the convergence/distributedness trade-off of \acl{WF} methods.


\begin{table*}[t]
\centering
\small
\renewcommand{\arraystretch}{1.1}

\scshape
\begin{tabular}{llcccc}
\textbf{\scshape Algorithm \qquad[source]}
	&\textbf{Feedback}
	&\textbf{Convergence}
	&\textbf{Conv. Speed}
	&\textbf{Distributed}
	&\textbf{Overhead}
	\\
\hline
\hline
\acs{IWF}
	\hfill
	\cite{YRBC04}
	&full matrix
	&\checkmark
	&$\bigoh(\kk\log(1/\eps))$
	&no
	&$\bigoh(\min\{\mm^2, \nn^2\})$ 
	\\
\hline
\acs{SWF}
	\hfill
	\cite{SPB06}
	&full matrix
	&no
	&\textemdash
	&\checkmark
	&$\bigoh(\min\{\kk\mm^2, \nn^2\})$
	\\
\hline
\acs{IWMMSE}
	\hfill
	\cite{CADC08,SRLH11}
	&full matrix
	&\checkmark
	&\textemdash
	&\checkmark
	&$\bigoh(\min\{ \kk\mm^2, \nn^2\})$
	\\
\hline
\acs{MXL}
	\hfill
	\cite{MM16}
	&full matrix (imp.)
	&\checkmark
	&$\bigoh(1/\eps^{2})$
	&\checkmark
 	&$\bigoh(\min\{ \kk\mm^2, \nn^2\})$
	\\
\hline
\hline
\acs{MXL0}
	\hfill
	[this paper]
	&scalar
	&\checkmark
	&$\bigoh(1/\eps^{4})$
	&\checkmark
	&$\bigoh(1)$
	\\
\hline
\acs{MXL+}
	\hfill
	[this paper]
	&scalar
	&\checkmark
	&$\bigoh(1/\eps^{2})$
	&\checkmark
	&$\bigoh(1)$
	\\
\hline
\end{tabular}
\smallskip
\caption{%
Overview of related work.
 For the purposes of this table ``full matrix feedback'' refers to the case where the network's users have perfect knowledge of
\emph{a})
their effective channel matrices;
and/or
\emph{b})
the aggregate signal-plus-noise covariance matrix at the receiver at each transmission frame.
The characterization ``imp.'' (for ``imperfect'') signifies that noisy measurements suffice;
on the contrary, ``scalar'' means that users only observe their \emph{realized} utility (in our case, their achieved throughput).
The ``convergence'' and ``conv. speed'' columns indicate the best theoretical guarantees for each algorithm: $f(\eps)$ denotes the maximum number of iterations required to reach an $\eps$-optimal state while ``\textemdash'' means that no guarantees are known.
Finally, the ``overhead'' column indicates the computation/communication overhead of each iteration;
here and throughout,
$\kk$ is the number of users,
$\mm$ is the maximum number of transmit antennas per user,
and
$\nn$ is the number of antennas at the receiver.}
\vspace{-2ex}
\label{tab:related}
\end{table*}


Importantly, the above schemes rely on each user having perfect knowledge of
\begin{enumerate*}
[(\itshape a\upshape)]
\item
their effective channel matrix (which typically changes from one transmission frame to another);
and/or
\item
the global, system-wide signal-plus-noise covariance matrix at the receiver.
\end{enumerate*}
These elements are highly susceptible to observation noise, asynchronicities, and other impediments that arise in the presence of uncertainty;
as a result, algorithms requiring feedback of this type cannot be reliably implemented in real-world \ac{MIMO} systems.

To relax this ``perfect matrix feedback'' requirement, \cite{MM16} introduced a stochastic, first-order semidefinite optimization method based on \acdef{MXL}. 
The \ac{MXL} algorithm proceeds incrementally by combining stochastic gradient steps with a matrix exponential mapping that ensures feasibility of the users' signal covariance variables.
In doing so, \ac{MXL} guarantees fast convergence in cases where \ac{WF} methods demonstrably fail:
specifically,  \ac{MXL} achieves an $\eps$-optimal state within $\bigoh(1/\eps^{2})$ iterations, even in the presence of noise and uncertainty, in which case \acl{WF} methods are known to produce suboptimal results \cite{SPB06,SPB08-it,SPB09-sp}.

On the negative side, \ac{MXL} still requires
\begin{enumerate*}
[\upshape(\itshape a\upshape)]
\item
inverting a large matrix at the receiver;
and
\item
transmitting the resulting (dense) matrix to all connected users.
\end{enumerate*}
In a \ac{MIMO} array with $\nn=128$ receive antennas, this means $65$ kB of data per transmission frame, thus exceeding typical frame size limitations by a factor of $50\times$ to $500\times$ (depending on the specific standard) \cite{LBON16}.
Coupled with the significant energy expenditures involved in matrix computations and the fact that entry-level antenna arrays may be ill-equipped for this purpose, the overhead of \ac{MXL} quickly becomes prohibitive as \ac{MIMO} systems ``go large''.

\para{Contributions and related work}

Our main objective in this paper is to lift the requirement that users have access to full matrix feedback at each transmission frame (\eg perfect knowledge of their effective channel matrices or the system-wide signal-plus-noise covariance matrix).
Our main tool to lift these feedback requirements is the introduction of a ``zeroth-order'' optimization framework in which gradients are estimated from observed throughput values using a technique known as \acdef{SPSA} \cite{Spa97,FKM05}.
By integrating this \ac{SPSA} technique in the chassis of the \ac{MXL} method, we obtain a novel algorithm, which we call \acdef{MXL0}, and which we show converges to $\eps$-optimality within $\bigoh(1/\eps^{4})$ iterations 
(on average and with high probability).

On the positive side, this analysis shows that \ac{MXL0} is an asymptotically optimal algorithm (similarly to \ac{MXL}, \ac{IWF} and \ac{IWMMSE}) but \emph{without} the full matrix feedback requirements of these methods.
On the negative side, despite the vastly reduced feedback and overhead requirements of \ac{MXL0}, the drop in convergence speed relative to the original \ac{MXL} scheme is substantial and makes the algorithm ill-suited for practical systems.
In fact, as we show via numerical experiments in realistic network conditions, \ac{MXL0} might take up to $10^{5}$ iterations to achieve a relative optimality threshold of $\eps = 10^{-1}$ (compared to between $10$ and $100$ iterations for \ac{MXL}).
This is  caused by the very high variance of the \ac{SPSA} estimator, which incurs a significant amount of state space exploration and leads to a dramatic drop in the algorithm's convergence speed.

To circumvent this obstacle, we introduce a variance reduction mechanism where information from previous transmit cycles is reused to improve the accuracy of the \ac{SPSA} gradient estimator.
We call the resulting algorithm \acdef{MXL+}, and we show that \emph{it combines the best of both worlds:}
it retains the fast $\bigoh(1/\eps^{2})$ convergence rate of the standard \ac{MXL} algorithm, despite the fact that it only requires a \emph{single scalar} worth of feedback per iteration.
In fact, in many instances, the reuse of past queries is so efficient that the gradient-free \ac{MXL+} algorithm ends up outperforming even \ac{MXL} (which requires first-order gradient feedback).

With regard to feedback reduction, the work which is closest in spirit to our own is the very recent paper \cite{LA19}, where the authors seek to minimize the informational exchange of \ac{MXL} methods applied to the maximization of transmit energy efficiency (as opposed to throughput).
There, instead of requiring an $\nn\times\nn$ Hermitian matrix as feedback, each transmitter is assumed to receive a random selection of gradient components.
This (batch) ``coordinate descent'' approach leads to a trade-off between signalling overhead and speed of convergence, but still relies on users having access to first-order gradient information.
In contrast, we do not make any such assumptions and work \emph{solely} with throughput observations;
in this way, the communication overhead is reduced to a single scalar, while retaining the possibility of asynchronous, distributed updates.

\beginrev
Finally, from a beamforming perspective, the algebraic power method can also be used to iteratively approximate optimal beamformer/combiner pairs without prior knowledge of the channel matrix.
However, this approach requires a stationary wireless background:
in the presence of multiple users, user-to-user interference can render the estimation of individual channel matrices impossible.
For this reason, we do not consider such methods in the sequel;
for an overview, see \cite{DCG04,OLR17}.
\endedit

%

\para{Notation}
Throughout the sequel, we use bold symbols for matrices, saving the letters $\nbk,\altnbk$ for user assignments and $\itr,\altitr$ for time indices, so that \eg matrix~$\matk{\nbk}$ relates to user $\nbk$, $\matn{\itr}$ to time $\itr$, and $\matkn{\nbk}{\itr}$ to user $\nbk$ at time $\itr$.
The symbols~$\zero{\cdot}$, $\magnitude{\cdot}$, and~$ \bigtheta{\cdot}$ are taken as in the common Bachmann-Landau notation.

\section{Problem Statement}
\label{section:problemdescription}

In this section, we present two archetypal multi-user \ac{MIMO} system models that are at the core of our considerations:
a centralized sum-rate optimization problem, and an individual rate maximization game.
In both cases, the optimization process is assumed to unfold in a distributed, online manner as follows:
\begin{enumerate}
\item
At each transmission frame, every user in the network selects an \emph{action} (an input signal covariance matrix).
\item
This choice generates each user's \emph{utility} (their sum- or individual rate, depending on the problem's specifics).
\item
Based on the observed utilities, the users update their actions and the process repeats.
\end{enumerate}

We stress here that
\emph{we do not assume}
the existence of a centralized control hub with access to all the primitives defining the problem (individual channel matrices, input signal covariance matrices, etc.)
and/or
the capability of implementing an \emph{offline} optimization algorithm to solve it.
Instead, we focus on wireless networks with light-weight deployment and implementation characteristics, such as multi-user \ac{MIMO} uplink networks in typical urban environments.
In the downlink, the decision process regarding all transmission aspects (including the input signal covariance matrices) is inherently centralized as it takes places at the unique transmitter, which makes the broadcast setting a more resource-hungry choice compared to the uplink;
nevertheless, the duality between the \ac{MAC} and the \ac{BC} \cite{JVG04} can be exploited to solve the analogous centralized problem in the downlink.

In terms of decoding, we consider two different schemes at the receiver:
\begin{enumerate*}
[(\itshape a\upshape)]
\item
\acdef{SIC}, which is suitable for networks with centralized user admission and control protocols;
and
\item
\acdef{SUD}, which is suitable for more decentralized, ad hoc networks.
\end{enumerate*}

\subsection{Centralized sum-rate maximization}
\label{subsec:SIC}

Consider a Gaussian vector \acl{MAC} consisting of $\kk$ users simultaneously transmitting to a wireless receiver equipped with $\nn$ antennas.
If the $\nbk$-th transmitter is equipped with $\mmk{\nbk}$ antennas, we get the baseband signal model
\begin{equation}
\signal = \sum\nolimits_{\nbk=1}^{\kk} \chmatk{\nbk} \msgk{\nbk} + \noise, 
\end{equation}
where:
\begin{enumerate*}
[(\itshape a\upshape)]
\item
$\msgk{\nbk}\in\Complex^{\mmk{\nbk}}$ denotes the signal transmitted by the $\nbk$-th user;
\item
$\chmatk{\nbk}\in \Complex^{\nn\times\mmk{\nbk}}$ is the corresponding channel matrix;
\item
$\signal\in\Complex^{\nn}$ is the aggregate signal reaching the receiver;
and
\item
$\noise\in\Complex^{\nn}$ denotes the ambient noise in the channel, including thermal and environmental interference effects (and modeled for simplicity as a zero-mean, circulant Gaussian vector with identity covariance).
\end{enumerate*}
In this general model, the transmit power of the $\nbk$-th user is given by $p_{\nbk} = \expectation{\withdagger{\msgk{\nbk}}\msgk{\nbk}}$.
Then, letting $\powerk{\nbk}$ denote the maximum transmit power of user $\nbk$, we also write
\begin{equation}
\matk{\nbk}
	= \expectation{\msgk{\nbk}\withdagger{\msgk{\nbk}}} \big/ \powerk{\nbk}
\end{equation}
for the normalized signal (or input) covariance matrix of user $\nbk$.
By definition, $\matk{\nbk}$ is Hermitian and positive-semidefinite, which we denote by writing $\matk{\nbk}\in\Herm{\mmk{\nbk}}$ and $\matk{\nbk} \mgeq 0$ respectively.

Assuming \acf{SIC} at the receiver, the users'achievable sum rate is given by the familiar expression
\begin{equation}
\label{rate}
\rate{\mat}
	= \log\det \contrzmat,
\end{equation}
where
\begin{equation}
\label{eq:MUI-total}
\contrzmat
	\equiv
	\contrzmat(\mat)
	= \identity + \sum\nolimits_{\nbk=1}^{\kk} \normalize{}{\powerk{\nbk} \,} \channelk{\nbk} \matk{\nbk} \withdagger{\channelk{\nbk}}
\end{equation}
is the aggregate signal-plus-noise covariance matrix at the receiver,
and
$\mat \equiv \vect{\matk{1},\dots,\matk{\kk}}$ denotes the users' aggregate signal covariance profile \cite{Tel99}.
\ac{SIC} decoding of this type has been exploited as a means to control the multi-user interference in the power-domain \ac{NOMA} technology \cite{SCBL19}, which provides a better spectrum utilization and spectral efficiency compared with traditional orthogonal schemes.

Since $\rate{\mat}$ is increasing in each user's total transmit power $p_{\nbk} = \power_{\nbk} \trace{\matk{\nbk}}$, the channel's throughput is maximized when the users individually saturate their power constraints, \ie when $\trace{\matk{\nbk}} = 1$ for all $\nbk=1,\dotsc,\kk$.
In this way, we obtain the \emph{power-constrained sum-rate optimization problem}
\begin{equation}
\label{eq:trproblem}
\tag{Opt}
\begin{aligned}
\text{maximize}
	&\quad
	\rate{\mat}
		\equiv \rate{\matk{1},\dotsc,\matk{\kk}}
	\\
\text{subject to}
	&\quad
	\matk{\nbk}\in\compactsetk{\nbk}
	\;
	\text{for all $\nbk = 1,\dotsc,\kk$},
\end{aligned}
\end{equation}
where each user's feasible power region $\compactsetk{\nbk}$ is given by
\begin{equation}
\label{eq:feasible}
\compactsetk{\nbk}
	= \left\{\matk{\nbk}\in\Herm{\mmk{\nbk}} \setst \trace{\matk{\nbk}} = 1, \matk{\nbk}\mgeq 0 \right\}.
\end{equation}
By definition, each $\compactsetk{\nbk}$ is a spectrahedron of (real) dimension $\dimensionk{\nbk} = \mmk{\nbk}^{2}-1$, so the problem's dimensionality is $\sum_{\nbk} \dimensionk{\nbk} = \magnitude{\sum_{\nbk}\mmk{\nbk}^{2}}$.
To avoid trivialities, we will assume in what follows that each transmitter possesses at least two antennas, so $\dimensionk{\nbk} > 0$ for all $\nbk = 1,\dotsc,\kk$.
Also, to further streamline our discussion, we will state our results in terms of the maximum number $\maxmm=\max_{\nbk} \mmk{\nbk}$ of antennas per transmitter\textemdash or, equivalently, in terms of the larger dimension $\maxdimension=\maxmm^{2}-1$.%
\footnote{The statement of our results can be fine-tuned at the cost of introducing further notation for other aggregate statistics of the number of antennas per transmitter (such as the arithmetic or geometric mean of $\mmk{\nbk}$).
The resulting expressions are fairly cumbersome, so we do not report them here.}

\subsection{Distributed individual rate maximization}
\label{subsec:SUD}

Moving beyond the sum-rate maximization problem above, if messages are decoded using \acl{SUD} at the receiver (\ie interference by \revise{all other} users is treated as additive colored noise), each user's \emph{individual} rate will be
\begin{equation}
\label{ratek}
\ratek{\nbk}{\matk{\nbk};\matk{-\nbk}} = 
	\rate{\matk{1},\dotsc,\matk{\kk}}
		- \rate{\matk{1},\dotsc,0,\dotsc,\matk{\kk}},
\end{equation}
where $(\matk{\nbk};\matk{-\nbk})$ is shorthand for the covariance profile $(\matk{1},\dotsc,\matk{\nbk},\dotsc,\matk{\kk})$.
In turn, this leads to the \emph{individual rate maximization} game
\begin{equation}
\label{eq:trproblemk}
\tag{Opt$_{\nbk}$}
\begin{aligned}
\text{maximize}
	&\quad
	\ratek{\nbk}{\matk{\nbk};\matk{-\nbk}}
	\\
\text{subject to}
	&\quad
	\matk{\nbk}\in\compactsetk{\nbk}
\end{aligned}
\end{equation}
to be solved unilaterally by each user $\nbk = 1,\dotsc,\kk$.

Given that $\rate{\mat}$ is concave in $\mat$ and $\ratek{\nbk}{\matk{\nbk};\matk{-\nbk}}$ is concave in $\matk{\nbk}$, it follows that the decentralized problem \eqref{eq:trproblemk} defines a concave potential game whose Nash equilibria coincide with the solutions of \eqref{eq:trproblem} \cite{MS96,Ney97,MM16}.
In view of this, the gradient-free optimization framework and algorithms derived in this paper and designed to solve the centralized sum-rate optimization \eqref{eq:trproblem} will also solve the game \eqref{eq:trproblemk};
conversely, \eqref{eq:trproblem} is amenable to a distributed approach where it is treated as the aggregation of the unilateral sub-problems \eqref{eq:trproblemk}, to be solved in parallel by the network's users.
We revisit this distributed approach in \cref{section:desynchronized}.

\subsection{\Acl{WF} and \acl{MXL}}
\label{subsec:MXL}

A basic online solution method for \eqref{eq:trproblem} is the \acf{WF} algorithm \cite{CV93,YRBC04,SPB09-sp} and its variants\textemdash iterative or simultaneous \cite{LP06,SPB06,SPB08-it}.
In \ac{WF} schemes, transmitters are tacitly assumed to have full knowledge of their channel matrices $\chmatk{\nbk}$ as well as the \ac{MUI} covariance matrix
\begin{equation}
\label{eq:MUI}
\contrzmatk{\nbk}
	= \identity
		+ \sum\nolimits_{\altnbk\neq\nbk} \powerk{\altnbk} \chmatk{\altnbk} \matk{\altnbk} \withdagger{\chmatk{\altnbk}}.
\end{equation}
These matrices are then used to ``water-fill'' the users' \emph{effective channel matrices}
\begin{equation}
\label{eq:effchan}
\tilde\chmat_{\nbk}
	= \contrzmat_{\nbk}^{-1/2} \chmatk{\nbk}
\end{equation}
either iteratively (\ie in a round-robin fashion),
or simultaneously (all transmitters at the same time);
the corresponding implementations are called \acf{IWF} and \acf{SWF} respectively.

We stress here that the users' effective channel matrices may change over time, even when the \emph{actual} channel matrix $\chmat_{\nbk}$ is static:
this is because $\tilde\chmat_{\nbk}$ depends on the transmission characteristics of \emph{all} other users in the network (via the \ac{MUI} matrix $\contrzmat_{\nbk}$), and these typically evolve over time according to each user's optimization policy.

In this context, \ac{IWF} converges \emph{always} (but slowly if the number of users is large), whereas \ac{SWF} \emph{may fail} to converge altogether \cite{SPB06,MBML12}.
In addition, as we discussed in the introduction, \acl{WF} is highly susceptible to observation noise, asynchronicities, and other impediments that arise in real-world systems, so the solution of \eqref{eq:trproblem} in the presence of uncertainty requires a different approach (see also the numerical experiments presented in \cref{section:experiments}).

These limitations are overcome by
the \acdef{MXL} algorithm \cite{MBM12,MM16}, which will serve both as a reference and an entry point for our analysis.
Heuristically, \ac{MXL} proceeds by aggregating incremental gradient steps (possibly evaluated with imperfect channel state and \ac{MUI} estimations), and then using a suitable matrix exponential mapping to convert these steps into a positive-semidefinite matrix that meets the transmit power constraints of \eqref{eq:trproblem} and/or \eqref{eq:trproblemk}.

More formally, let
\begin{equation}
\label{gradkrate}
\gradk{\nbk} \rate{\mat} 
=	\normalize{}{\powerk\nbk}{\withdagger{\channelk{\nbk}}}
	\left[
		\identity
			+ \sum\nolimits_{\altnbk=1}^{\kk}
			\normalize{}{\powerk{\altnbk} \,} \channelk{\altnbk} \matk{\altnbk} \withdagger{\channelk{\altnbk}}
	\right]^{-1}
		\channelk{\nbk}.
\end{equation}
denote the individual gradient of $\ratefunction$ (or $\ratefunction_{\nbk}$) relative to the signal covariance matrix of the $\nbk$-th user,
and let
\begin{equation}
 \Dualmatk{\nbk}=\{\dualmatk{\nbk}\in\Herm{\mmk{\nbk}}: \trace{\dualmatk{\nbk}}= 0 \}
\end{equation}
denote the subspace tangent to~$\compactset$.
Then, given an initialization $\dualmatn{1} \in \Dualmat \equiv \prod_{\nbk}\Dualmatk{\nbk}$,  the \ac{MXL} algorithm is defined via the basic recursion
\begin{equation}
\label{MXL}
\tag{MXL}
\begin{aligned}
\matn{\itr}
	&= \aggregatematsol{\dualmatn{\itr}},
	\\
\dualmatn{\itr+1}
	&= \dualmatn{\itr} + \stepsizen{\itr} \estgradn{\itr},
\end{aligned}
\end{equation}
where:
\begin{enumerate}
[\itshape i\hspace*{.5pt}\upshape),ref=\alph*]

\item
$\matn{\itr}$ denotes the users' input signal covariance profile at the $\itr$-th iteration of the algorithm ($\itr=1,2,\dotsc$).

\item
$\estgradn{\itr}=\vect{\estgradkn{1}{\itr},\dots,\estgradkn{\kk}{\itr}}$ is an estimate of the tangent component of the gradient $\grad\ratefunction$ relative to~$\compactset$.%
\footnote{More precisely, \eqref{MXL} only requires estimates of~$\projgrad\ratefunction:\compactset \mapsto\Dualmat$, which here denotes  the tangent component of the gradient~$\grad\ratefunction$ relative to~$\compactset$, given  by 
$\projgrad \ratefunction = \vect{\projgradk{1} \ratefunction,\dots,\projgradk{\kk} \ratefunction}$ where $\projgradk{\nbk} \ratefunction = \gradk{\nbk} \ratefunction - \trace{\gradk{\nbk} \ratefunction}\, \identity$.  
All technical details in regards to~\eqref{MXL} are deferred to the appendix.}

\item
$\stepsizen{\itr} > 0$ is a non-increasing sequence of step-sizes whose role is examined in detail below.

\item
$\dualmatn{\itr}$ is an auxiliary matrix that aggregates gradient steps.

\item
$\matsol{\dualmat} = (\matsolk{1}{\dualmatk{1}},\dotsc,\matsolk{\kk}{\dualmatk{\kk}})$ denotes the matrix exponential mapping given in (block) components by
\begin{equation}
\label{trmatsolk}
\matsolk{\nbk}{\dualmatk{\nbk}}
	= \frac{\exponential{\dualmatk{\nbk}}}{\trace{\exponential{\dualmatk{\nbk}}}}.
\end{equation}
\end{enumerate}

The intuition behind \eqref{MXL} is that the exponential mapping assigns more power to the spatial directions that are aligned to the objective's gradient (as estimated via $\estgradn{\itr}$).
In fact, the \ac{MXL} algorithm can be explained as a matrix-valued instance of Nesterov's dual averaging method \cite{Nes09};
the key innovation of \ac{MXL} is the matrix exponentiation step which lifts the need to do a costly projection on the users' feasible region (a trace-constrained spectrahedron).
The output of each iteration of the algorithm is a positive-semidefinite matrix with unit trace, so the problem's constraints are automatically satisfied.
We defer the details of this derivation to \cref{appendix:DA}.

As was shown in \cite{MM16}, the \ac{MXL} algorithm achieves an $\eps$-optimal signal covariance profile within $\magnitude{1/\eps^{2}}$ iterations.
However, to do so, the algorithm still requires access to noisy observations of the gradient matrices \eqref{gradkrate}.
Typically, this involves inverting a (dense) $\nn\times\nn$ Hermitian matrix at a central hub and subsequently transmitting the result to the network's users, so the algorithm's computation and communication overhead is considerable (see Table~\ref{tab:related}).
On that account, our main focus in the sequel will be to lift the assumption that the network's users have access to the gradient matrices \eqref{gradkrate}, all the while maintaining the $\magnitude{1/\eps^{2}}$ convergence speed of \eqref{MXL}.

\subsection{Technical preliminaries and notation}

For the analysis to come, it will be convenient to introduce the following constants.
First, we will write $\compactset = \prod_{\nbk} \compactsetk{\nbk}$ for the feasible region of \eqref{eq:trproblem}, and we will denote by $\Lipschitz$ the Lipschitz constant of $\ratefunction$ over $\compactset$ relative to the nuclear norm;
specifically, this means that:
\begin{equation}
\label{lipschitz}
\modulus{\rate{\mat}-\rate{\altmat}}
	\leq \Lipschitz \norm{\mat-\altmat}
	\quad
	\text{for all $\mat,\altmat\in\compactset$}.
\end{equation} 
Moreover, we will also write $\lipschitzkl{\nbk}{\altnbk}$ for the user-specific Lipschitz constants of $\gradk{\nbk}\ratefunction$, understood in the following sense:
\begin{equation}
\label{lipschitzkl}
\dualnorm{
	\gradk{\nbk}\rate{\matk{\altnbk};\matk{-\altnbk}}
		- \gradk{\nbk}\rate{\altmatk{\altnbk};\matk{-\altnbk}}} 
	\leq \lipschitzkl{\nbk}{\altnbk}
		\frobeniusnorm{\matk{\altnbk}-\altmatk{\altnbk}},
\end{equation}
for all $\matk{\altnbk},\altmatk{\altnbk}\in\compactsetk{\altnbk}$,
$\matk{-\altnbk} \in \compactsetk{-\altnbk} \equiv \prod_{\altaltnbk\neq\altnbk} \compactsetk{\altaltnbk}$,
and all $\nbk,\altnbk=1,\dots,\kk$.
We also let $\meanlipschitzk{\nbk}= (1/\kk) \sum_{\altnbk=1}^{\kk} \lipschitzkl{\nbk}{\altnbk} $ denote the ``averaged'' Lipschitz constant of user $\nbk$, and we write $\meanlipschitz = (1/\kk) \sum_{\nbk=1}^{\kk} \meanlipschitzk{\nbk}$ for the overall ``mean'' Lipschitz constant.
For a detailed discussion of the nuclear norm $\norm\cdot$ and its dual $\dualnorm\cdot$, we refer the reader to \cref{appendix:DA}.

\section{\ac{MXL} without Gradient Information}
\label{section:MXL0}

\newcommand{\showMXLzerolastiterate}[2]{#1}

As we noted above, the existing implementations of \ac{MXL} invariably rely on the availability of gradient feedback\textemdash full \cite{MBM12}, noisy \cite{MM16}, or partial \cite{LA19}.
Our aim in this section is to show that this requirement can be obviated by means of a (possibly biased) gradient estimator, which only requires observations of a \emph{single} scalar\textemdash the users' achieved throughput. Our approach builds on the method of \acdef{SPSA}, a gradient estimation procedure which has been studied extensively in the context of large-scale, derivative-free optimization \cite{Spa97,FKM05}, and which we discuss in detail below.

\subsection{Gradient estimation: intuition and formal construction}
\label{section:intuition}

We start by providing some intution behind the \ac{SPSA} method.
For this, consider the scalar case and a simple differentiable function $\genericfunction:\Real\mapsto\Real$.
Then, by definition, the derivative of $\genericfunction$ at any point $x$ satisfies 
\begin{equation} 
\dgeneric{x}
	= \frac{\generic{x{\,+\,}\delta}{\,-\,}\generic{x{\,-\,}\delta}}{2\delta}
		+ \zero{\delta}.
\end{equation}
Therefore, if $\delta>0$ is small enough, an estimate for $\dgeneric{x}$ can be obtained from two queries of the value of $\genericfunction$ at the neighboring points $x-\delta$ and $x+\delta$ as follows:
\begin{equation}
\label{basicvtp}
\basicvtp{x}{\delta}
	= \frac{\generic{x+\delta}-\generic{x-\delta}}{2\delta}.
\end{equation}
Thus, if $\dgenericfunction$ is $\gradlipschitz$-Lipschitz continuous on the search domain, it is easy to see that the error of the estimator $\basicvtp{x}{\delta}$ is uniformly bounded as
$
\modulus{\basicvtp{x}{\delta}-\dgeneric{x}}
\leq \gradlipschitz\, \delta/2
$,
\ie the estimator \eqref{basicvtp} is accurate up to $\magnitude{\delta}$.

Taking this idea further, it is possible to estimate $\dgeneric{x}$ using only a \emph{single} function query at either of the test points $x-\delta$, or $x+\delta$, chosen uniformly at random.
To carry this out, let $z$ be a random variable taking the value $-1$ or $+1$ with equal probability $1/2$, and define the \emph{one-shot} \ac{SPSA} estimator
\begin{equation}
\label{basicvop}
\basicvop{x}{\delta}
	= \frac{\generic{x + \delta z}}{\delta} z.
\end{equation}
Then, a straightforward calculation gives $\expectation{\basicvop{x}{\delta}}=\basicvtp{x}{\delta}$, \ie $\basicvopfunction$ is a stochastic estimator of $\dgenericfunction$ with accuracy
\begin{equation}
\modulus{\expectation{\basicvop{x}{\delta}-\dgeneric{x}}}
	= \modulus{\basicvtp{x}{\delta}-\dgeneric{x}}
	\leq \gradlipschitz\,\delta / 2
	= \magnitude{\delta}.
\end{equation}

The \ac{SPSA} approach described above can be applied to our \ac{MIMO} setting as follows.
First, each user $\nbk$ draws, randomly and independently, a matrix $\trzmatk{\nbk}$ from the unit sphere%
\footnote{Note that the dimension of $\trmspherei{\dimensionk{\nbk}-1}$ as a manifold is $\dimensionk{\nbk}-1$, \ie one lower than that of the feasible region $\compactsetk{\nbk}$;
this is due to the unit norm constraint $\frobeniusnorm{\trzmatk{\nbk}}= 1$.}
\begin{equation} 
\label{trsphere}
\trmspherei{\dimensionk{\nbk}-1}
	= \{\trzmatk{\nbk}\in\trZmatk{\nbk} \setst \frobeniusnorm{\trzmatk{\nbk}}= 1 \}.
\end{equation} 
Then, translating \eqref{basicvop} to the distributed, Hermitian setting of \cref{section:problemdescription} yields, for all $\nbk=1,\dots,\kk$, the gradient estimator
\begin{equation}
\label{estgradratekfirst}
\basicestgradratek{\nbk}{\mat}{\trzmat}
	= \frac{\dimensionk{\nbk}}{\delta} \rate{ \mat +\delta \trzmat} \, \trzmatk{\nbk},
\end{equation}
where $\trzmat=\vect{\trzmatk{1},\dots,\trzmatk{\kk}}$ collects the random shifts of all users.

\begin{remark}
The factor $\dimensionk{\nbk}=\mmk{\nbk}^{2}-1$ in \eqref{estgradratekfirst} has a geometric interpretation as the ratio between the volumes of the sphere $\trmspherei{\dimensionk{\nbk}-1}$ (where $\trzmatk{\nbk}$ is drawn from) and 
the containing $\dimensionk{\nbk}$-dimensional ball $
\trmballi{\dimensionk{\nbk}} = \{\trzmatk{\nbk}\in\trZmatk{\nbk} \setst \frobeniusnorm{\trzmatk{\nbk}}\leq 1\}
$.
Its presence is due to Stokes' theorem, as detailed in \cref{lemma:bias}.
\end{remark}

A further complication that arises in our constrained setting is that the query point $\mat + \delta \trzmat$ in \eqref{estgradratekfirst} may lie outside the feasible set $\compactset$ if $\mat$ is too close to the boundary of $\compactset$.
To avoid such an occurrence, we introduce below a ``safety net'' mechanism which systematically carries back the pivot points $\matk{\nbk}$ towards the ``prox-center'' $\cmatk{\nbk} = \identityd{\mmk{\nbk}}/\mmk{\nbk}$
of $\compactsetk{\nbk}$ before applying the random shift $\trzmatk{\nbk}$.
Specifically, taking $\radiusk{\nbk}>0$ sufficiently small so that the Frobenius ball centered at $\cmatk{\nbk}$ lies entirely in $\compactsetk{\nbk}$,
we consider the homothetic adjustment
\begin{equation}
\label{linearperturbator}
\testmatk{\nbk}
	= \matk{\nbk}
		+ \frac{\delta}{\radiusk{\nbk}} (\cmatk{\nbk} - \matk{\nbk}) + \delta \trzmatk{\nbk}.
\end{equation}
By an elementary geometric argument, it suffices to take \begin{equation}
\label{radius}
\radiusk{\nbk}
	= 1/\sqrt{\mmk{\nbk}(\mmk{\nbk}-1)}.
\end{equation}

With this choice of $\radiusk{\nbk}$, it is easy to show that, for $\delta < \radiusk{\nbk}$, the adjusted query point $\testmatk{\nbk}$ lies in $\compactsetk{\nbk}$ for all $\nbk=1,\dotsc,\kk$.
On that account, we redefine the \ac{SPSA} estimator for \eqref{eq:trproblem} as
\begin{equation}
\label{estgradratek}
\tag{SPSA}
\basicestgradratek{\nbk}{\mat}{\trzmat}
	= \frac{\dimensionk{\nbk}}{\delta} \rate{ \testmat } \, \trzmatk{\nbk},
\end{equation}
where, in obvious notation, we set $\testmat = \vect{\testmatk{1},\dots,\testmatk{\kk}}$.
The distinguishing feature of \eqref{estgradratek} is that it is well-posed:
\emph{any} query point $\testmat$ is feasible under \eqref{estgradratek}.
Thus, extending the one-dimensional analysis in the beginning of this section, \cref{lemma:bias} claims that the accuracy of the estimator \eqref{estgradratek} is uniformly bounded as $\dualnorm{\expectation{\basicestgradratek{\nbk}{\mat}{\trzmat}-\projgrad\rate\mat}}= \magnitude{\delta}$.
In the rest of this section, we exploit this property to derive and analyze a first \emph{gradient-free} variant of \eqref{MXL}.

\subsection{A \acl{MXL0} scheme}

To integrate the gradient estimator \eqref{estgradratek} in the chassis of \eqref{MXL}, we will use a (non-increasing) query radius sequence $\deltan{\itr}$ satisfying the basic feasibility condition:
\setcounter{hypothesis}{-1} 
\begin{equation}
\hyplabel{deltacondition} 
\deltan{\itr} 
	< \min\nolimits_{\nbk}\radiusk{\nbk}
	= 1/\sqrt{\maxmm(\maxmm-1)}
	\quad
	\text{for all $\itr\geq 1$}.
\end{equation}
Then, under \eqref{MXL}, the task of user $k$ at the $\itr$-th stage of the algorithm will be given by the following sequence of events:
\begin{enumerate}
[labelwidth=!,labelindent=5pt]
\item
Draw a random direction $\trzmatkn{\nbk}{\itr}\in\trmspherei{\dimensionk{\nbk}-1}$.
\item
Transmit with the covariance matrix $\testmatkn{\nbk}{\itr}$ given by \eqref{linearperturbator}.
\item
Get the achieved throughput $\ratevalue_{\itr} = \rate{\testmatn{\itr}}$.
\item
Construct the gradient estimate $\estgradkn{\nbk}{\itr}$ given by \eqref{estgradratek}.
\item
Update $\dualmatkn{\nbk}{\itr}$ and $\matkn{\nbk}{\itr}$ in accordance with \eqref{MXL}. 
\end{enumerate}
The resulting algorithm will be referred to as \acdef{MXL0};
for a pseudocode implementation, see \cref{alg:trMXL} above.
\medskip


\begin{algorithmenv}[t]
\caption{\Acf{MXL0}\label{alg:trMXL}}
\removelatexerror
\RestyleAlgo{tworuled}%
\centering
\vspace{-2.5pt}
\begin{algorithm}[H]
 \Parameters{ 
$\stepsizen{\itr}$, $\deltan{\itr}$}
 \KwInit{$\itr\leftarrow 1$, $\dualmat\leftarrow 0$;\\
 \hspace{8.5ex}
 $\forall\nbk\colon \matk{\nbk} \leftarrow (\powerk{\nbk}/\mmk{\nbk})\, \identityk{\nbk}$}
\mainnl \Repeat{}{
\mainnl \inlineAteachDIP{$\nbk\in\{1,\dots,\kk\}$}{$\DAk{\nbk}{\stepsizen{\itr},\deltan{\itr}}$ } \;
\mainnl $\itr \leftarrow \itr+1$ \;
 }
\smallskip
\rule{0.96\linewidth}{.4pt} \\
\setcounter{AlgoLine}{0}
\Fctn{$\DAk{\nbk}{\stepsize,\delta}$}{
\routinenl \Sampleuniformly{$\trzmatk{\nbk}$}{$\trmspherei{\dimensionk{\nbk}-1}$} \;
\routinenl \Play $ \testmatk{\nbk} \leftarrow \matk{\nbk} + \frac{\delta}{\radiusk{\nbk}} (\cmatk{\nbk} - \matk{\nbk}) + \delta \trzmatk{\nbk} $ \;
\routinenl \Get $ \ratevalue \leftarrow \rate{ \testmat}$ \;
\routinenl \Set $\estgradk{\nbk} \leftarrow \frac{\dimensionk{\nbk}}{\delta} \ratevalue \, \zmatk{\nbk} $ \;
\routinenl \Set $\dualmatk{\nbk} \leftarrow \dualmatk{\nbk} + \stepsize \estgradk{\nbk}$ \;
\routinenl \Set $\matk{\nbk} \leftarrow \normalize{\powerk{\nbk}\,}{} \matsolk{\nbk}{\dualmatk{\nbk}}$ \;
}
\end{algorithm}%
\end{algorithmenv}


Our first convergence result for \ac{MXL0} is as follows:

\newcommand{\reflasttrMXLi}{\textup{(\ref{lasttrMXL}\firstsymbol)}}%
\newcommand{\reflasttrMXLii}{\textup{(\ref{lasttrMXL}\secondsymbol)}}%
\newcommand{\reflasttrMXLiii}{\textup{(\ref{lasttrMXL}\thirdsymbol)}}%
\newcommand{\firstsymbol}{a}%
\newcommand{\secondsymbol}{b}%
\newcommand{\thirdsymbol}{c}%
\newcommand{\firstcond}{\textup{(\firstsymbol)}}%
\newcommand{\secondcond}{\textup{(\secondsymbol)}}%
\newcommand{\thirdcond}{\textup{(\thirdsymbol)}}%
\begin{theorem}[Convergence of \ac{MXL0}]
\label{thm:lasttrMXL}
Suppose that \ac{MXL0} \textpar{\cref{alg:trMXL}} is run with non-increasing step-size and query-radius policies satisfying \eqref{deltacondition} and 
\begin{equation} \label{lasttrMXL}
\textstyle
\firstcond \;
\sum_{\itr} \stepsizen{\itr}
	= \infty,
	\;\;
\;\secondcond \;
\sum_{\itr} \stepsizen{\itr} \deltan{\itr}
	< \infty,
	\;\;
\;\thirdcond \;
\sum_{\itr} \stepsizen{\itr}^{2}/\deltan{\itr}^{2} <\infty.
\end{equation}
Then, with probability $1$, the sequence of the users' transmit covariance matrices $\testmatn{\itr}$ converges to the solution set of \eqref{eq:trproblem}.
\end{theorem}%

\Cref{thm:lasttrMXL} provides a strong asymptotic convergence result, but it does not give any indication of the algorithm's convergence speed.
To fill this gap, our next result focuses on the algorithm's value convergence rate relative to the maximum achievable transmission rate $\optrate = \max\ratefunction$ of \eqref{eq:trproblem}.

\begin{theorem}[Convergence rate of \ac{MXL0}]
\label{thm:trMXLiinew}   
Suppose that \ac{MXL0} \textpar{\cref{alg:trMXL}} is run for $T$ iterations with constant step-size and query radius parameters of the form $\stepsizen{\itr} = \stepsize/T^{3/4}$ and $\deltan{\itr} = \delta/T^{1/4}$, $\delta < 1/\sqrt{\maxmm(\maxmm-1)}$.
Then,
the algorithm's ergodic average $\meanmatn{T} = (1/T) \sum_{\itr=1}^{T} \matn{\itr}$ enjoys the bounds:
\renewcommand{\dimensionk}[1]{\maxdimension}
\renewcommand{\meansqdimension}{\maxdimension}
\renewcommand{\meandimension}{\maxdimension}
\begin{enumerate}[(a),ref=\alph*,wide,labelwidth=!,labelindent=5pt]
\addtolength{\itemsep}{\smallskipamount}

\item
\label{thm:trMXLiifirst}
In expectation, 
\begin{equation}
\label{trsimplerateboundmeanexpectation}
\expectation{\optrate - \rate{ \meanmatn{T}}}
	\leq \frac{\Acoefficient{\stepsize}{\delta}}{T^{1/4}}
	= \Magnitude{T^{-1/4}},
\end{equation}
where 
\(
\Acoefficient{\stepsize}{\delta} 
	=  (\kk/\stepsize) \log\maxmm
		+ 4 \kk^{2} \meanlipschitz \delta
		+ 2^{1-2\kk} (\optrate\meansqdimension)^{2} \kk \stepsize / \strongly \delta^{2}.
\)
\item
\label{thm:trMXLiisecond} 
In probability, for any small enough tolerance $\zmargin>0$,
\begin{equation}
\label{trsimplerateboundmeaninterval} 
\Proba{\optrate -\rate{\meanmatn{T}} \geq  \frac{\Acoefficient{\stepsize}{\delta}}{T^{1/4}} + \zmargin}
	\leq \exp{\left(-\frac{2^{2\kk-5} \delta^{2} \zmargin^{2} T^{1/2}}{(\optrate\kk \meandimension)^{2}} \right)}.
\end{equation}
\end{enumerate}
\end{theorem}

In words, \cref{thm:trMXLiinew} shows that \cref{alg:trMXL} converges at a rate of $\magnitude{T^{-1/4}}$ on average, and the probability of deviating by more than $\zmargin$ from this rate is exponentially small in $\zmargin$ and $T$.
Compared to \eqref{MXL}, this indicates an increase in the number of iterations required to achieve $\eps$-optimality from $\magnitude{1/\eps^{2}}$ to $\magnitude{1/\eps^{4}}$.
As we illustrate in detail in \cref{section:experiments}, this performance drop is quite significant and makes \ac{MXL0} prohibitively slow in practice.
The rest of our paper is devoted precisely to bridging this vital performance gap.

\renewcommand{\strongly}{}
\newcommand{\sqrtstrongly}{}
\newcommand{\sqrtsqrtstrongly}{}
\newcommand{\squaredstrongly}{}
\newcommand{\forcesqrtstrongly}{1}
\newcommand{\forcesquaredstrongly}{1}
\newcommand{\forcestrongly}{1}
\newcommand{\overstrongly}[1]{#1}
\newcommand{\oneoverstrongly}{}

\section{Accelerated \ac{MXL} without Gradient Information}
\label{section:MXL0plus}

Going back to the heuristic discussion of \ac{MXL0} in the previous section, we see that the one-shot estimator $\basicvopfunction$ is bounded as $\abs{\basicvopfunction} \leq \sup\modulus{\genericfunction}/\delta = \magnitude{1/\delta}$.
This unveils a significant trade-off between the $\magnitude{\delta}$ bias of the estimator and its $\magnitude{1/\delta}$ deviation from the true derivative:
the more accurate $\basicvopfunction$ becomes (smaller bias), the less precise it will be (higher variance).
In the context of iterative optimization algorithms, this \emph{bias\textendash variance dilemma} induces strict restrictions on the design of the query-radius and step-size policies, with deleterious effects on the algorithm's convergence rate (\cf \cref{section:MXL0,section:experiments}).
Motivated by this drawback of the \ac{SPSA} approach,
we proceed in the sequel to design a gradient estimator which requires a single function query per iteration, whilst at the same time enjoying a uniform bound on the norms of the estimates.

\subsection{\ac{SPSA} with callbacks}
\label{section:SPSAplus} 

To proceed with our construction, let $z$ take the value $-1$ or~$+1$ with equal probability, and consider the estimator
\begin{equation}
\label{basicvopO}
\basicvopO{x}{\delta}
 	=\frac{\generic{x + \delta z}-\offset}{\delta} z.
\end{equation}
The offset value $\offset$ is decided \emph{a priori}, independently of the random variable $z$, so that $\expectation{\offset z} = \offset \expectation{z} = 0$.
In turn, this implies that $\expectation{\basicvopO{x}{\delta}} = \basicvtp{x}{\delta}$, and hence:
\begin{equation}
\modulus{\expectation{\basicvopO{x}{\delta} - \dgeneric{x}}}
	\leq \gradlipschitz \delta /2
\end{equation}
\ie the accuracy (bias) of $\basicvopO{x}{}$ is again $\magnitude{\delta}$.

The novelty of \eqref{basicvopO} is as follows:
if we take $\offset = \generic{x}$, then $\modulus{\basicvopO{x}{\delta}} = (1/\delta)\, \modulus{\generic{x + \delta z}-\generic{x}} \leq \Lipschitz $ where~$\Lipschitz$ denotes the Lipschitz constant of~$\genericfunction$,
so the choice $\offset = \generic{x}$ would be ideally suited for our purposes;
however, taking $\offset = \generic{x}$ would also involve an additional function query.
To circumvent this, we will instead approximate $\generic{x}$ with the closest available surrogate, namely the function value observed at the previous iteration of the process.

To make this precise in our \ac{MIMO} context, we will consider the \emph{enhanced} \ac{SPSA} estimator
\acused{SPSA+}
\begin{equation}
\label{estgradratekMO}
\tag{\SPSAplus}
\estgradkn{\nbk}{\itr}
	= \frac{\dimensionk{\nbk}}{\deltan{\itr}} \big[ \rate{ \testmatn{\itr}} - \rate{ \testmatn{\itr-1}} \big] \, \zmatkn{\nbk}{\itr},
\end{equation}
where:
\begin{enumerate}[wide,labelwidth=!,labelindent=5pt]
\item
$\deltan{\itr}$ is the given query radius at time $\itr$.
\item
$\zmatkn{\nbk}{\itr}$ is drawn randomly from the sphere $\trmspherei{\dimensionk{\nbk}-1}$
\item
$\testmatn{\itr}$ 
is the transmit covariance matrix defined along \eqref{linearperturbator}.
\end{enumerate}


\begin{algorithmenv}[t]
\caption{\Acf{MXL+}\label{alg:trMXLMO}}
\removelatexerror
\RestyleAlgo{tworuled}%
\centering
\vspace{-2.5pt}
\begin{algorithm}[H]
\Parameters{ 
$\stepsizen{\itr}$, $\deltan{\itr}$}
 \KwInit{$\itr\leftarrow 1$, $\dualmat\leftarrow 0$;\\
 \hfill
 $\forall\nbk\colon \matk{\nbk} \leftarrow (\powerk{\nbk}/\mmk{\nbk})\, \identityk{\nbk}$, $\offsetk{\nbk}\leftarrow\rate{\mat}$}
\mainnl \Repeat{}{
\mainnl \inlineAteachDIP{$\nbk\in\{1,\dots,\kk\}$}{$\DAplusk{\nbk}{\stepsizen{\itr},\deltan{\itr}}$ }\;
\mainnl $\itr\leftarrow \itr+1$ \;
 }
\smallskip
\rule{0.96\linewidth}{.4pt} \\
\setcounter{AlgoLine}{0}%
\Fctn{$\DAplusk{\nbk}{\stepsize,\delta}$}{
\routinenl \Sampleuniformly{$\trzmatk{\nbk}$}{$\trmspherei{\dimensionk{\nbk}-1}$} \;
\routinenl \Play $ \testmatk{\nbk} \leftarrow \matk{\nbk} + \frac{\delta}{\radiusk{\nbk}} (\cmatk{\nbk} - \matk{\nbk}) + \delta \trzmatk{\nbk} $ \;
\routinenl \Get $ \ratevalue \leftarrow \rate{ \testmat}$ \;
\routinenl \Set $\estgradk{\nbk} \leftarrow \frac{\dimensionk{\nbk}}{\delta} (\ratevalue-\offsetk{\nbk}) \, \zmatk{\nbk} $ \;
\routinenl \Set $\offsetk{\nbk} \leftarrow \ratevalue $ \;
\routinenl \Set $\dualmatk{\nbk} \leftarrow \dualmatk{\nbk} + \stepsize \estgradk{\nbk}$ \;
\routinenl \Set $\matk{\nbk} \leftarrow \normalize{\powerk{\nbk}\,}{} \matsolk{\nbk}{\dualmatk{\nbk}}$ \;
}
\end{algorithm}%
\end{algorithmenv}
Then, integrating \eqref{estgradratekMO} in the chassis of \ac{MXL}, we obtain a similarly enhanced version of \ac{MXL0}, which we call \acdef{MXL+}.
For concreteness, we present a pseudocode implementation of the resulting method in \cref{alg:trMXLMO}.


In terms of parameter values, \ac{MXL+} supports a broad class of policies satisfying the so-called Robbins\textendash Monro conditions:
\noeqref{nsss}%
\begin{equation} 
\hyplabel{nsss}
\begin{array}{lcr}\nocolsep
\textup{(a)}\ \, 
\sum\nolimits_{\itr} \stepsizen{\itr} = \infty
,
&\quad&
\textup{(b)}\ \, 
\sum\nolimits_{\itr} \stepsizen{\itr}^2 <\infty
.
\end{array}
\end{equation}
In addition, \ac{MXL+} also requires the following precautions regarding the allowable step-size and query-radius sequences:
\begin{align}
\hyplabel{trconditionbias}
\textstyle
\sum\nolimits_{\itr} \stepsizen{\itr} \deltan{\itr}
	&< \infty,
	\\[\smallskipamount]
\hyplabel{generalfirstcondition}
\sup\nolimits_{\itr} \stepsizen{\itr}/\deltan{\itr+1}
	&< 2\strongly / (\dimension\Lipschitz\kk),
	\\[\smallskipamount]
\hyplabel{generalsecondcondition} 
\sup\nolimits_{\itr} \deltan{\itr}/\deltan{\itr+1}
	&< \infty,
\end{align}
Of the above conditions, \eqref{generalfirstcondition}\textendash \eqref{generalsecondcondition} guarantee the uniform boundedness of the gradient estimator, while \eqref{trconditionbias} is an additional condition needed for convergence of the algorithm.

In practice, these conditions are easy to verify when $\stepsizen{\itr} = \stepsize / \itr^{\alpha}$ and $\deltan{\itr} = \delta/\itr^{\beta}$ for some $\alpha,\beta > 0$.
In this case, the conditions \eqref{deltacondition}\textendash\eqref{generalsecondcondition} reduce to:
\begin{subequations}
\refstepcounter{hypothesis}
\label{fullhypothesis}
\begin{align}
\label{firstcondition}
\tag{Ha}
\begin{array}{l}
\nocolsep
(\dimension\Lipschitz \kk/2) \, \stepsize
	< \strongly \, \delta
	< \strongly 1/\sqrt{\maxmm(\maxmm-1)},
\end{array}
\\
\label{secondcondition}
\tag{Hb}
\begin{array}{l}
\nocolsep
0 \leq \beta \leq \alpha \leq 1
\quad\text{and}\quad
\alpha + \beta >1,
\end{array}
\end{align}
\end{subequations}

With all this in hand, we are finally in a position to state our main convergence results for the \ac{MXL+} algorithm.
We begin by establishing the algorithm's almost sure convergence:
 
\begin{theorem}
[Convergence of \ac{MXL+}]
\label{thm:lasttrMXLMO}
Suppose that \ac{MXL+} \textpar{\cref{alg:trMXLMO}} is run with step-size and query-radius policies satisfying \eqref{deltacondition}\textendash\eqref{generalsecondcondition}.
Then, with probability $1$, the sequence of the users' transmit covariance matrices $\testmatn{\itr}$ converges to the solution set of \eqref{eq:trproblem}.
\end{theorem}

As in the case of \cref{thm:lasttrMXL}, \cref{thm:lasttrMXLMO} provides a strong asymptotic convergence result, but it leaves open the crucial question of the algorithm's convergence speed.
Our next result justifies the introduction of \eqref{estgradratekMO} and shows that \cref{alg:trMXLMO} achieves the best of both worlds:
one-shot throughput measurements with an $\tilde{\mathcal{O}}(1/\sqrt{T})$ convergence rate.

\begin{theorem}[Convergence rate of \ac{MXL+}]
\label{timecomplexity:trMXLMO} \
Suppose that \ac{MXL+} \textpar{\cref{alg:trMXLMO}} is run for $T$ iterations.
We then have:

\begin{enumerate}
\item
\label{thm:trMXLMOi}
If $\stepsizen{\itr} = \stepsize/\sqrt{\itr}$ and $\deltan{\itr} = \delta/\sqrt{\itr}$ with $\stepsize$ and $\delta$ satisfying \eqref{firstcondition}:
\begin{equation}
\label{asymptoticrate}
\expectation{\optrate - \rate{ \meanmatn{T}}}
	= \Magnitude{\frac{\log T}{\sqrt{T}}}.
\end{equation}

\item
\label{thm:trMXLMOiinew}
If $\stepsizen{\itr} = \stepsize/\sqrt{T}$ and $\deltan{\itr} = \delta/\sqrt{T}$ with $\stepsize$ and $\delta$ satisfying \eqref{firstcondition}:
\begin{enumerate}
\addtolength{\itemsep}{.5ex}
\item
\label{thm:trMXLMOiifirst}
In expectation, 
\begin{equation}
\label{troptrateboundmeanexpectation}
\expectation{\optrate - \rate{ \meanmatn{T}} }
	\leq \frac{\Bcoefficient{\stepsize}{\delta}}{\sqrt{T}}
	=\Magnitude{\frac{1}{\sqrt{T}}},
\end{equation}
where
\(
\Bcoefficient{\stepsize}{\delta}
	= (\kk/\stepsize) \log\maxmm
		+ 4\kk^2\meanlipschitz \delta
		+ \frac
			{\twosquaredtwitchFlipschitz \kk \dimension \squaredstrongly \stepsize}
			{[2/(\dimension\Lipschitz\kk) - \stepsize/\delta]^{2}}
\).

\item
\label{thm:trMXLMOiisecond} 
In probability, for any small enough tolerance $\zmargin>0$,
\begin{equation}\label{troptrateboundmeaninterval}
\Proba{\optrate -\rate{ \meanmatn{T}} \geq \frac{\Bcoefficient{\stepsize}{\delta}}{\sqrt{T}} + \zmargin}
	\leq \exp\left(- \frac{\zmargin^2 T}{\Ccoefficient{\stepsize}{\delta}} \right),
\end{equation}
where
\(
\Ccoefficient{\stepsize}{\delta}
	= \twopowersevensquaredtwitchFlipschitz \dimension \kk^{2} \squaredstrongly
		\big/ [2/(\dimension\Lipschitz\kk) - \stepsize/\delta]^{2}
\). 
\end{enumerate}
\end{enumerate}
\end{theorem}

 
Importantly, \cref{timecomplexity:trMXLMO} shows that \ac{MXL+} recovers the $\magnitude{1/\sqrt{T}}$ convergence rate of \ac{MXL} with \emph{full} gradient information, even though the network's users are no longer assumed to have \emph{any} access to a gradient oracle.
In fact, the guarantees of \cref{timecomplexity:trMXLMO} can be optimized further by finetuning the choice of $\stepsize$ and $\delta$;
doing just that (and referring to \cref{appendix:MXLzeroplus} for the details), we have:


\begin{table}[!t]
\heavyrulewidth=0.8pt
\lightrulewidth=0.4pt
\renewcommand{\arraystretch}{1.3}
\setlength{\tabcolsep}{2pt}
\caption{Parameters of \ac{MXL+} for \cref{convergenceratres:trMXLMO}}
\centering
\begin{tabular}{rl}
\toprule
\ref{cr:trMXLMOi}) 
& 
\(
\displaystyle
\stepsizecoefficient{\errorproba}
	= \frac{\sqrt{\log\maxmm \big/ (\dimension\Lipschitz\kk^{2})}}{\sqrt{\meanlipschitz} + \sqrt{2\dimension\Lipschitz\kk}}
;
\quad
\deltacoefficient{\errorproba}
{\,=\,}
\frac{\sqrt{(\dimension\Lipschitz/\meanlipschitz)\log\maxmm}}{2}
;
\quad
T\geq \frac{\Lipschitz\maxmm^4\log\maxmm}{4\meanlipschitz}
\)
\\
\midrule
\ref{cr:trMXLMOii}) 
&
$
\stepsizecoefficient{\errorproba}
=
%
%
\frac{\renrmlzdcsteps{\errorproba}}{ \sqrt{\Lipschitz}}   \bigg[ {\scriptstyle \sqrt{ \twitchFlipschitz   \Lipschitz }\dimension\, \renrmlzdcsteps{\errorproba} } +  \sqrt{\sqrt{{2\halftwitchFlipschitz\strongly}/{\logarithm{1/\errorproba}}} \, \strongly \meanlipschitz }  \bigg]^{-1} \frac{ \sqrt{ \log\maxmm} }{\kk \sqrt{\dimension}} 
$ ;
\\
&
$
\deltacoefficient{\errorproba}
=
%
%
 \frac{\renrmlzdcsteps{\errorproba}}{2}  \sqrt{{   \Lipschitz}/{\strongly \meanlipschitz}  \sqrt{{ \logarithm{1/\errorproba}  }/{ \twitchFlipschitz  \strongly}  }  \, \dimension    \log\maxmm  }
 $ ;
\\
&
$ 
T 
=
%
%
{ 4  \repownrmlzdcsteps{\errorproba}{4}  \Lipschitz^2 } {\bigg[ \forcesqrthalftwitchFlipschitzoneoverstrongly {+}         \frac{1}{\renrmlzdcsteps{\errorproba} \dimension }\sqrt{  {2\meanlipschitz}/{\strongly\Lipschitz}\sqrt{{2}/{ \logarithm{1/\errorproba}} }} \bigg]^2}  \Big(\frac{\logarithm{1/\errorproba}\kk^4 \dimension^{3}  \log\maxmm }{\errortolerance^2} \Big)
$ ;
\hspace{6mm}
\\ & 
\hfill
with 
$ 
\renrmlzdcsteps{\errorproba}
= 
%
\Big[     {1}/{\sqrt{\logarithm{1/\errorproba}}} + {4\sqrtstrongly}/{\sqrt{ \logarithm{\maxmm}}} \Big]^{1/ 2} 
$
\\
\bottomrule
\end{tabular}
\label{table:convergenceratres:trMXLMO}
\end{table}


\begin{corollary}
\label{convergenceratres:trMXLMO}
Suppose that  \ac{MXL+} is run with $\stepsizen{\itr} = \stepsize/\sqrt{T}$, $\deltan{\itr} = \delta/\sqrt{T}$, and $T$, $\stepsize$, $\delta$ as in \cref{table:convergenceratres:trMXLMO}.
Then:
\begin{enumerate}
[a),ref=\alph*,wide,labelwidth=!,labelindent=5pt]
\item
\label{cr:trMXLMOi}
In expectation, we have:
\begin{equation}
\expectation{\optrate - \rate{ \meanmatn{T}}}
	\leq 
		2\Lipschitz\left( \forcesqrthalftwitchFlipschitzoneoverstrongly + 
        \frac{2^{3/4}\sqrt{ { \meanlipschitz}/{  \strongly\Lipschitz}  }}{\maxdimension}  \right) 
		\sqrt{\frac{\kk^4\maxmm^6 \log\maxmm  }{T}} 
.
\end{equation}

\item
\label{cr:trMXLMOii}

In probability, given a small enough tolerance $\errortolerance>0$ and a confidence level $1- \errorproba\in(0,1)$, we have:
\begin{equation}
\label{standardtroptrateboundmeaninterval}
\proba{\optrate - \rate{\meanmatn{\ceil{T}}} \leq \errortolerance}
	\geq 1 - \errorproba
.
\end{equation}
\obsolete{with $
T 
\equiv
T(\kk,\maxmm,\errortolerance,\errorproba)
=
\Bigtheta{{\logarithm{1/\errorproba}\kk^4 \maxmm^6  \log\maxmm}/{\errortolerance^{2} }}
$.}
\end{enumerate}
\end{corollary}

An important feature of the convergence rate guarantee \eqref{standardtroptrateboundmeaninterval} is that it does not depend on the number of antennas $\nn$ at the receiver.
As such, \cref{alg:trMXLMO} exhibits a \emph{scale-free} behavior relative to $\nn$, which makes it particularly appealing for distributed massive-\ac{MIMO} systems.
In the next section, we further relax the requirement that all users update their transmit covariance matrices in a synchronous manner, and we derive a fully distributed version of the \ac{MXL+} algorithm.

\section{Distributed Implementation}
\label{section:desynchronized}

In this section, we propose a distributed variant of the \ac{MXL+} method which can account for randomized and asynchronous user decisions (independent or in alternance with other users).
Specifically, we now assume that, at each stage of the process, only a random subset of users perform an update of their individual covariances matrices, while the remaining users maintain the same covariance matrix, without updating.

To state this formally, suppose that a random subset of users $\Updatingn{\itr} \subseteq \users \equiv \{1,\dotsc,\kk\}$ is drawn at stage $\itr$ following an underlying probability law $\probas \equiv (\probasU{\Updating})_{\Updating\subseteq\users}$ (\ie $\Updating\subseteq\users$ is drawn with probability $\probasU{\Updating}$).
%
From the distributed  perspective of individual users, we write $\probk{\nbk} = 
\sum_{\Updating \ni \nbk} \probasU{\Updating}$ to denote the \emph{marginal probability} that user~$\nbk$ updates their covariance at any stage~$\itr$;
as such, the participation of all users is enforced by imposing the condition $\probk{\nbk} > 0$. 
We thus obtain the asynchronous \ac{MXL+} scheme:
\begin{equation}
\label{randomMXL}
\tag{\randomMXL}
\begin{aligned}
\matn{\itr}
	&= \aggregatematsol{\dualmatn{\itr}},
	\\
\dualmatn{\itr+1}
	&= 
	\dualmatn{\itr}
		+ \stepsizen{\itr} \, 
		\partialestgradn{\itr}
,
\end{aligned}
\end{equation} 
where $\partialestgradkn{\nbk}{\itr} = \estgradkn{\nbk}{\itr}$ if $\nbk\in\Updatingn{\itr}$, and  $\partialestgradkn{\nbk}{\itr}=0$ otherwise. 
For a pseudocode implementation, see also \cref{alg:randomtrMXLMO} above.


\begin{algorithmenv}[t]
\caption{The \ac{AMXL+} method\label{alg:randomtrMXLMO}}
\removelatexerror
\RestyleAlgo{tworuled}%
\centering
\vspace{-2.5pt}
\begin{algorithm}[H]
  \Parameters{ 
$\probas$, $\stepsizen{\itr}$, $\deltan{\itr}$}
 \KwInit{$\itr\leftarrow 1$, $\dualmat\leftarrow 0$;\\
 \hfill
 $\forall\nbk\colon \matk{\nbk} \leftarrow (\powerk{\nbk}/\mmk{\nbk})\, \identityk{\nbk}$, $\offsetk{\nbk}\leftarrow\rate{\mat}$}
 \mainnl \Repeat{}{
\mainnl    \Drawfrom{$\Updating$}{$\probas$} \;
 \mainnl   \Ateach{$\nbk\in\{1,\dots,\kk\}$}{
 \mainnl    \inlineIfElse{$\nbk\in\Updating$}{$\DAplusk{\nbk}{\stepsizen{\itr},\deltan{\itr}}$}{$\passkfunction{\nbk}$} \;
 } 
 \mainnl $\itr \leftarrow \itr+1$ \;
 }
\smallskip
\rule{0.96\linewidth}{.4pt} \\
 \setcounter{AlgoLine}{0}
\Fctn{$\passkfunction{\nbk}$}{
\routinenl \Play  $ \matk{\nbk}  $ \;
\routinenl \Get $ \offsetk{\nbk} \leftarrow \rate{ \mat}$ \; 
}
\end{algorithm}%
\end{algorithmenv} 


As we show below, \ac{AMXL+} recovers the $\magnitude{1/\sqrt{T}}$ convergence rate of \ac{MXL+}, despite being distributed across users:

\begin{theorem}
[Convergence rate of  \ac{AMXL+}]
\label{timecomplexity:randomtrMXLMO}
Suppose that \ac{AMXL+} \textpar{\cref{alg:randomtrMXLMO}} is run for $T$ iterations.
We then have:

\begin{enumerate}
\item
\label{alg:randomtrMXLMOi}
If $\stepsizen{\itr} = \stepsize/\sqrt{\itr}$ and $\deltan{\itr} = \delta/\sqrt{\itr}$ with $\stepsize$ and $\delta$ satisfying \eqref{firstcondition}:
\begin{equation}
\expectation{\optrate - \rate{\meanmatn{T}}}
	= \Magnitude{\frac{\log T}{\sqrt{T}}}.
\end{equation}

\item
\label{alg:randomtrMXLMOiinew}
If $\stepsizen{\itr} = \stepsize/\sqrt{T}$ and $\deltan{\itr} = \delta/\sqrt{T}$ with $\stepsize$ and $\delta$ satisfying \eqref{firstcondition}:

\begin{enumerate}
\item
\label{alg:randomtrMXLMOiifirst} 
In expectation, 
\begin{equation}
\label{randomtroptrateboundmeanexpectation}
\expectation{\optrate - \rate{ \meanmatn{T}} }
	\leq \frac{\Bcoefficientprob{\prob}{\stepsize}{\delta}}{\sqrt{T}}
		= \Magnitude{\frac{1}{\sqrt{T}}},
\end{equation} 
where
\(
\Bcoefficientprob{\prob}{\stepsize}{\delta}  
=  
\sum_{\nbk=1}^{\kk} \frac{\log\mmk{\nbk}}{\probk{\nbk}\stepsize}
		+ 4\kk^2\meanlipschitz \delta
		+ \frac
			{\twosquaredtwitchFlipschitz \kk \dimension \squaredstrongly \stepsize}
			{[2/(\dimension\Lipschitz\kk) - \stepsize/\delta]^{2}}
\).
  
\item
\label{alg:randomtrMXLMOiisecond} 
In probability, for any small enough tolerance $\zmargin>0$,
\begin{equation}
\label{randomtroptrateboundmeaninterval}
\Proba{\optrate -\rate{ \meanmatn{T}} \geq \frac{\Bcoefficientprob{\prob}{\stepsize}{\delta}}{\sqrt{T}} + \zmargin}
	\leq \exp\left(- \frac{\zmargin^2 T}{\Ccoefficientprob{\prob}{\stepsize}{\delta}} \right),
\end{equation}
where
\(
\Ccoefficientprob{\prob}{\stepsize}{\delta} 
{\,=\,}  
\Big[
	\halftwitchFlipschitz \forcestrongly
		+ \halftwitchFlipschitz \frac{\probdist{\prob}}{2} \strongly
		+ \frac
			{\squaredhalftwitchFlipschitz \strongly \probdist{\prob} \dimension^{3/2} \Lipschitz \kk  \stepsize\delta}
			{2\strongly\delta-\dimension \Lipschitz\kk \stepsize}
	\Big]^{2}
\Ccoefficient{\stepsize}{\delta}
\),
with
$\probdist{\prob} = \kk^{-1} \sum_{\nbk=1}^{\kk} \max(1,\probk{\nbk}^{-1} - 1)$
and
$\Ccoefficient{\stepsize}{\delta}$ as in \cref{timecomplexity:trMXLMO}.
\end{enumerate}
\end{enumerate}
\end{theorem}
Note here that the quantity $\Bcoefficientprob{\prob}{\stepsize}{\delta}$ above only differs from its counterpart $\Bcoefficient{\stepsize}{\delta}$ of \cref{timecomplexity:trMXLMO} in the first term, which measures the cost of asynchronicity in terms of expected convergence. 
A similar increase in the deviation from the mean transpires through an impeding factor in the expression for $\Ccoefficientprob{\prob}{\stepsize}{\delta}$, quantifying the impact of asynchronicity in both mean and fluctuation terms.

In \cref{appendix:randomMXLzeroplus}, we show how the parameters $\vect{\stepsize,\delta}$ can be optimized for general $\probas$;
for concreteness, we present below the particular case where at any stage each user is active with probability $\probk{\nbk} = 1/\kk$:

\begin{table}[t]
\heavyrulewidth=0.8pt
\lightrulewidth=0.4pt
\renewcommand{\arraystretch}{1.3}
\setlength{\tabcolsep}{2pt}
\label{table:convergenceratres:randomtrMXLMO}
\centering
\renewcommand{\renrmlzdcstepsprobfunction}[1]{\hat\psi}%
\caption{Parameters of \CDrandomMXL{} of Corollary \ref{convergenceratres:randomtrMXLMO}}
\begin{tabular}{rl}
\toprule
\ref{cr:randomtrMXLMOi}) 
& 
$
\stepsizecoefficient{\errorproba}
{\,=\,}
%
\Big[{ 1+\sqrt{\frac{2\Lipschitz\kk}{\meanlipschitz}}\dimension} \Big]^{-1}\sqrt{\frac{\log\maxmm}{\meanlipschitz\Lipschitz\kk\dimension}}
\ ;
\ 
\deltacoefficient{\errorproba}
{\,=\,}
\frac{1 }{2}\sqrt{\frac{\Lipschitz\kk\dimension\log\maxmm}{\meanlipschitz}} 
$ ;
\
$T\geq \frac{1}{4\meanlipschitz} [{\Lipschitz\kk\maxmm^4\log\maxmm}]$
\\
\midrule
\ref{cr:randomtrMXLMOii}) 
&
$
\stepsizecoefficient{\errorproba}
=
{\strongly \halftwitchFlipschitz \sqrtsqrthalftwitchFlipschitz \sqrtsqrtstrongly  \renrmlzdcstepsprob{\prob}{\errorproba}}\Big[{\sqrt{  \renrmlzdstuff{\errorproba} \Lipschitz\meanlipschitz}  {+}   2 \halftwitchFlipschitz \sqrtsqrthalftwitchFlipschitz \sqrtsqrtstrongly \, \powlogarithm{{\frac1\errorproba}}{3/ 8} \, \renrmlzdstuff{\errorproba}\renrmlzdcstepsprob{\prob}{\errorproba}\, \Lipschitz \sqrthalftwitchFlipschitz \sqrtstrongly [\kk \dimension]^{{3}/{4}}}\Big]^{-1}
\bigg( \frac{ \powlogarithm{{1/\errorproba}}{1/ 8} \sqrt{\log\maxmm} }{ [\kk\dimension]^{{3}/{4}} } \bigg)
$ ;
\\
&
$
\deltacoefficient{\errorproba}
=
\Big(\frac{\renrmlzdcstepsprob{\prob}{\errorproba}}{2} \sqrt{\frac{  \squaredhalftwitchFlipschitz \strongly    \Lipschitz }{\sqrtstrongly \sqrthalftwitchFlipschitz  \renrmlzdstuff{\errorproba} \meanlipschitz }}\Big)\, \big[\sqrt{\logarithm{1/\errorproba}} \, \kk\dimension\big]^{{1}/{4}}\sqrt{\log\maxmm}
$ ;
\\ 
&
$
T 
{\,=\,}
%
16 \Lipschitz^2 {\bigg[}  \squaredhalftwitchFlipschitz \strongly    { \repownrmlzdcstepsprob{\prob}{\errorproba}{2}}   {+} \frac{\renrmlzdcstepsprob{\prob}{\errorproba} \, \sqrt{{ \sqrtstrongly\meanlipschitz  \halftwitchFlipschitz\sqrthalftwitchFlipschitz    }/[{ \renrmlzdstuff{\errorproba} \Lipschitz } ] }}{\powlogarithm{{1/\errorproba}}{{3/ 8}}\,\kk^{ 3/ 4}\dimension}    {\bigg]^2} \Big(\frac{\logarithm{{1/\errorproba}}\kk^6\dimension^3 \log\maxmm}{\errortolerance^2}
\Big)$ ;
\hspace{18mm}
\\
& \hfill
with 
$
\renrmlzdcstepsprob{\prob}{\errorproba} 
= 
%
%
\Big[ \frac{   \renrmlzdstuff{\errorproba}  }{ \sqrt{\kk\sqrt{\logarithm{1/\errorproba}}} } + \sqrt{{2 }/{ \logarithm{\maxmm}}} \, (1+{1}/{\kk})   \Big]^{1/ 2}
$,
\\
& \hfill
$
\renrmlzdstuff{\errorproba} 
=
%
%
\Big[ \sqrt{2  }  \squaredhalftwitchFlipschitz \strongly    (1-{1}/{\kk})  +  \frac{1}{2\kk \sqrt{\logarithm{1/\errorproba}} } \Big]^{ 1/ 2}
$ 
\\
\bottomrule
\end{tabular}
\end{table}
%
%
\begin{corollary}
[Uniform \ac{AMXL+}, $\kk{\,\geq\,}2$]
\label{convergenceratres:randomtrMXLMO}
\renewcommand{\renrmlzdcstepsprobfunction}[1]{\hat\psi}%
Suppose that  \ac{AMXL+} is run with  $\probk{1}=\cdots=\probk{\kk}= 1 /\kk$,   $\stepsizen{\itr} = \stepsize/\sqrt{T}$, $\deltan{\itr} = \delta/\sqrt{T}$, and $T$, $\stepsize$, $\delta$ as in \cref{table:convergenceratres:randomtrMXLMO}.
Then:

\begin{enumerate}[a),ref=\alph*,wide,labelwidth=!,labelindent=5pt]
\item \label{cr:randomtrMXLMOi}
In expectation, we have:
\begin{equation} \label{CD-expectation}
\optrate -\Expectation{ \rate{ \meanmatn{T}} }
\leq 
2\Lipschitz\left( \forcesqrthalftwitchFlipschitzoneoverstrongly +         \sqrt{\meanlipschitz/\strongly\Lipschitz}
\right)  \sqrt{\frac{\kk^5\maxmm^6 \log\maxmm  }{T}}
.
\end{equation}

\item \label{cr:randomtrMXLMOii}
In probability, given a small enough tolerance $\errortolerance>0$ and a confidence level $1- \errorproba\in(0,1)$, we have:
\begin{equation}
\proba{ 
\optrate -\rate{ \meanmatn{\ceil T}}
\leq 
\errortolerance }
 \geq 1-\errorproba
 .
 \end{equation}
\end{enumerate}
\end{corollary}

\begin{remark}[Coordinate descent]\label{remark:CD}
The case $\probasU{\{1\}}=\dots=\probasU{\{\kk\}} = 1/\kk$ where a single user is active at each time step with probability $\probk{\nbk} = 1/\kk$ covers the alternated optimization scheme known as \ac{CD}\textemdash the coordinates in this context refer to the wireless users.
In this regard, Corollary~\ref{convergenceratres:randomtrMXLMO} provides us with a quantification of the impact of alternation on the convergence speed of \ac{MXL+}.
Looking for instance at  Corollaries \ref{convergenceratres:trMXLMO}(\ref{cr:trMXLMOi})   and \ref{convergenceratres:randomtrMXLMO}(\ref{cr:randomtrMXLMOi}), we observe that the expected convergence of the time average, if regarded as a function of the total number $n$ of user updates, is %
$\magnitude{\sqrt{{\kk^5\maxmm^6 \log\maxmm  }/{n}}}$ 
both for the synchronized algorithm \ac{MXL+} and for \ac{CD}.
The impact of the network size $\kk$ on the number of user updates needed for $\errortolerance$-convergence with probability $1-\errorproba$, however, is more pronounced by an order of magnitude for \ac{CD}, $\Bigtheta{{\logarithm{1/ \errorproba}\kk^6 \maxmm^6  \log\maxmm}/{\errortolerance^{2} }}$, than it is for \ac{MXL+}, $ \Bigtheta{ {\logarithm{ 1/\errorproba}\kk^5 \maxmm^6  \log\maxmm}/{\errortolerance^{2} }} $.
\end{remark}

\section{Numerical Experiments}
\label{section:experiments}

\newcommand\KK{20}%
\newcommand\NN{16}%
\newcommand\MM{3}%

In this section, we perform a series of experiments to validate our results in realistic network conditions.
Throughout what follows, and unless specified otherwise, our numerical experiments are performed in a simulated wireless network setup with parameters as summarized in \cref{tab:params}. 
In more detail, we consider a cellular wireless network occupying a central frequency of $f_{c} = 2.5\,\mathrm{GHz}$ and a total bandwidth of $10\,\mathrm{MHz}$.
\revise{Signal propagation in the wireless medium is modeled following the widely utilized COST 2100 channel model for moderately dense urban environments \cite{COST2100}.
This is a geometry-based stochastic extension of the original COST Hata model \cite{COST99} which has been designed to reproduce the stochastic properties of \ac{MIMO} channels over the frequency, space and time domains.
As such, even though it is not 5G-specific, the COST 2100 model is generic and flexible, making it suitable to model a broad range of multi-user or distributed \ac{MIMO} scenarios \cite{COST2100}.}

Network coverage is provided by a \ac{BS} with an effective service radius of $1\,\mathrm{km}$ (for the wider network in play, we consider a hexagonal cell coverage structure).
The \ac{BS} serves the uplink of $\kk$ wireless transmitters that are positioned uniformly at random within the coverage area following a homogeneous Poisson point process.
All communications occur over a \ac{TDD} transmission scheme with an asynchronous frame duration of $T_{f} = 5\,\mathrm{ms}$.
Finally, in line with state-of-the-art mobile and portable device specifications, transmitting devices are assumed to have a maximum transmit power of $33\,\mathrm{dBm}$.


\begin{table}[t]
\centering
\renewcommand{\arraystretch}{1.1}

\begin{tabular}{|c|c|}
\hline
\textbf{Parameter}
	&\textbf{Value}
	\\
	\hline
Time frame duration
	&$5\,\mathrm{ms}$
\\
\hline
\beginrev
\ac{MIMO} channel model
	&\beginrev
	COST 2100 \cite{COST2100}
\\
\hline
BS/MS antenna height
	&$32\,\mathrm{m}$ / $1.5\,\mathrm{m}$
\\
\hline
Central frequency
	&$2.5\,\mathrm{GHz}$
\\
\hline
Total bandwidth
	&$11.2\,\mathrm{MHz}$
\\
\hline
Spectral noise density ($20\,^{\circ}\textrm{C}$)
	&$-174\,\mathrm{dBm}/\mathrm{Hz}$
\\
\hline
Maximum transmit power
	&$P = 33\,\mathrm{dBm}$
\\
\hline
Transmit antennas per device
	&$\mm\in\{2,4,8\}$
\\
\hline
Receive antennas
	&$\nn=128$
\\
\hline
\end{tabular}
\caption{Wireless network simulation parameters.}
\label{tab:params}
\vspace{-2ex}
\end{table}


\subsection{Comparison with \acl{WF} methods}

We begin by examining the performance of \ac{MXL}-type methods relative to conventional \acl{WF} schemes.
To provide a broad basis for this comparison, we focus on two complementing scenarios:
\begin{enumerate*}
[(\itshape a\upshape)]
\item
the \emph{full feedback} case,
\ie when transmitters are assumed to know their individual channel matrices $\chmatk{\nbk}$ and the induced signal-plus-noise covariance matrix $\contrzmat$;
and
\item
the \emph{limited feedback} case,
\ie when transmitters only observe their realized utility (\ie their sum rate).
\end{enumerate*}
For the purposes of our experiments,
\revise{and in line with other recent works on large antenna arrays \cite{LETM14,SBH19,HtBD13,RPL+13},}
we consider a system with $\kk=60$ users, each with $2$, $4$ or $8$ transmit antennas, and a \ac{BS} with $\nn=128$ receive antennas;
all other network parameters are as in \cref{tab:params}.


\begin{figure*}[t]
\begin{subfigure}[b]{.475\linewidth}
\includegraphics[width=\textwidth]{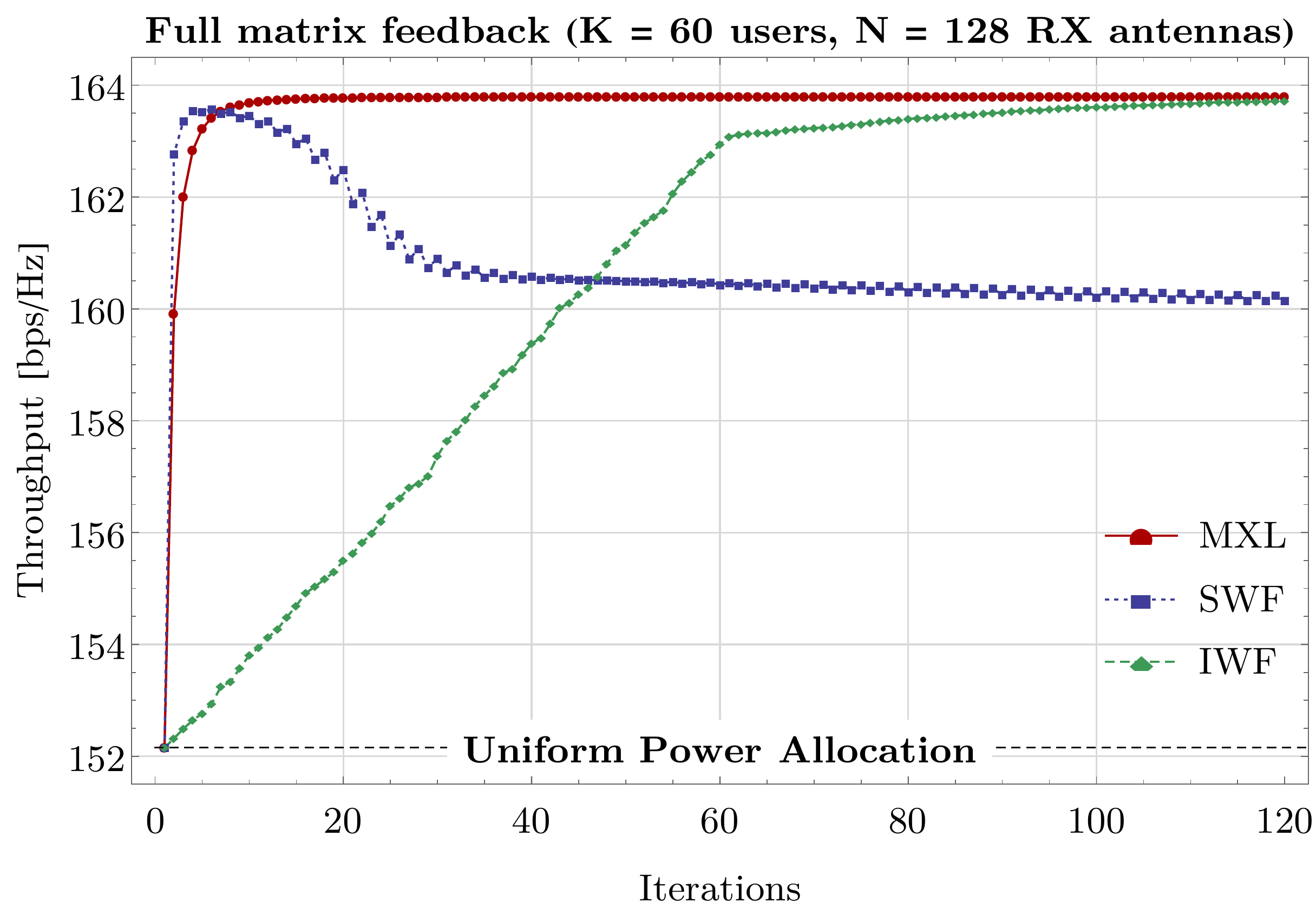}
\caption{Full matrix feedback}
\label{fig:MXL-WF-full}
\end{subfigure}
\hfill
\begin{subfigure}[b]{.48\linewidth}
\includegraphics[width=\textwidth]{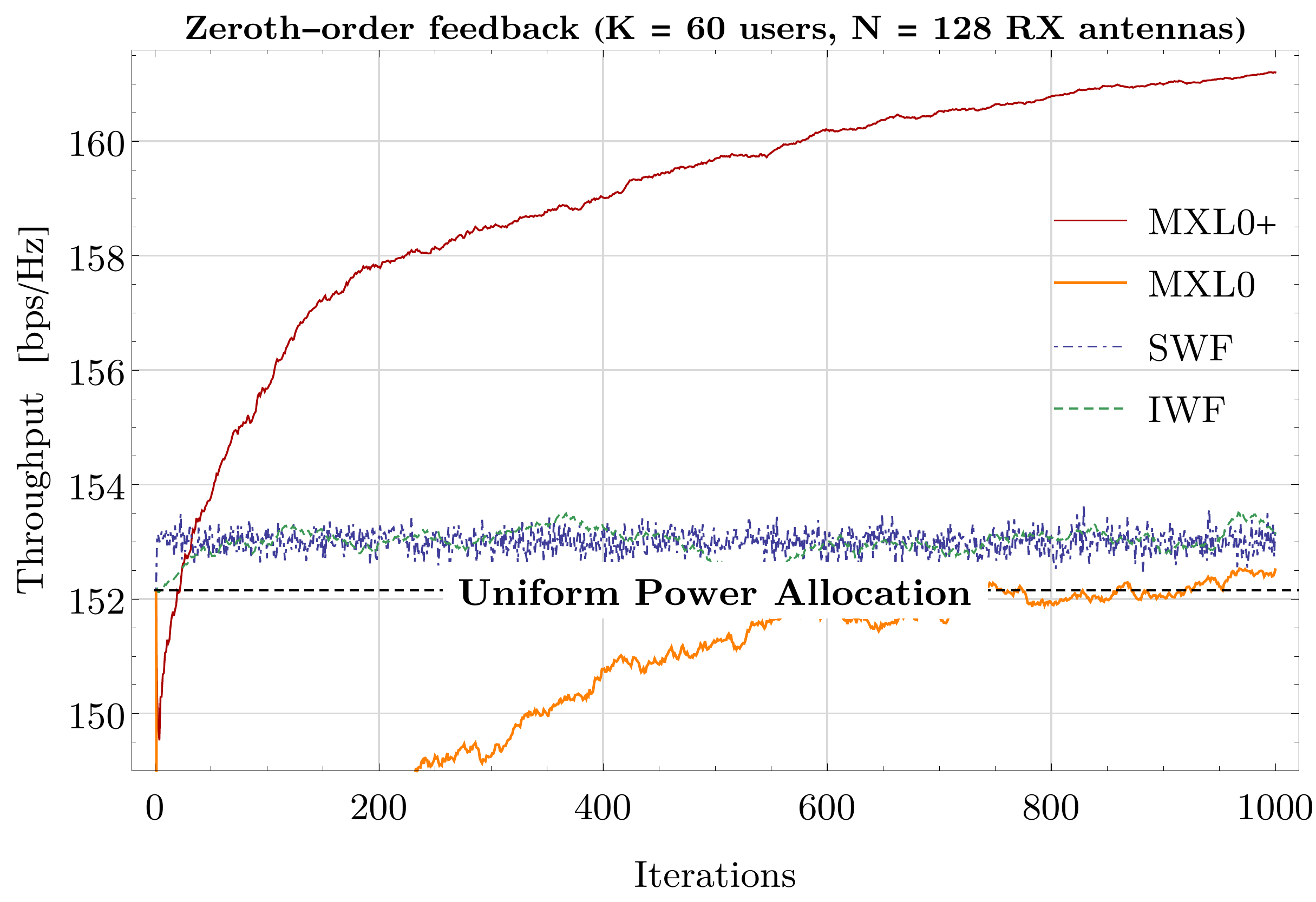}
\caption{Limited, zeroth-order feedback}
\label{fig:MXL-WF-zeroth}
\end{subfigure}
\caption{%
Comparison between \acl{WF} and \acl{MXL} in a wireless network with $\kk=60$ users and $\nn=128$ receive antennas.
In the full feedback case (left), the transmitters are running \ac{MXL} against \ac{IWF}/\ac{SWF} with full matrix information (perfect knowledge of effective channel matrices, system-wide signal-plus-noise covariance matrix, etc.).
In the zeroth-order case (right), the transmitters only have access to their realized utility (the achieved throughput) and are running \ac{MXL0} and \ac{MXL+} against \ac{IWF}/\ac{SWF} with one-shot pilot estimates of the required matrix information.
In both instances, \ac{MXL}/\ac{MXL+} exhibits consistent \textendash\ and significant \textendash\ performance gains over \ac{WF} methods.}
\vspace{-2ex}
\label{fig:MXL-WF}
\end{figure*}


In the first case (full matrix feedback), we simulated the iterative and simultaneous variants of \acl{WF} against the \ac{MXL} algorithm as presented in \cref{subsec:MXL}.
The iterative \ac{WF} variant converges to an optimum solution;
however, because user updates need to be taken in a sequential, round-robin fashion, the algorithm's convergence speed is inversely proportional to the number of users in the system, and hence quite slow.
On the other hand, the simultaneous \ac{WF} variant achieves significant performance gains within the first few iterations, but because it has no way of mitigating conflicting user updates, these gains subsequently evaporate and the algorithm converges to a suboptimal state.
By comparison, the \ac{MXL} algorithm achieves convergence to an optimal state within a few iterations, without suffering from the slow convergence speed of the iterative \ac{WF} algorithm or the oscillatory behavior of its simultaneous counterpart.
The results of these simulations are plotted in \cref{fig:MXL-WF-full}.

Moving forward, to establish a fair \revise{comparison} in the limited feedback case, we consider a \revise{baseline} setting where, at each transmission frame $\itr = 1,2\dotsc$, each user has access to one-point pilot estimates of their effective channel matrix $\tilde{\mathbf{H}}_{k} = \chmatk{\nbk} \contrzmatk{\nbk}^{-1/2}$ (\eg via randomized directional sampling) \cite{LBON16}.
Since $\contrzmatk{\nbk} \equiv \contrzmatk{\nbk}(\matn{\itr})$ evolves over time (because of the signal covariance modulation $\matn{\itr}$ of all other users in the network), these measurements must be repeated over time;
otherwise, knowledge of $\chmatk{\nbk}$ alone would not suffice to run \acl{WF} in a multi-user environment.
By comparison, for the \ac{MXL+} algorithm, we only assume that users observe their realized throughput as described in detail in \cref{section:MXL0plus}.

The results of our simulations are plotted in \cref{fig:MXL-WF-zeroth}.
Because \acl{WF} methods require perfect knowledge of $\tilde{\mathbf{H}}_{k}$ at each transmission frame, the imperfections introduced by one-point pilot contamination effects cause a complete breakdown of the algorithm's convergence.
In particular, both iterative and simultaneous variants fail to exhibit any significant performance gains over a uniform (isotropic) input signal covariance profile.
The performance of \ac{MXL0} is underwhelming in the first iterations (due to exploration), but it improves steadily over time;
however, this improvement is very slow over the simulation window.
On the other hand, the callback mechanism of \ac{MXL+} achieves dramatically better results, even with one-shot, zeroth-order feedback.

In terms of per-iteration computational complexity, \cref{fig:runtime} compares the wall-clock runtime of an iteration of each algorithm (\ac{IWF}, \ac{SWF} and \ac{MXL+}).
All computations were performed in a commercial laptop with 16 GB RAM and a 2.6 GHz 6-core Intel i7 CPU;
for statistical significance, they were averaged over $S = 1000$ sample runs.
Network parameters were as above, except for the number of receive antennas which was taken in the range $\{4, \hdots, 64\}$ to assess scalability.
For small values of $\nn$, \ac{IWF} has the fastest runtime per iteration because only one user updates per iteration and the inversion of the \ac{MUI} matrix at the receiver is relatively fast.
However, for larger values of $\nn$, this advantage evaporates and \ac{MXL+} becomes the fastest because the \ac{SPSA} estimator is sparse, so the resulting matrix operations are the lightest.
This provides an additional layer to the results of \cref{fig:MXL-WF}:
even though \ac{IWF}/\ac{SWF} methods fail to produce any measurable performance gains in limited feedback environments,
\ac{MXL+} remains optimal and achieves considerably better throughput values, all with a lighter per-iteration runtime. 



\begin{figure}[t]
\centering
\includegraphics[width=.95\columnwidth]{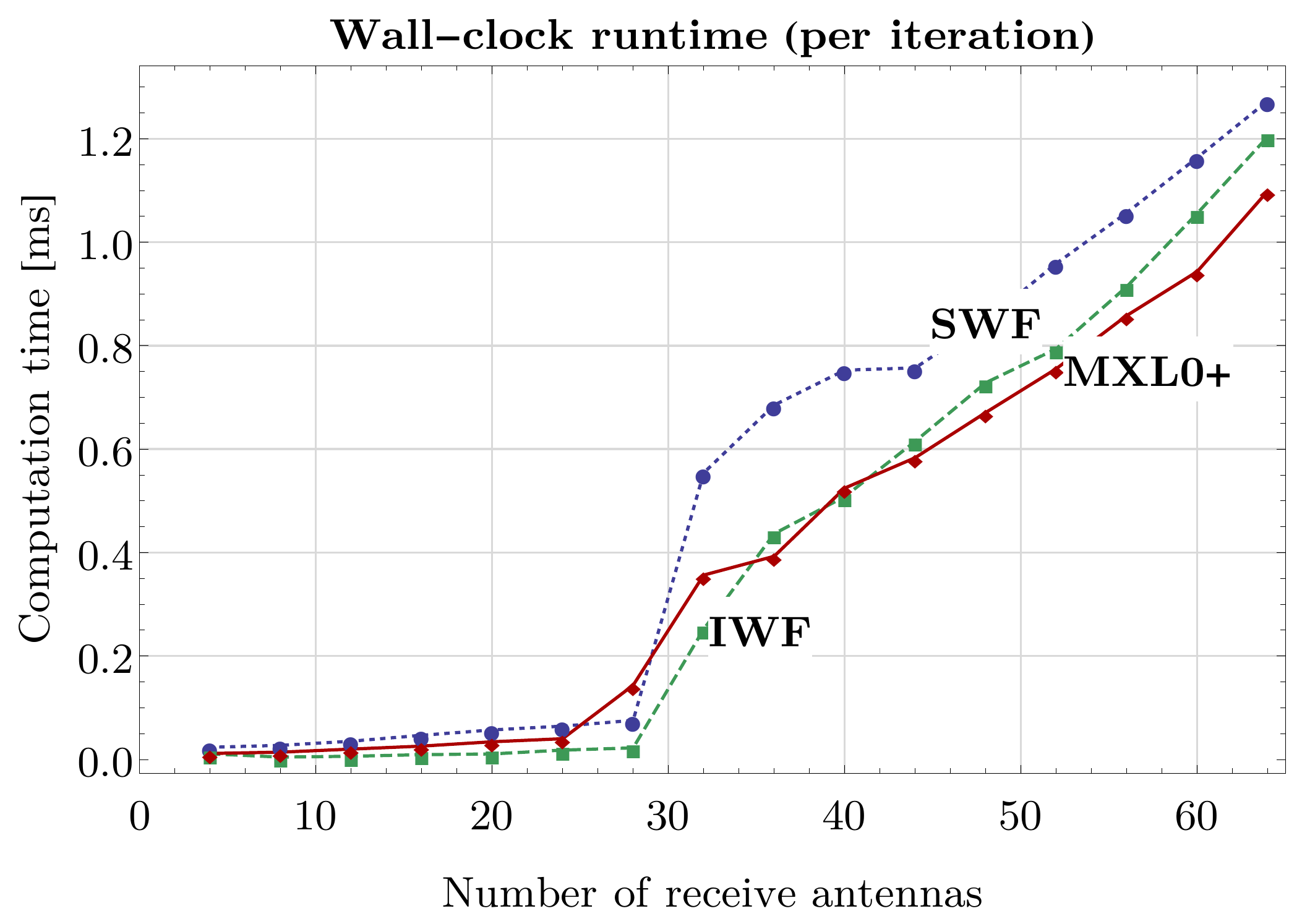}
\caption{Per-iteration runtime of the \ac{IWF}, \ac{SWF} and \ac{MXL+} algorithms as a function of the number of antennas at the receiver (lower is better).}
\vspace{-2ex}
\label{fig:runtime}
\end{figure}


\endedit

\subsection{Convergence speed analysis}

For completeness, we also examine below the convergence speed of the different \ac{MXL} methods with limited, zeroth-order feedback.
The results of our experiments are reported in \cref{fig:loglog} where we plot the users' relative distance to optimality in a log-log scale under
the three gradient-free algorithms discussed in the previous sections, \ac{MXL0}, \ac{MXL+} and \ac{AMXL+} (\cref{alg:trMXL,alg:trMXLMO,alg:randomtrMXLMO} respectively, the third in the coordinate descent form \ac{CD} discussed in \cref{remark:CD}).
We plot the relative ratio $\rho = (\ratefunction^{\ast} - \ratefunction(\meanmatn{\itr}))/(\ratefunction^{\ast} - \ratefunction_{1})$, so $\rho=1$ corresponds to the initialization of each algorithm while $\rho=0$ corresponds to optimality.
All algorithms were run with constant step size and query radius in a system with $\kk=\KK$ users.
Despite the severe feedback limitations, we see that \anyMXLzeroplus{}
rapidly closes the initial optimality gap (in line with \cref{fig:MXL-WF-zeroth}).

A close inspection of the slopes of the various curves on the log-log graph further reveals the $\magnitude{1/\sqrt[4]{t}}$ complexity of \ac{MXL0} and the $\magnitude{1/\sqrt{t}}$ complexity of \anyMXLzeroplus{},
in full accordance with \cref{thm:trMXLiinew,timecomplexity:trMXLMO,timecomplexity:randomtrMXLMO}.
The $\log\kk$ shift between \ac{CD} and \ac{MXL+} predicted in \cref{remark:CD} can also be clearly observed.

\begin{figure}[t] 
\centering
\includegraphics%
[width=\columnwidth]%
{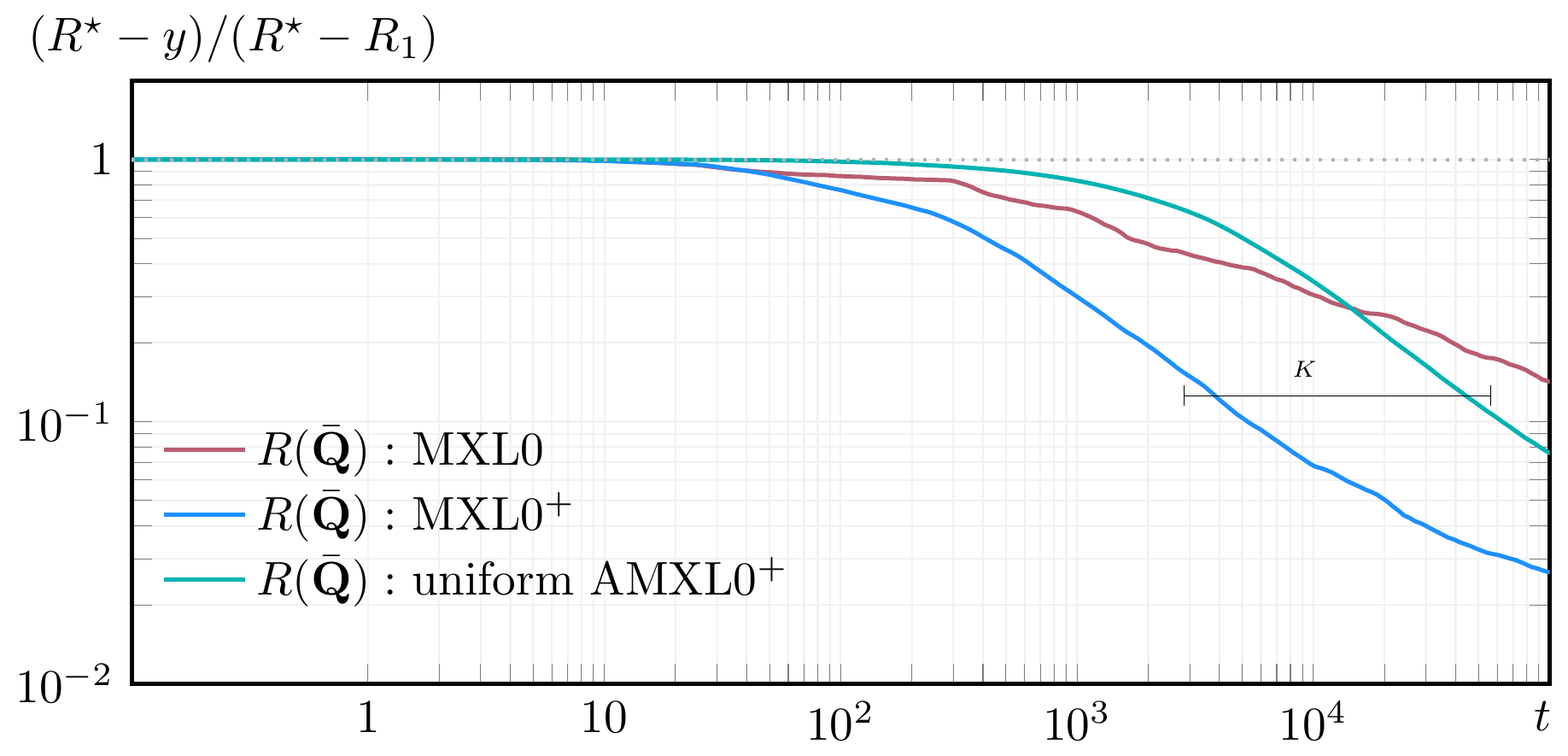}
\caption{
Convergence speed of the proposed methods ($\nn=\NN$, $\kk=\KK$).
The callback in \ac{MXL+} greatly improves performance over \ac{MXL0}.
}
\label{fig:loglog}
\vspace{-2ex}
\end{figure}%
%
%

\renewcommand\KK{50}%
\renewcommand\NN{128}%
\renewcommand\MM{3}%
\begin{figure}[t] 
\centering
\includegraphics%
[height=4.3cm]%
{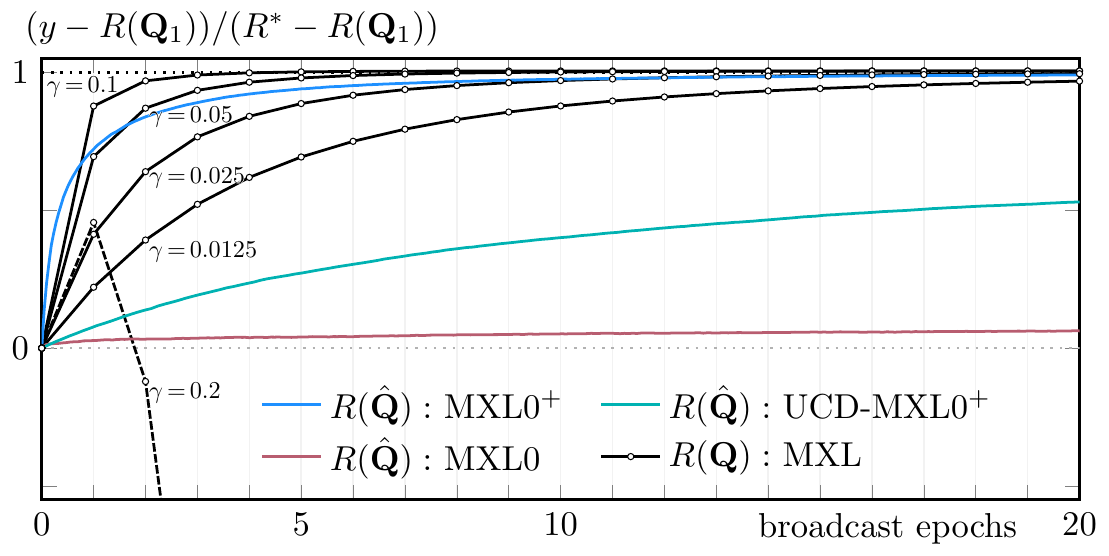}
\caption{
Overhead of the algorithms under study in a large network ($\kk=\KK$, $\nn=128$).
When normalized for overhead, \ac{MXL+} matches
the performance of finely tuned gradient-based methods.}
\label{fig:large}
\vspace{-2ex}
\end{figure}%

Finally, \cref{fig:large} provides a normalized comparison to gradient-based methods
in a network with $\nn=\NN$ receive antennas and $\kk=\KK$ users.
Here, access to full matrix feedback would require $\nn^{2} = 32\,\mathrm{MB}$ of $16$-bit data per frame;
in view of this, we examine instead the algorithms' convergence speed in terms of the feedback epochs required for convergence.
For benchmarking purposes,
we ran \ac{MXL} with a constant step-size (the most principled choice given the smoothness of $\ratefunction$).
Quite remarkably, we see that \ac{MXL+}
remains competitive with\textemdash and even outperforms!\textemdash the fastest implementations of \ac{MXL}.
On the other hand, \ac{CD}  was approximately $\kk=\KK$ times slower than \ac{MXL+}, while \ac{MXL0} was essentially non-convergent.

\section{Discussion}
\label{section:discussion}

In this paper, we proposed a series of online optimization schemes for \emph{distributed, feedback-limited} multi-user \ac{MIMO} systems that circumvent the need for matrix feedback (perfect, noisy, or otherwise).
Gradient estimation methods based on conventional \acf{SPSA} techniques lead to an $\magnitude{1/T^{1/4}}$ convergence rate, which is catastrophically slow for large \ac{MIMO} systems.
To overcome this deficiency, we introduced an acceleration mechanism which achieves an $\magnitude{1/T^{1/2}}$ convergence rate through the reuse of previous throughput measurements.
In this regard, the proposed \acs{MXL+} algorithm enjoys the best of many worlds:
it achieves convergence with minimal feedback requirements (a single scalar), it matches the convergence speed of conventional methods that require full mtrix feedback,
all the while remaining simple in principle and easy to implement.

Although we focused on the throughput maximization problem in the single-cell MIMO multiple-access channel, our proposed algorithms can also be applied to multi-cell networks operating in orthogonal frequency bands so that the inter-cell interference is canceled;
the sum rate in each cell can be optimized separately and independently without any loss of global optimality.
In dense small-cell networks, in which the interference cannot be canceled this way, the network sum-rate optimization problem is a known difficult non-convex problem \cite{CADC08,SRLH11}. A possible workaround is to consider autonomous small-cells that aim at maximizing their own sum rate (similar in spirit to \eqref{eq:trproblemk} in \cref{subsec:SUD}), which leads to a concave non-cooperative game.
In our previous work \cite{MBNS17}, we showed that the original \ac{MXL} converges to the Nash equilibrium solution of such games under milder assumptions compared to iterative water-filling;
studying the performance of our gradient-free algorithms \ac{MXL0} and \ac{MXL+} in such settings is an interesting and non-trivial extension of the present work.

Moving beyond throughput maximization, the gradient-free methodology presented in this work can also be tailored to a wide range of resource allocation problems that arise in signal processing and wireless communications (from power control to energy efficiency).
For example, by using the Charnes-Cooper transformation to turn non-convex fractional optimization problems into convex ones \cite{MB16}, the material developed in this paper can be applied to the core problem of energy-efficiency maximization problem in multi-user \ac{MIMO} systems.
These applications, which are deferred to future work, highlight the potential of the gradient-free algorithms derived here.

\beginrev
Finally, in terms of practical implementation, we should note that our analysis provides precise computational complexity and runtime bounds;
however, it does not address the processing power expenditure on ``off-the-shelf'' wireless devices.
Investigating this aspect of the proposed methods is a very fruitful research direction which we intend to address in future work.
\endedit


\numberwithin{equation}{subsection}
\numberwithin{lemma}{subsection}
\numberwithin{table}{subsection}

\appendix[Technical Proofs]

\subsection{Matrix exponential learning as a dual averaging scheme}
\label{appendix:DA}

In our developments, the space of the covariance matrices of each user is equipped with the nuclear norm, given for any Hermitian matrix $\mat$ by 
$ 
\norm{\mat} = \trace{\sqrt{\mat\mat}} 
$, 
and equivalent to the $L_1$-norm of the vector of the eigenvalues of $\mat$. 
%
%
The dual of the nuclear norm, 
$
\dualnorm{\mat}= \max_{\altmat} \{ \trace{\mat\altmat} \setst \norm{\altmat}\leq 1 \}  
$, reduces to the $L_\infty$-norm of the vector of eigenvalues.
For every $m\times m$ Hermitian matrix $\mat$, one has
\begin{equation}\label{normsMO}
 \dualnorm{\mat} 
\leq 
\frobeniusnorm{\mat} 
\leq 
\norm{\mat}  
\leq
\sqrt{m} \frobeniusnorm{\mat} 
 \leq 
 m  \dualnorm{\mat} , 
\end{equation} 
where $\frobeniusnorm{\mat}=\sqrt{\trace{\mat\mat}}$ 
denotes the (Frobenius) $L^2$-norm of $\mat$.
From the  global perspective of matrix arrangements $\mat=\vect{\matk{1},...,\matk{\kk}}$---now regarded as block diagonal covariance matrices---, the trace norm and its dual naturally extend as
\begin{equation} \textstyle
\label{globalnorm}
\norm{\mat} = \sum\nolimits_{\nbk=1}^{\kk} \norm{\matk{\nbk}} ,
\qquad 
\dualnorm{\mat} = \max\limits_{\nbk\in\{1,\dots,\kk\}} \dualnorm{\matk{\nbk}}.
\end{equation}
%


We now derive the matrix exponential learning step and some properties of it.
To this end, we place ourselves in the compact set $\compactset=\{\mat\in\Herm{\mm}:\trace\mat=1, \mat\succeq 0\}$ of the $\mm$-dimensional positive semidefinite Hermitian matrices with unit trace---the parameter $\mm$ stands for the number of antennas of any of the $\kk$ users. 
 Let the inner product $\inner{\dualmat}{\mat}=\trace{\dualmat\mat}$ denote the value at $\mat\in\trMat$ of the linear function induced by $\dualmat\in\trZmat$, where $\trZmat=\{\trzmat\in\Herm{\mm}\setst\trace{\trzmat}=0\}$ 
is tangent to~$\compactset$. 
For any differentiable function~$\genericfunction$ on~$\Herm{\mm}$, we denote by $\projgrad\genericfunction:\compactset \mapsto\trZmat$ the \emph{orthogonal projection} of the gradient~$\grad\genericfunction$ on the tangent space~$\trZmat$, given  by 
 $\projgrad \genericfunction = \grad \genericfunction - \trace{\grad \genericfunction}\, \identity$.  

\begin{lemma} \label{lemma:reg}
\

\begin{enumerate}[(i),ref=\roman*,wide,labelwidth=!,labelindent=5pt]
\item \label{lemma:regi}
The regularization function\footnote{We use here the convention $0\log0=0$.}
    $
\reg{\mat}=\trace{\mat \log\mat}
    $ 
is $\forcestrongly$-strongly convex over $\compactset$ with respect to $\norm\cdot$.
\item \label{lemma:regii}
The conjugate of $\regfunction$,  $\conjregfunction : \trZmat\mapsto\Real$, defined by
\begin{equation}\label{conjregfunction}
\textstyle
\conjreg{\dualmat} =\max_{\mat\in \compactset} \{ \inner{\dualmat}{\mat} - \reg{\mat}\},
\end{equation}  
 is differentiable with gradient $\grad\conjregfunction = \matsolfunction$, where $\matsolfunction$ is the exponential learning mapping defined by
\begin{equation} \label{trmatsol}
\matsol{\dualmat} =\frac{ \exponential{ \dualmat } }{\trace{\exponential{ \dualmat } }} 
.
\end{equation}
%
\item \label{lemma:regiii}
For $\mat\in\compactset$ and $\dualmat\in\trZmat$, 
\begin{equation} \label{reciprocity}
\mat = \matsol{\dualmat}  \Leftrightarrow 
\dualmat = \projgrad\reg{\mat} .
\end{equation}
\item \label{lemma:regiv}
$\conjregfunction$ is $\forcestrongly$-smooth with respect to the dual norm $\dualnorm\cdot$.
%
\end{enumerate}
\end{lemma} 
\begin{proof}
We refer to \cite{yu13} for the strong convexity of $\regfunction$.
For \eqref{lemma:regii}, the differentiablity of $\conjregfunction$ is a consequence of Danskin's theorem (e.g. \cite{shapiro09}),
which, besides, gives us the gradient of \eqref{conjregfunction},
\begin{equation} \label{beforetrmatsol}
\textstyle
\grad\conjreg{\dualmat}
= \arg\max_{\mat\in \compactset} \{ \inner{\dualmat}{\mat} - \reg{\mat}\} 
.
\end{equation} 
Relaxing the constraint $ \trace{\mat}  - 1 = 0$ in  the subproblem \eqref{beforetrmatsol} and using $\grad\reg{\mat}=\identity+\log\mat$ 
yields the stationarity condition  
\begin{equation}
 \log\mat - \dualmat + (1+\multiplier)\, \identity = 0  
\tag{S} \label{KKTS}
,
\end{equation}
where $\multiplier\in\Real$ is the Lagrange multiplier related to  the constraint. Condition \eqref{KKTS} rewrites as
$
\mat 
= \exponential{- (1+\multiplier)} \, \exponential{ \dualmat } 
$,
which implies the primal feasibility condition $\mat\succeq 0$. The remaining KKT conditions $\trace{\mat}  - 1 \leq 0$ and $\multiplier \, (\trace{\mat}  - 1 ) = 0 $ 
 yield $\multiplier = \logarithm{\trace{\exponential{ \dualmat } } }-1$, and $\mat=\matsol{\dualmat}$ as the unique maximizer of \eqref{beforetrmatsol}, which completes the proof of \eqref{lemma:regii}.
%

Now, it follows from \eqref{beforetrmatsol} that, for any $\dualmat\in\trZmatk{\nbk}$, one has $\mat = \matsol{\dualmat}$ if and only if $\inner{\dualmat}{\altmat} - \reg{\altmat} \leq \inner{\dualmat}{\mat} - \reg{\mat}$ holds for all $\altmat \in \trMat$, i.e., iff $\dualmat$ is a subgradient of $\regfunction$ at $\mat$. Claim \eqref{lemma:regiii} follows by differentiability of $\regfunction$.
\obsolete{
Let $\identity$ denote the identity matrix in $\Herm{\mm}$. By relaxing the constraint $ \trace{\mat}  - 1 = 0$ in  the subproblem \eqref{beforetrmatsol} and using $\grad\reg{\mat}=\identity+\log\mat$, we end up with  the KKT conditions   
\begin{align}
\log\mat - \dualmat + (1+\multiplier)\identity = 0 
\tag{S} \label{KKTS}
\\
\mat\succeq 0
\tag{PF1} \label{KKTPF1}
\\
\trace{\mat}  - 1 \leq 0
\tag{PF2} \label{KKTPF2}
\\
\multiplier \in\Real
\tag{DF} \label{KKTDF}
\\
\multiplier \, (\trace{\mat}  - 1 ) = 0
\tag{CS} \label{KKTCS}
\end{align}
where $\multiplier$ denotes the Lagrange multiplier related to  the constraint.
Equation \eqref{KKTS} rewrites as
\obsolete{
\begin{equation} \label{MXL1}
\mat = \exponential{ \dualmat - (1+\multiplier)\identity } = \exponential{- (1+\multiplier)} \, \exponential{ \dualmat } ,
\end{equation}
}
$
\mat 
= \exponential{- (1+\multiplier)} \, \exponential{ \dualmat } 
$,
which implies \eqref{KKTPF1}. The remaining conditions \eqref{KKTPF2}, \eqref{KKTDF} and \eqref{KKTCS} yield $\multiplier = \logarithm{\trace{\exponential{ \dualmat } } }-1$, and $\mat=\matsol{\dualmat}$ as the unique maximizer of \eqref{beforetrmatsol}.
%
Now, if $\dualmat\in\trZmatk{\nbk}$,
\begin{align}
\notag &
\mat = \matsol{\dualmat}&
\\
\notag \Leftrightarrow \quad&
\inner{\dualmat}{\altmat} - \reg{\altmat} \leq \inner{\dualmat}{\mat} - \reg{\mat}
,& \forall \altmat \in \trMat
\obsolete{
\\
\label{LemmaB1-} \Leftrightarrow \quad&
\inner{\dualmat}{\altmat-\mat} \leq  \reg{\altmat}  - \reg{\mat}
,& \forall \altmat \in \trMat
}
\\
\label{LemmaB1} \Leftrightarrow \quad&
\dualmat \in \subdiff{\regfunction}{\mat} .
&
\end{align}
and \eqref{lemma:regiii} follows by differentiability of $\regfunction$.
}%
Finally, \eqref{lemma:regiv} is a property of convex conjugation \cite{kakade12}.
Indeed, let $\dualmat,\altdualmat\in\trZmat$ and $\mat=\matsol{\dualmat}$. 
By convexity, 
\begin{equation} \label{strongconvexity}
\textstyle
\reg{\altmat} \geq \reg{\mat}+\inner{\projgrad\reg{\mat} }{\altmat-\mat} + \frac \forcestrongly 2 \norm{\altmat-\mat}^2
\end{equation}
holds for any $\altmat\in\compactset$. It follows that
\begin{equation} \label{stronglysmooth}
\nocolsep
\begin{array}{rl}
\conjreg{\altdualmat} \hspace{-9mm}
&\hspace{9mm}\refereq{\eqref{conjregfunction}}{=} 
\max_{\altmat\in \compactset} \{ \inner{\altdualmat}{\altmat} - \reg{\altmat}\}
\\
&{\refereq{{\eqref{strongconvexity}}}{\leq}}
\max_{\altmat\in \compactset} \{ \inner{\altdualmat}{\altmat} {\,-\,}  \reg{\mat}{\,-\,}\inner{\projgrad\reg{\mat} }{\altmat{\,-\,}\mat} 
{\,-\,} \frac \forcestrongly 2 \norm{\altmat{\,-\,}\mat}^2 \}
\\
&{\refereq{{\eqref{reciprocity}}}{=}} 
\inner{\dualmat }{\mat}  -  \reg{\mat}  + \inner{\altdualmat -\dualmat }{\mat}
 + \max_{\altmat\in \compactset} \{ \inner{\altdualmat -\dualmat }{\altmat{\,-\,}\mat} 
 \ \\&\hfill
- \frac \forcestrongly 2 \norm{\altmat-\mat}^2 \} 
\\
&{\refereq{{\eqref{beforetrmatsol}}}{\leq}} 
\conjreg{\dualmat}+ \inner{\altdualmat -\dualmat }{\grad\conjreg{\dualmat}} + \frac{1}{2\strongly} \dualnorm{\altdualmat-\dualmat}^2
\end{array}
\end{equation} 
and $\conjregfunction$ is $\forcestrongly$-smooth. Equivently, \eqref{stronglysmooth} rewrites as\cite{nesterov14}
\begin{equation} \label{DAlipschitz}
 \norm{\matsol{\dualmat}-\matsol{\altdualmat}} \leq \oneoverstrongly \dualnorm{\dualmat-\altdualmat}
\quad 
\forall \dualmat,\altdualmat\in\trZmatk{\nbk}
.
\end{equation}
\end{proof}

We now consider the Fenchel primal-dual coupling $\fenchelcoupling:\compactset\times\trZmat\mapsto\Real$ associated with the entropic regularizer $\regfunction$.
\begin{lemma} \label{lemma:fenchel}
The Fenchel coupling
\begin{equation}\label{fenchelkMO}
\fenchel{\mat}{\dualmat} = \reg{\mat} + \conjreg{\dualmat} - \inner{\dualmat}{\mat}
\end{equation}
satisfies the following properties. \begin{subequations}%
For $\mat\in\compactset$ and $\dualmat,\altdualmat\in\trZmat$,
\begin{align}
\label{fenchelsmooth} 
  \textstyle
\hspace{-2mm}
\fenchel{\mat}{\altdualmat} 
&\leq \textstyle
\fenchel{\mat}{\dualmat }  {\,+\,}  \inner{  \altdualmat{-}\dualmat}{\matsol{\dualmat}{-}\mat} 
{\,+\,} \frac{1}{2\strongly} \dualnorm{\altdualmat{-}\dualmat }^2 
,  
 \\  \textstyle
\label{fenchelstronglyconvex}
\hspace{-6mm}
\fenchel{\mat}{\dualmat} 
&\geq \textstyle
\frac \forcestrongly 2    \norm{\mat-\matsol{\dualmat}}^2
,
 \\  \textstyle\label{fenchelnonnegativity}
\hspace{-6mm}
\fenchel{\mat}{\dualmat} 
&\geq  0 
\text{ with }
\fenchel{\mat}{\dualmat} = 0 
  \Leftrightarrow
 \dualmat =\projgrad\reg{\mat}
. 
\end{align}
\end{subequations}
\end{lemma}
\begin{proof}
Equations \eqref{fenchelsmooth} and \eqref{fenchelstronglyconvex} follow from the smoothness of $\conjregfunction$ and from the strong convexity of $\regfunction$, respectively. Indeed, we get \eqref{fenchelsmooth} by combining \eqref{fenchelkMO} with \eqref{stronglysmooth}, 
while
\begin{equation} \label{derivationfenchelstronglyconvex}
\nocolsep
\begin{array}{rcl}
\fenchel{\mat}{\dualmat} 
&\refereq{\eqref{fenchelkMO}}{=}&
 \max_{\altmat\in \compactset} \{   \reg{\mat} - \reg{\altmat} - \inner{\dualmat}{\mat-\altmat} \} 
\\ 
&\refereq{\eqref{reciprocity}}{\geq}&
  \reg{\mat} - \reg{\matsol{\dualmat}} - \inner{\projgrad\reg{\matsol{\dualmat}}}{\mat-\matsol{\dualmat}} 
\\
&\refereq{\eqref{strongconvexity}}{\geq}&
\frac \forcestrongly 2    \norm{\mat-\matsol{\dualmat}}^2
\end{array}
\end{equation}
yields \eqref{fenchelstronglyconvex}.
Then, \eqref{fenchelnonnegativity} follows from \eqref{fenchelstronglyconvex} and \eqref{reciprocity}.
\end{proof}

\renewcommand{\strongly}{}%

\subsection{The \ac{SPSA} estimator}
\label{appendix:SPSA}

This section is concerned with the bias of the gradient estimator defined, for $\nbk=1,\dots,\kk$, by 
\begin{equation}\label{estgradratekO}\begin{array}{c}
\estgradratek{\nbk}{\mat}{\trzmat}{\offset}
=
\frac{\dimensionk{\nbk}}{\delta} \big[ \rate{ \testmat} - \offset \big] \, \trzmatk{\nbk} ,
\end{array}
\end{equation}
where $\delta>0$ is a given query radius, $\testmat = \vect{\testmatk{1},\dots,\testmatk{\kk}}$ is given by \eqref{linearperturbator}, $\trzmat=\vect{\trzmatk{1},\dots,\trzmatk{\kk}}$ with $\trzmatk{\nbk}$ sampled uniformly on the sphere $\trmspherei{\dimensionk{\nbk}-1}$, and $\offset$ an arbitrary scalar offset quantity independent of $\trzmat$. Observe that \eqref{estgradratekO} covers the gradient estimators of both \ac{MXL0} and \ac{MXL+}. 

The computation of a  bound for the bias of estimator \eqref{estgradratekO} is based on Stokes' theorem, applied to the sphere $\trmspherei{\dimensionk{\nbk}-1}  $:
\begin{equation}\label{stokes}
\int_{\trmspherei{\dimensionk{\nbk}-1}} \generic{ 
\trzmatk{\nbk} } \, \trzmatk{\nbk} \,  d\lebesgue{\trzmatk{\nbk}}
= 
\int_{\trmballi{\dimensionk{\nbk}}} \projgrad\generic{ 
 \momega} \,  d\lebesgue{\momega}
,
\end{equation}
where $\genericfunction$ is any function on $\Herm{\mmk{\nbk}}$ and $\lebesguemeasure$ denotes the Lebesgue measure.
Before proceeding, observe
%
that each
test covariance marix $\testmatk{\nbk}$  is bound to the initial matrix $\matk{\nbk}$ by
$
\frobeniusnorm{\testmatk{\nbk} - \matk{\nbk} }\leq \twitchFlipschitz \delta \frobeniusnorm{\trzmatk{\nbk}} 
$,
where, under our assumption $\dimensionk{\nbk}>0$,
$
\dualnorm{ \trzmatk{\nbk}}\leq {1}/{2}
$ 
for every
$
\trzmatk{\nbk}\in\trmspherei{\dimensionk{\nbk}-1}
$.
%
It follows from \eqref{normsMO} that  any test configuration $\testmat$ in \eqref{estgradratek} and \eqref{estgradratekMO} satisfies
\begin{equation} \label{perturbatorbound}
\frobeniusnorm{\testmat - \mat }
\leq 
\twitchFlipschitz\delta \kk,
\quad \text{and} \quad
\norm{\testmat  - \mat }
\leq 
\twitchFlipschitz\delta \kk \sqrt{\maxdimension} 
.
\end{equation}

\begin{lemma}\label{lemma:bias}
The estimator \eqref{estgradratekO} satisfies 
\begin{align}\textstyle &
\label{lipschitzconstantMO}
\dualnorm{ \expectation{\estgradratek{\nbk}{\mat}{\trzmat}{\offset} -\gradk{\nbk} \rate{\mat}   } }  
\leq
{2\kk\meanlipschitzk{\nbk}}\, \delta
,
\\&
\label{trboundgradientestimateMO}
\textstyle
\dualnorm{\estgradratek{\nbk}{\mat}{\trzmat}{\offset}}
\leq
\frac{  \dimensionk{\nbk}}{2^{\kk}\delta}  \max_{\altmat \in\compactset} \modulus{ \rate{ \altmat } - \offset }
.
\end{align}
%
\end{lemma}
\begin{proof}[Proof of \cref{lemma:bias}]
We argue as in \cite{FKM05,BLM18}.
\obsolete{
First observe that
\begin{equation}\label{integrationresult}
\begin{array}{r}
\frac{ 
\int_{\trmballi{\dimensionk{\nbk}}}
 \frobeniusnorm{  \momega} \, d\lebesgue{\momega}
}{\int_{\trmballi{\dimensionk{\nbk}}}
 \, d\lebesgue{\momega}} 
=
\frac{ 
\int_0^1 d\varrho
\int_{\varrho\trmspherei{\dimensionk{\nbk}-1}}
 \varrho \, d\lebesgue{\momega}
}{
\int_0^1 d\varrho
\int_{\varrho\trmspherei{\dimensionk{\nbk}-1}}
\, d\lebesgue{\momega}
} 
=
\frac{ 
\int_0^1 
\varrho\, \Volume{\varrho\trmspherei{\dimensionk{\nbk}-1}}
 d\varrho
}{
\int_0^1 \Volume{\varrho\trmspherei{\dimensionk{\nbk}-1}} d\varrho
} 
\quad \\
=
\frac{\dimensionk{\nbk}}{\dimensionk{\nbk}+1}
.
\end{array}
\end{equation}
}%
By introducing the notation
$ \contrmatkdeltaZ{\nbk}{\delta}{\momega} = \vect{\testmatk{1}, \dots, \testmatk{\nbk-1}, \contrmatkdeltaZk{\nbk}{\delta}{\momega}{\nbk}  ,\testmatk{\nbk+1}, \dots, \testmatk{\kk} }$,
in which $ \contrmatkdeltaZk{\nbk}{\delta}{\momega}{\nbk}  = \matk{\nbk} + ({\delta}/{\radiusk{\nbk}}) (\cmatk{\nbk} - \matk{\nbk})  + \delta \momega$,
%
we find
\begin{equation}\label{earlylipschitzconstantMO}
\begin{array}{l}
\dualnorm{ \expectation{\estgradratek{\nbk}{\mat}{\trzmat}{\offset} -\projgradk{\nbk} \rate{\mat}   } }  
\\
\qquad
\qquad
\refereq{\eqref{estgradratekO}}{ =}
\dualNorm{ \expectation{\frac{\dimensionk{\nbk}}{\delta} [ \rate{ \testmat} - \offset ] \, \trzmatk{\nbk}
-\projgradk{\nbk} \rate{\mat}   } }  
\\
\qquad
\qquad
=
\dualNorm{ \frac{\dimensionk{\nbk}}{\delta}\expectation{  \rate{ \testmat} \, \trzmatk{\nbk}
   }  -\projgradk{\nbk} \rate{\mat} }  
\\
\qquad
\qquad
=  
\dualNorm{
\frac{\dimensionk{\nbk}}{\delta} \Expectation{ \frac{ \int_{\trmspherei{\dimensionk{\nbk}-1}}   \rate{\contrmatkdeltaZ{\nbk}{\delta}{\trzmatk{\nbk}}} \, \trzmatk{\nbk}
  \,  d\lebesgue{\trzmatk{\nbk}} 
}{\volume{\trmspherei{\dimensionk{\nbk}-1}}}}
-\projgradk{\nbk} \rate{\mat}
}
\\
\qquad
\qquad
\refereq{\eqref{linearperturbator}}{=}
\dualNorm{
\Expectation{
\frac{ \int_{\trmspherei{\dimensionk{\nbk}-1}} \rate{\contrmatkdeltaZ{\nbk}{\delta}{\trzmatk{\nbk}}} \, \trzmatk{\nbk}
 \,  d\lebesgue{\trzmatk{\nbk}} 
}{\delta\volume{\trmballi{\dimensionk{\nbk}}}}
-\projgradk{\nbk} \rate{\mat} }
}
.
\end{array}\end{equation}
\obsolete{
where 
$ \contrmatkdeltaZ{\nbk}{\delta}{\momega} = \vect{\testmatk{1}, \dots, \testmatk{\nbk-1}, \contrmatkdeltaZk{\nbk}{\delta}{\momega}{\nbk}  ,\testmatk{\nbk+1}, \dots, \testmatk{\kk} }$,
with $ \contrmatkdeltaZk{\nbk}{\delta}{\momega}{\nbk}  = \matk{\nbk} + {\delta} (\cmatk{\nbk} - \matk{\nbk}) /{\radiusk{\nbk}} + \delta \momega. $
}%
%
It follows from Stokes' theorem that \eqref{earlylipschitzconstantMO} reduces to
\begin{equation}\label{computationlipschitzconstantMO}
\begin{array}{l}
\dualnorm{ \expectation{\estgradratek{\nbk}{\mat}{\trzmat}{\offset} -\projgradk{\nbk} \rate{\mat}   } }  
\\ \quad
\refereq{\eqref{stokes}}{=}
\dualNorm{
\Expectation{
\frac{ \int_{\trmballi{\dimensionk{\nbk}}} \delta \projgradk{\nbk}\rate{\contrmatkdeltaZ{\nbk}{\delta}{\momega}} 
\, d\lebesgue{\momega}}{\delta\volume{\trmballi{\dimensionk{\nbk}}}}
-\projgradk{\nbk} \rate{\mat}}
}
\\ \quad
\obsolete{
=
\dualNorm{
\Expectation{
\frac{ \int_{\trmballi{\dimensionk{\nbk}}} \left[ \projgradk{\nbk}\rate{\contrmatkdeltaZ{\nbk}{\delta}{\momega}} 
-\projgradk{\nbk} \rate{\mat}\right]\, d\lebesgue{\momega}}{\volume{\trmballi{\dimensionk{\nbk}}}}
}
}
\\ \quad
}
\leq  
\bigexpectation{
\frac{1}{\volume{\trmballi{\dimensionk{\nbk}}}} \int_{\trmballi{\dimensionk{\nbk}}} \dualnorm{ \projgradk{\nbk}\rate{\contrmatkdeltaZ{\nbk}{\delta}{\momega}} 
-\projgradk{\nbk} \rate{\mat} } \, d\lebesgue{\momega}
}
\\ \quad
\obsolete{
\refereq{\eqref{smoothness}}{\leq}
\frac{ 
\gradlipschitz\int_{\trmballi{\dimensionk{\nbk}}}
 \frobeniusNorm{ \twitchk{\nbk}{\matk{\nbk}}{\delta\momega} - \matk{\nbk} } \, d\lebesgue{\momega}
}{\Volume{\trmballi{\dimensionk{\nbk}}}}
+
\gradlipschitz \Expectation{
\sum_{\altnbk\neq\nbk}    \frobeniusnorm{\testmatk{\altnbk}-\matk{\altnbk}}
}
\\  \quad
\refereq{\eqref{linearperturbator}}{\leq}
\gradlipschitz\left(1+
\frac{ 
\int_{\trmballi{\dimensionk{\nbk}}}
 \frobeniusnorm{  \momega} \, d\lebesgue{\momega}
}{\Volume{\trmballi{\dimensionk{\nbk}}}} 
\right) \delta
+
2 \gradlipschitz (\kk-1)  \,  \delta
\\  \quad
=
\gradlipschitz\big(1+ \frac{\mmk{\nbk}^2-1}{\mmk{\nbk}^2} \big) \delta
+
2 \gradlipschitz (\kk-1)  \,  \delta
\\ \quad
\leq
2\kk\gradlipschitz\,\delta
}%
\refereq{\eqref{lipschitzkl}}{\leq}
\frac{ 
\int_{\trmballi{\dimensionk{\nbk}}}\lipschitzkl{\nbk}{\nbk}
 \, \frobeniusnorm{ \contrmatkdeltaZk{\nbk}{\delta}{\momega}{\nbk} - \matk{\nbk} } \, d\lebesgue{\momega}
}{\volume{\trmballi{\dimensionk{\nbk}}}}
+
\sum_{\altnbk\neq\nbk}  \lipschitzkl{\nbk}{\altnbk} \expectation{ \frobeniusnorm{\testmatk{\altnbk}-\matk{\altnbk}}
}
\\  \quad
\refereq{\eqref{linearperturbator}}{\leq}
\lipschitzkl{\nbk}{\nbk}\bigg(1+
\frac{ 
\int_{\trmballi{\dimensionk{\nbk}}}
 \frobeniusnorm{  \momega} \, d\lebesgue{\momega}
}{\volume{\trmballi{\dimensionk{\nbk}}}}
\bigg) \delta
+
2\sum_{\altnbk\neq\nbk} \lipschitzkl{\nbk}{\altnbk} \Expectation{  \frobeniusnorm{\trzmatk{\altnbk}}} \delta
\\ \quad
\obsolete{
\refereq{\eqref{integrationresult}}{=}
\lipschitzkl{\nbk}{\nbk}
\left(\frac{2 \mmk{\nbk}+1}{\mmk{\nbk}+1}\right)  
\delta 
+
\sum_{\altnbk\neq\nbk}  \lipschitzkl{\nbk}{\altnbk}\Expectation{  \frobeniusnorm{\testmatk{\altnbk}-\matk{\altnbk}}}
\\
\leq   
\lipschitzkl{\nbk}{\nbk}
\left(\frac{2 \dimensionk{\nbk}}{\mmk{\nbk}^2}\right)  
\delta
+
\sum_{\altnbk\neq\nbk}  \lipschitzkl{\nbk}{\altnbk}\Expectation{  \frobeniusnorm{\twitchk{\altnbk}{\delta\matk{\altnbk}}{\trzmatk{\altnbk}} 
-\matk{\altnbk}}}
\\
\refereq{\eqref{perturbatorlipschitz}}{\leq}
}
\leq
\lipschitzkl{\nbk}{\nbk}\big(\frac{2 \mmk{\nbk}+1}{\mmk{\nbk}+1}\big)  \delta + 2 \sum_{\altnbk\neq\nbk}  \lipschitzkl{\nbk}{\altnbk}  \delta
\,\leq\,
{2\kk\meanlipschitzk{\nbk}}\delta
,  
\end{array}
\end{equation}
and we recover \eqref{lipschitzconstantMO}.
Then \eqref{trboundgradientestimateMO} is immediate from the definition of $\estgradratek{\nbk}{\mat}{\trzmat}{\offset}$ and the fact that
$
\dualnorm{ \trzmatk{\altnbk}}\leq {1}/{2}
$ 
 for all $\altnbk$. 
\end{proof}

\subsection{Analysis of the \ac{MXL0} algorithm}
\label{appendix:MXL0}


Let $\optMat$ denote the solution set of \eqref{eq:trproblem}. Given any~$\optmat\in\optMat$, we consider, for analysis purposes, the Lyapunov function
\begin{equation} \label{lyapunov}
\textstyle
\lyapunov{\optmat}{\dualmat} = \sum_{\nbk=1}^{\kk}\fenchel{\optmatk{\nbk}}{\dualmatk{\nbk} } ,
\end{equation}
where 
$\fenchelcoupling$ is the Fenchel coupling defined in
~\eqref{fenchelkMO}.
If 
$
\filtrationn{\itr-1} 
=
( \dualmatn{1},\zmatn{1},\dots,\dualmatn{\itr-1},\zmatn{\itr-1} )
$
denotes  the history of \ac{MXL0} up to step $\itr-1$,
the gradient estimator \eqref{estgradratekO} decomposes into
\begin{equation}\label{estimatordecomposition}
\estgradkn{\nbk}{\itr}= \projgradk{\nbk}\rate{\matn{\itr}} + \biaskn{\nbk}{\itr} + \devmatkn{\nbk}{\itr} ,
\end{equation}
where $\biaskn{\nbk}{\itr}  =  \expectation{  \estgradkn{\nbk}{\itr} \cond \filtrationn{\itr-1}   } - \projgradk{\nbk}\rate{\matn{\itr}}$ is the systematic error on $\estgradkn{\nbk}{\itr}$, 
bounded by
\begin{equation} \label{biasbound}
\dualnorm{ \biaskn{\nbk}{\itr} }
\refereq{\eqref{lipschitzconstantMO}}{\leq} 
{2\kk\meanlipschitzk{\nbk}}\deltan{\itr} 
,
\end{equation}
 and $\devmatkn{\nbk}{\itr}$ is the random deviation of $\estgradkn{\nbk}{\itr}$ from its expected value $ \expectation{ \estgradkn{\nbk}{\itr} \cond \filtrationn{\itr-1}   } $, so that $ \expectation{  \devmatkn{\nbk}{\itr} \cond \filtrationn{\itr-1}   }=0 $, and
\begin{equation} \label{earlydeviationbound} 
\nocolsep
\begin{array}{rcl}
\dualnorm{\devmatkn{\nbk}{\itr} } 
\obsolete{
&\leq& 
\dualnorm{\estgradkn{\nbk}{\itr}} + \dualnorm{\expectation{ \estgradkn{\nbk}{\itr}\cond \filtrationn{\itr-1}   } }  
\\
}
&\leq&
\dualnorm{\estgradkn{\nbk}{\itr}} + \expectation{\dualnorm{ \estgradkn{\nbk}{\itr}   } \cond \filtrationn{\itr-1}}  .
\end{array}
\end{equation}
In our analysis we consider the 
%
following random sequence:
\begin{equation} \label{MDSn}
\textstyle
\MDSn{\itr}= \stepsizen{\itr} \sum_{\nbk=1}^{\kk}   \inner{ \devmatkn{\nbk}{\itr}}{\matkn{\nbk}{\itr}-\optmatk{\nbk}} 
.
\end{equation}
Since $ \modulus{\inner{ \devmatkn{\nbk}{\itr}}{\matkn{\nbk}{\itr}{\,-\,}\optmatk{\nbk}}}{\,\leq\,}\dualnorm{ \devmatkn{\nbk}{\itr}}\norm{\matkn{\nbk}{\itr}{\,-\,}\optmatk{\nbk}} {\,\leq\,} 2 \dualnorm{ \devmatkn{\nbk}{\itr}}$, one has
\newcommand{\martingaledifference}{(\refeq{martingaledifferenceboth}a)}%
\newcommand{\modulusmartingaledifference}{(\refeq{martingaledifferenceboth}b)}%
\begin{equation}
\label{martingaledifferenceboth}
\textstyle
\text{a)}\quad 
\expectation{\MDSn{\itr}   \cond \filtrationn{\itr-1}   }
= 0 ,
\qquad
\text{b)}\quad
\modulus{\MDSn{\itr}} \leq 2\stepsizen{\itr} \sum_{\nbk=1}^{\kk}  \dualnorm{  \devmatkn{\nbk}{\itr}}
.
\end{equation}
\obsolete{
\begin{align}
\label{martingaledifference}
& \textstyle
\expectation{\MDSn{\itr}   \cond \filtrationn{\itr-1}   }
= 0 ,
\\
\label{modulusmartingaledifference}
& \textstyle
\modulus{\MDSn{\itr}} \leq 2\stepsizen{\itr} \sum_{\nbk=1}^{\kk}  \dualnorm{  \devmatkn{\nbk}{\itr}},
\end{align}
}%

\obsolete{
\begin{lemma}\label{lemma:azuma}
 If $\stepsizen{\itr}=\cststepsize$ and $\dualnorm{\estgradkn{\nbk}{\itr}}\leq\vboundk{\nbk}$ for $\nbk=1,\dots,\kk$ and for $\itr=1,\dots,T$, then 
 \begin{equation} \label{azumabis}
\textstyle
\Proba{\frac{\sum_{\itr=1}^{T}  \MDSn{\itr}  }{ T\zmargin\,\cststepsize  }  \leq \zmargin }
\geq 
1-\Exponential{-\frac{T\zmargin^2}{32     \left( \sum_{\nbk=1}^{\kk}   \vboundk{\nbk}  \right)^2 }}
.
\end{equation}
\end{lemma}
\begin{proof}
By assumption, we find from \eqref{earlydeviationbound} and \\martingaledifference{} 
that  
$
\modulus{\MDSn{\itr}} 
\leq
4\cststepsize \sum_{\nbk=1}^{\kk}   \vboundk{\nbk}
$ for $\itr=1,\dots,T$.
Since $\{\MDSn{\itr}\}$ is the bounded difference sequence of the martingale $\sum_{\itr=1}^{T} \MDSn{\itr}$, it follows from Azuma's inequality that, for  any $\zmargin>0$,
\begin{equation} \label{standardazuma}
\begin{array}{c}
\Proba{\sum_{\itr=1}^{T}  \MDSn{\itr}   > T\zmargin\,\cststepsize }
 \leq 
\Exponential{-\frac{(T\cststepsize\zmargin)^2}{2   T   \left(4\cststepsize \sum_{\nbk=1}^{\kk}   \vboundk{\nbk}  \right)^2 }}
\end{array}
\end{equation}
which is equivalent to \eqref{azumabis}.
\end{proof}

}%
\begin{lemma}\label{lemma:globaldescentargument}
\renewcommand{\dimensionk}[1]{\maxdimension}%
\renewcommand{\meansqdimension}{\maxdimension}%
\renewcommand{\meandimension}{\maxdimension}%
Run \ac{MXL0}/\ac{MXL+} for $\itr$ iterations under \eqref{deltacondition}.
\begin{enumerate}[(i),ref=\roman*,wide,labelwidth=!,labelindent=5pt]
\item \label{lemma:globaldescentargumenti}
 With any step-size and query radius policy $\vect{\stepsizen{\itr},\deltan{\itr}}$,
\begin{equation} \label{descentargumenttwo} 
\begin{array}{l}
\lyapunov{\optmat}{\dualmatn{\itr+1}}
 \leq
\lyapunov{\optmat}{\dualmatn{\itr}} - \stepsizen{\itr} \left[\optrate-\rate{\matn{\itr}}\right]
+  \MDSn{\itr}
\qquad
\\
\hfill 
+ 4\kk^2 \meanlipschitz \stepsizen{\itr}\deltan{\itr}  + \frac{\stepsizen{\itr}^2}{2\strongly} \sum_{\nbk=1}^{\kk}   \dualnorm{\estgradkn{\nbk}{\itr}}^2
\end{array}
\end{equation}
holds for~$\optmat\in\optMat$, where 
the sequence $\MDSn{\itr} $   is defined as in \eqref{MDSn}.

\item \label{lemma:globaldescentargumentii}
With decreasing policy $\vect{\stepsizen{\itr},\deltan{\itr}}=\vect{\cststepsize\,\itr^{-\alpha},\cstdelta\,\itr^{-\beta}}$, such that $\alpha,\beta\geq0$ and $\cststepsize,\cstdelta>0$,
\begin{equation} \label{expectationyergodicaverageidentity} 
\begin{array}{l}
\optrate -\bigexpectation{\rate{ \meanmatn{\itr}}}
 \leq
\frac{\lyapunov{\optmat}{\dualmatn{1}}}{\cststepsize\sum_{\altitr=1}^{\itr} \altitr^{-\alpha} } 
+ \frac{4\kk^2\meanlipschitz \cstdelta\sum_{\altitr=1}^{\itr} \altitr^{-\alpha-\beta} }{\sum_{\altitr=1}^{\itr} \altitr^{-\alpha} }  
\qquad\qquad
\\
\hfill 
+ \frac{ \cststepsize\sum_{\altitr=1}^{\itr}  \altitr^{-2\alpha} \sum_{\nbk=1}^{\kk}   \dualnorm{\estgradkn{\nbk}{\altitr}}^2}{2\strongly\sum_{\altitr=1}^{\itr} \altitr^{-\alpha} }
.
\end{array}
\end{equation}

\item \label{lemma:globaldescentargumentiii}
With constant policy $\vect{\stepsizen{\itr},\deltan{\itr}}{\,=\,}\vect{\cststepsize,\cstdelta}$, such that $\cststepsize,\cstdelta>0$,
\begin{equation} \label{standardtroptrateboundmeanexpectation} 
\textstyle
\optrate -\bigexpectation{ \rate{ \meanmatn{T}} }
 \leq
\frac{\lyapunov{\optmat}{\dualmatn{1}}}{T \cststepsize } 
+ 4\kk^2\meanlipschitz\cstdelta  
+ \frac{\cststepsize\sum\nolimits_{\itr=1}^{T}  \sum\nolimits_{\nbk=1}^{\kk}   \dualnorm{\estgradkn{\nbk}{\itr}}^2}{2\strongly T}
\end{equation}
for any $T\geq 1$.
Further, if  there exists $\vbound > 0  $ such that  $\dualnorm{\estgradkn{\nbk}{\itr}}\leq\dimensionk{\nbk}\vbound$ for $\nbk=1,\dots,\kk$ and $\itr=1,\dots,T$, then 
 \begin{equation} \label{standardazumabis}
\textstyle
\bigproba{\frac{ 1}{ T\cststepsize  } \sum_{\itr=1}^{T}  \MDSn{\itr}  \leq \zmargin }
\geq 
1-\exponential{-\frac{T\zmargin^2}{32 \vbound^2 \meandimension^2 }}
.
\end{equation}
\end{enumerate}
\end{lemma}
\begin{proof}[Proof of \cref{lemma:globaldescentargument}]
\renewcommand{\dimensionk}[1]{\maxdimension}%
\renewcommand{\meansqdimension}{\maxdimension}%
\renewcommand{\meandimension}{\maxdimension}%
\eqref{lemma:globaldescentargumenti}  If~$\optmat\in\optMat$, the concavity of $\ratefunction$ gives
\begin{equation}\label{concavityone}
\textstyle
\sum_{\nbk=1}^{\kk}\inner{ \projgradk{\nbk}\rate{\mat}  }{\matk{\nbk}-\optmatk{\nbk}} 
\leq
\rate{\mat}-\optrate
, 
\quad \forall\mat\in\compactset.
\end{equation}
%
It follows from  \cref{lemma:fenchel} that 
\begin{equation} \label{earlydescentargumenttwo} 
\begin{array}{l}
\lyapunov{\optmat}{\dualmatn{\itr+1}}
 \refereq{\eqref{MXL}}{=} 
 \sum_{\nbk=1}^{\kk}\fenchel{\optmatk{\nbk}}{\dualmatkn{\nbk}{\itr} + \stepsizen{\itr} \estgradkn{\nbk}{\itr}}
\\
\quad \refereq{\eqref{fenchelsmooth}}{\leq}    
\lyapunov{\optmat}{\dualmatn{\itr}} +\sum_{\nbk=1}^{\kk}\Big[ \stepsizen{\itr}  \inner{  \estgradkn{\nbk}{\itr}}{\matkn{\nbk}{\itr}-\optmatk{\nbk}}  + \frac{\stepsizen{\itr}^2}{2\strongly} \dualnorm{\estgradkn{\nbk}{\itr}}^2
\Big]
\\
\quad \refereq{\eqref{estimatordecomposition}}{=}   
\lyapunov{\optmat}{\dualmatn{\itr}} + \stepsizen{\itr}  \sum_{\nbk=1}^{\kk}\inner{ \projgradk{\nbk}\rate{\matn{\itr}}  }{\matkn{\nbk}{\itr}-\optmatk{\nbk}}  +  \MDSn{\itr} 
\\
\hfill 
+ \sum_{\nbk=1}^{\kk}  \Big[\stepsizen{\itr}  \inner{   \biaskn{\nbk}{\itr}  }{\matkn{\nbk}{\itr}-\optmatk{\nbk}}  + \frac{\stepsizen{\itr}^2}{2\strongly} \dualnorm{\estgradkn{\nbk}{\itr}}^2\Big]
,
\end{array}
\end{equation}
%
Besides, \eqref{biasbound} gives
$
\modulus{ \inner{   \biaskn{\nbk}{\itr}  }{\matkn{\nbk}{\itr}-\optmatk{\nbk}} }
\leq
4\kk\meanlipschitzk{\nbk}\deltan{\itr} 
$,
%
which combined with \eqref{concavityone} and \eqref{earlydescentargumenttwo} yields Inequality \eqref{descentargumenttwo}.
%

\eqref{lemma:globaldescentargumentii}
By telescoping \eqref{descentargumenttwo} $\itr-1$ times, dividing by $\sum_{\altitr=1}^{\itr} \stepsizen{t} $, and using 
$\lyapunov{\optmat}{\dualmatn{n+1}}\geq 0$, 
we find
\begin{equation} \label{ergodicaverageidentity} 
\begin{array}{l}
\optrate - \frac{\sum_{\altitr=1}^{\itr} \stepsizen{\altitr} \rate{\matn{\altitr}}}{\sum_{\altitr=1}^{\itr} \stepsizen{\altitr} }
 \leq
\frac{\lyapunov{\optmat}{\dualmatn{1}}}{\sum_{\altitr=1}^{\itr} \stepsizen{\altitr} } + \frac{\sum_{\altitr=1}^{\itr}  \MDSn{\altitr} }{\sum_{\altitr=1}^{\itr} \stepsizen{\altitr} }
\qquad\qquad\quad
\\
\hfill 
+ 4\kk^2\meanlipschitz\frac{\sum_{\altitr=1}^{\itr} \stepsizen{\altitr}\deltan{\altitr} }{\sum_{\altitr=1}^{\itr} \stepsizen{\altitr} }  
+ \frac{1}{2\strongly}\frac{\sum_{\altitr=1}^{\itr} \stepsizen{\altitr}^2 \sum_{\nbk=1}^{\kk}   \dualnorm{\estgradkn{\nbk}{\altitr}}^2}{\sum_{\altitr=1}^{\itr} \stepsizen{\altitr} }
.
\end{array}
\end{equation}
By concavity of $\ratefunction$, the time average of the estimates satisfies
\begin{equation}\label{concavitymean}
\textstyle
\rate{ \meanmatn{\itr}}
\geq
\big(\frac{
1
}{
\sum_{\altitr=1}^{\itr} \stepsizen{\altitr} 
}\big)
\sum_{\altitr=1}^{\itr} \stepsizen{\altitr}   \rate{\matn{\altitr}}   
.
\end{equation}
Introducing  the suggested policies in \eqref{ergodicaverageidentity} and using \eqref{concavitymean} gives
\begin{equation} \label{readyergodicaverageidentity} 
\begin{array}{l}
\optrate - \rate{ \meanmatn{\itr}}
 \leq
\frac{\lyapunov{\optmat}{\dualmatn{1}}}{\cststepsize\sum_{\altitr=1}^{\itr} \altitr^{-\alpha} } + \frac{\sum_{\altitr=1}^{\itr}  \MDSn{\altitr} }{\cststepsize \sum_{\altitr=1}^{\itr} \altitr^{-\alpha} }
\qquad\qquad\qquad\quad
\\
\hfill 
+ 4\kk^2\meanlipschitz\frac{ \cstdelta\sum_{\altitr=1}^{\itr} \altitr^{-\alpha-\beta} }{\sum_{\altitr=1}^{\itr} \altitr^{-\alpha} }  
+ \frac{\cststepsize}{2\strongly}\frac{ \sum_{\altitr=1}^{\itr}  \altitr^{-2\alpha} \sum_{\nbk=1}^{\kk}   \dualnorm{\estgradkn{\nbk}{\altitr}}^2}{\sum_{\altitr=1}^{\itr} \altitr^{-\alpha} }
.
\end{array}
\end{equation}
\obsolete{
Applying $\itr$ times \martingaledifference{} 
gives
\begin{equation} \label{nullexpectationofsum} 
\begin{array}{r}
 \Expectation{ \sum_{\altitr=1}^{\itr}  \MDSn{\altitr}} 
=
\Expectation{ \Expectation{ \sum_{\altitr=1}^{\itr}  \MDSn{\altitr} \cond \filtrationn{0}  }} 
\ \ \\ \hfill
 =
\Expectation{ \Expectation{ \sum_{\altitr=2}^{\itr}  \MDSn{\altitr} \cond \filtrationn{1}  } }
 = \hdots =
\Expectation{ \Expectation{   \MDSn{\itr} \cond \filtrationn{\itr-1}  } } = 0.
\end{array} 
\end{equation}
}%
Since \martingaledifference{} 
lends $\{\MDSn{\itr}\}$ the quality of a martingale difference sequence, $\expectation{ \sum_{\altitr=1}^{\itr}  \MDSn{\altitr}}=0$, and \eqref{expectationyergodicaverageidentity} follows by expectation of \eqref{readyergodicaverageidentity}. 

\eqref{lemma:globaldescentargumentiii}
Setting $\alpha=\beta=0$ in \eqref{expectationyergodicaverageidentity}  gives us \eqref{standardtroptrateboundmeanexpectation}.
By using the bounds $\vboundk{1},\dots,\vboundk{\kk}  $ in combination with \eqref{earlydeviationbound} and \modulusmartingaledifference{}, 
we find that
$
\modulus{\MDSn{\itr}} 
\leq
4\vbound\meandimension\cststepsize 
$ for $\itr=1,\dots,T$, and the martingale difference sequence $\{\MDSn{\itr}\}$ is bounded.
It follows from Azuma's inequality that
$
\textstyle
\proba{\sum_{\itr=1}^{T}  \MDSn{\itr}   > T\zmargin\,\cststepsize }
 \leq 
\brackexponential{-\frac{(T\cststepsize\zmargin)^2}{2   T   (4\vbound\meandimension\cststepsize)^2 }} 
$ 
 for  any $\zmargin>0$,
which is equivalent to \eqref{standardazumabis}.
\end{proof}
\showMXLzerolastiterate{
\cref{thm:lasttrMXL,thm:trMXLiinew} follow from \cref{lemma:bias,lemma:globaldescentargument}.
\begin{proof}[Proof of \cref{thm:lasttrMXL}]
Following the line of thought of the proof of \cite[Theorem 5.1]{BLM18}, we first show there one can find a solution $\optmat\in\optMat$ such that
\begin{equation} \label{liminfassumption}
\textstyle
\liminf_{\itr\to\infty} \Lyapunov{\optmat}{\dualmatn{\itr}} = 0
\quad
\as \, .
\end{equation}
Next, we see that $\{\lyapunov{\optmat}{\dualmatn{\itr}}\}$ converges almost surely (\as{}) towards a finite quantity which, in view of \eqref{liminfassumption}, can only be~$0$. 
%
 \As{} convergence of $\{\matn{\itr}\}$ towards $\optmat$ can then be inferred from \cref{lemma:fenchel}-\eqref{fenchelstronglyconvex}.
%
The assumption of  non-increasing $\{\deltan{\itr}\}$, together with \reflasttrMXLi{} and \reflasttrMXLii{}, implies $\deltan{\itr}\downarrow 0$ which, in view of \eqref{perturbatorbound}, secures \as{} convergence of $\{\testmatn{\itr}\}$ as well.

First observe that \eqref{liminfassumption} holds if, almost surely, there exists a subsequence of $\{\matn{\itr}\}$ that converges towards a solution~$\optmat\in\optMat$.
Suppose this condition not to hold, and let $\limitpoints$ denote the set of the limit points of all  subsequences of $\{\matn{\itr}\}$. 
%
Then, almost surely, we have $\optMat\cap \limitpoints=\emptyset$ and, since $\limitpoints$ is closed 
by construction and $\ratefunction$ is continuous and  convex, 
$ \varrho \defeq \optrate - \max_{\mat\in \limitpoints}{\rate{\mat}}>0$.
\obsolete{
\footnotetext{Indeed, suppose there is  a sequence $\{\compactmatn{\altitr}\}$ in $\limitpoints$ converging to $\compactmatn{\infty}\notin \limitpoints$. For any be any sequence of positive numbers $\{\epsilon_\altitr\}$ with $\epsilon_\altitr\downarrow 0$, one can find $\{\matn{\itr_\altitr}\}$, subsequence of $\{\matn{\itr}\}$, such that $\norm{\matn{\itr_\altitr}-\compactmatn{\altitr}}<\epsilon_\altitr$ for all $\altitr$. Hence, $\matn{\itr_\altitr}\to\compactmatn{\infty}$ and we reach the contradiction $\compactmatn{\infty} \in \limitpoints$.}%
}%

Telescoping \eqref{descentargumenttwo}  in  \cref{lemma:globaldescentargument}\eqref{lemma:globaldescentargumenti} and using
 \eqref{trboundgradientestimateMO}, yields
%
\begin{equation} \label{adabsurdumlyapunovdescentbound} 
\begin{array}{l}
\hspace{-0mm}
\lyapunov{\optmat}{\dualmatn{\itr+1}}
 \leq
\lyapunov{\optmat}{\dualmatn{1}} - \sum_{\altitr=1}^{\itr} \stepsizen{\altitr} \left[\optrate-\rate{\matn{\altitr}}\right]
\quad \  \\ 
\hfill 
{+}  \sum_{\altitr=1}^{\itr}\MDSn{\altitr}  {+\,} 4\kk^2 \meanlipschitz \sum_{\altitr{=}1}^{\itr} \stepsizen{\altitr}\deltan{\altitr}  
{+}\frac{\kk(\optrate\maxdimension)^2}{2^{2\kk{+}1}\strongly} \sum_{\altitr{=}1}^{\itr} \frac{\stepsizen{\altitr}^2}{\deltan{\altitr}^2} 
,
\end{array}
\end{equation}
where $\{\MDSn{\itr}\}$ is the difference sequence of a martingale  with respect to the filtration $\{\filtrationn{\itr}\}$.
In view of \reflasttrMXLii{} and \reflasttrMXLiii{},  the last two terms in the second member of \eqref{adabsurdumlyapunovdescentbound} converge as $\itr\to\infty$. As for the third term, since
\begin{equation}   \label{sumMDSbound}
\textstyle 
\sum_{\altitr=1}^{\infty}  \bigexpectation{\MDSn{\altitr}^2\cond \filtrationn{\altitr-1}  } 
\leq 
\frac{4\kk(\optrate\maxdimension)^2}{2^{2\kk}}   \sum_{\altitr=1}^{\infty}  \frac{ \stepsizen{\altitr}^2 }{\deltan{\altitr}^2}        
 \refereq{\reflasttrMXLiii{}}{<} 
 \infty ,
\end{equation}
\cite[Theorem 2.18]{hall80} applies with parameter $p=2$, and it follows that $\sum_{\altitr=1}^{\itr} \MDSn{\altitr}$ converges \as{} as $\itr\to\infty$. 
Finally, one can find a subsequence $\{\matn{\itr_\altitr}\}$ that converges to a point of $\limitpoints$ and thus satisfies $  \optrate - \rate{\matn{\itr_\altitr}} >  \varrho/2 $ for $\altitr$ large enough. It follows from \reflasttrMXLi{} that the second term  $\sum_{\altitr=1}^{\infty} \stepsizen{\itr} [   \rate{\matn{\itr}} - \rate{\optmat} ] \to-\infty$.
All in all we find that $\Lyapunov{\optmat}{\dualmatn{\itr}}\to-\infty$ \as{}, which is in contradiction with the nonnegativity of $\lyapunovfunction$. We infer that \eqref{liminfassumption} is true.
%

It remains to show that $\{\lyapunov{\optmat}{\dualmatn{\itr}}\}$ is almost surely convergent.
%
%
To do so we rely on
Doob's  convergence theorem for supermartingales \cite[Theorem 2.5]{hall80}.
Recalling \eqref{descentargumenttwo}, and using \eqref{trboundgradientestimateMO} and $\rate{\matn{\itr}} - \rate{\optmat}  \leq 0$, we find
\begin{equation}\label{nottightlyapunovdescentbound}
\begin{array}{l} 
\hspace{-1mm}
\lyapunov{\optmat}{\dualmatn{\itr+1}}
 { \,\leq\,}
\lyapunov{\optmat}{\dualmatn{\itr}}
{\,+\,}\MDSn{\itr} 
{\,+\,} 4\kk^2\meanlipschitz\stepsizen{\itr}\deltan{\itr}   {\,+\,} \frac{  \kk ( \optrate    \maxdimension\stepsizen{\itr})^2}{2^{2\kk{+}1}\strongly \deltan{\itr}^2 }
.
\end{array}
\end{equation}
Consider $\supermartn{\itr}=  \sum_{\altitr=\itr+1}^{\infty} [  4\kk^2\meanlipschitz\stepsizen{\altitr}\deltan{\altitr}   + {  \kk ( \optrate    \maxdimension\stepsizen{\altitr})^2}/({2^{2\kk{+}1}\strongly \deltan{\altitr}^2 })
] + \lyapunov{\optmat}{\dualmatn{\itr+1}} $. 
Under assumptions \reflasttrMXLii{} and \reflasttrMXLiii{}, $\supermartn{0}$ is  finite by construction.
We infer from \martingaledifference{} 
and \eqref{nottightlyapunovdescentbound} that $\expectation{\supermartn{\itr}\cond \filtrationn{\itr-1} } \leq \supermartn{\itr-1}$ for     $\itr\geq 1$,
and $\{\supermartn{\itr}\}$ is a supermartingale with respect to $\{ \filtrationn{\itr}\}$, thus satisfying $\Expectation{\supermartn{\itr} } \leq \supermartn{0} <\infty$.
Hence, $\{\supermartn{\itr}\}$ is uniformly $L^1$-bounded and Doob's theorem applies.  It follows that $\{\supermartn{\itr}\}$, and consequently $\{\lyapunov{\optmat}{\dualmatn{\itr}}\}$, are almost surely convergent, which completes the proof.
\end{proof}
}{ 
\cref{thm:trMXLiinew} is a consequence of \cref{lemma:bias,lemma:globaldescentargument}.
}%

\begin{proof}[Proof of \cref{thm:trMXLiinew}]
\renewcommand{\dimensionk}[1]{\maxdimension}%
\renewcommand{\meansqdimension}{\maxdimension}%
\renewcommand{\meandimension}{\maxdimension}%
%
%
By considering \cref{lemma:globaldescentargument}\eqref{lemma:globaldescentargumentiii} with the upper bounds $(\dimensionk{\nbk}\vbound)={  \dimensionk{\nbk}}  \optrate/({2^{\kk-1}\cstdelta})$ supplied by \eqref{trboundgradientestimateMO}, we find
\begin{equation} \label{basictroptrateboundmeanexpectation} 
\textstyle
\optrate -\Expectation{ \rate{ \meanmatn{T}} }
\refereq{\eqref{standardtroptrateboundmeanexpectation}}{\leq}
%
\frac{\kk\log\maxmm}{T \cststepsize } + 4\kk^2\meanlipschitz\cstdelta  + \frac{\kk   (\optrate\meansqdimension )^2 \cststepsize}{\strongly  2^{2\kk-1}  \cstdelta^2}
,
\end{equation} 
where  we have used $\lyapunov{\optmat}{\dualmatn{1}} \leq \kk\log\maxmm  $, and
 \begin{equation} 
 \label{basicazumabis}
\textstyle
\proba{\frac{1 }{ T\cststepsize  } \sum_{\itr=1}^{T}  \MDSn{\itr}  \leq \zmargin }
\refereq{\eqref{standardazumabis}}{\geq}
1-\exponential{-\frac{2^{2\kk-5} T\zmargin^2\cstdelta^2}{(\optrate \kk \meandimension)^2 }  }
.
\end{equation}
The right member of \eqref{basictroptrateboundmeanexpectation} is convex in $\vect{\cststepsize,\cstdelta}$ and minimized for the policy $\vect{\cststepsize,\cstdelta}=\vect{\stepsize\,T^{- 3/4},\delta\, T^{- 1 / 4}}$, where
\obsolete{
$$ a = \kk\log\maxmm, \ b = 4\kk^2\meanlipschitz  , \ c= \frac{\kk  (\optrate \meansqdimension)^2  }{\strongly  2^{2\kk-1}  } $$
$$
\sqrt{ \frac{2}{b\sqrt{c}}  } = \sqrt{ \frac{\sqrt{ \strongly  2^{2\kk-3}  } }{\meanlipschitz   \optrate \sqrt{ \kk^5   }\meansqdimension}  }  = \frac{1}{\sqrt{\meanlipschitz   \optrate}} \left(\frac{\strongly  2^{2\kk-3}}{\kk^5 \meansqdimension^2}\right)^{1/4}
$$

$$ \sqrt{ \frac{b\sqrt{c}}{2}  } = \sqrt{ \frac{\meanlipschitz   \optrate \sqrt{ \kk^5   }\meansqdimension}{\sqrt{ \strongly  2^{2\kk-3}  } }  } = \sqrt{\meanlipschitz   \optrate} \left(\frac{\kk^5 \meansqdimension^2}{\strongly  2^{2\kk-3}}\right)^{1/4}   $$

$$  f(\cststepsize,\cstdelta) =  \frac{a}{T \cststepsize } + b \cstdelta  +c \frac{\cststepsize}{  \cstdelta^2} $$
$$  -  \frac{a}{T \cststepsize^2} +c \frac{1}{  \cstdelta^2} =0 \Leftrightarrow \cststepsize = \sqrt{\frac{a}{cT}} \cstdelta  $$
$$   b  - 2 c \frac{\cststepsize}{  \cstdelta^3} = 0 \Leftrightarrow \cststepsize = \frac{b}{2c} \cstdelta^3  $$
$$ \Rightarrow  \cststepsize = \frac{b}{2c} \left(\sqrt{\frac{cT}{a}}\cststepsize\right)^3= \sqrt{ \frac{2c}{b} \left(\sqrt{\frac{a}{cT}}\right)^3 } $$
$$ \Rightarrow \cststepsize = \sqrt{ \frac{2}{b\sqrt{c}}  } \left(\frac a T\right)^{3/4}$$
$$ \Rightarrow \cstdelta =   \sqrt{\frac{cT}{a}} \cststepsize =  \sqrt{\frac{cT}{a}} \sqrt{ \frac{2}{b\sqrt{c}}  } \left(\frac a T\right)^{3/4} $$
$$ \Rightarrow \cstdelta =  \sqrt{ \frac{2\sqrt{c}}{b}  } \left(\frac a T\right)^{1/4}  $$
$$
\begin{array}{rl}
   f(\cststepsize,\cstdelta) 
   &=  
   \frac{a}{T \sqrt{ \frac{2}{b\sqrt{c}}  } \left(\frac a T\right)^{3/4} } + b \sqrt{ \frac{2\sqrt{c}}{b}  } \left(\frac a T\right)^{1/4}  +c \frac{\sqrt{ \frac{2}{b\sqrt{c}}  } \left(\frac a T\right)^{3/4}}{  \left(\sqrt{ \frac{2\sqrt{c}}{b}  } \left(\frac a T\right)^{1/4}\right)^2} 
\\   &=  
   \sqrt{ \frac{b\sqrt{c}}{2}  }\left(\frac a T\right)^{1/4} +  2 \sqrt{ \frac{b\sqrt{c}}{2}  } \left(\frac a T\right)^{1/4}  + \sqrt{ \frac{b\sqrt{c}}{2}  }  \left(\frac a T\right)^{1/4}
\\ &=
4 \sqrt{ \frac{b\sqrt{c}}{2}  } \left(\frac a T\right)^{1/4} .
\end{array}
$$
------------
$$ \Rightarrow \cststepsize 
=\sqrt{ \frac{\sqrtstrongly  2^{\kk}}{\meanlipschitz   \optrate \kk \meansqdimension}} \left(\frac{ \log\maxmm}{2}\right)^{3/4} T^{-3/4}
$$
$$ \Rightarrow \cstdelta 
=  2 \sqrt{\frac{\meanlipschitz   \optrate\kk^3 \meansqdimension}{\strongly  2^{\kk}}}   \left(\frac{\log\maxmm}{2}\right)^{1/4}  T^{-1/4}
$$
$$
\begin{array}{rl}
 \Rightarrow  f(\cststepsize,\cstdelta) 
 &=
8   \sqrt{ \frac{\meanlipschitz   \optrate \kk^3 \meansqdimension}{\strongly  2^{\kk}}}  \left(\frac{\log\maxmm}{2T}\right)^{1/4} .
\end{array}
$$
---------

\begin{equation} \label{optbasictroptrateboundmeanexpectation} 
\textstyle
\optrate -\Expectation{ \rate{ \meanmatn{T}} }
\leq
\left[
\frac{\kk\log\maxmm}{\stepsize} + 4\kk^2\meanlipschitz\, \delta  + \left(\frac{\kk   \meansqdimension^2  (\optrate)^2 }{\strongly  2^{2\kk-1}   }\right) \frac{\stepsize}{\delta^2 }
\right]  T^{- 1/4}
,
\end{equation}

}%
\begin{equation}
 \textstyle
\stepsize 
=
\sqrt{ \frac{\sqrtstrongly  2^{\kk}}{\meanlipschitz   \optrate \kk \meansqdimension}} \left(\frac{ \log\maxmm}{2}\right)^{3/4} 
, \quad
\delta 
=
2 \sqrt{\frac{\meanlipschitz   \optrate\kk^3 \meansqdimension}{\strongly  2^{\kk}}}   \left(\frac{\log\maxmm}{2}\right)^{1/4}
.
\end{equation}
We find \eqref{troptrateboundmeanexpectation} by substituting $\cststepsize$ and $\cstdelta$ in \eqref{basictroptrateboundmeanexpectation} with the suggestion $\vect{\cststepsize,\cstdelta}=\vect{\stepsize\,T^{- 3/4},\delta\, T^{- 1 / 4}}$. Then, \eqref{troptrateboundmeaninterval} follows from \eqref{troptrateboundmeanexpectation} and \eqref{basicazumabis} after setting $ \cstdelta=\delta\, T^{- 1 / 4}$ in the right member of \eqref{basicazumabis}. 
Claims \eqref{thm:trMXLiifirst} and \eqref{thm:trMXLiisecond} have been shown.
\obsolete{ 
\begin{equation} 
\textstyle
\Proba{\frac{\sum_{\itr=1}^{T}  \MDSn{\itr}  }{ T\cststepsize  }  \leq \zmargin }
\refereq{\eqref{standardazumabis}}{\geq}
1-\Exponential{-\left[\frac{2^{2\kk-5} \delta^2}{\optrate^2 \kk^2 \meandimension^2 }\right]  \sqrt{T}\zmargin^2 }
.
\end{equation}
}%
\end{proof}

\subsection{Analysis of the \ac{MXL+} algorithm}
\label{appendix:MXLzeroplus}

The $\magnitude{\delta}$ bound for the bias in \cref{lemma:bias} still holds when the \ac{SPSA}plus{} gradient estimator is used.
The offset in \eqref{estgradratekMO} allows us,  however, to derive an $\magnitude{1/\delta}$ bound for the norm, in place of the harmful $\magnitude{1/\delta}$ bound  inherent with \ac{SPSA}{}.

\obsolete{
For this, observe that the sequence generated by \eqref{MXL} satisfies
\begin{equation}\label{BIBO} \begin{array}{rcl}
\norm{\matn{\itr} -\matn{\itr-1} }
&\refereq{\globalnorm{}}{=}&
\sum_{\nbk=1}^{\kk} \norm{\matkn{\nbk}{\itr}-\matkn{\nbk}{\itr-1}} 
\\
&\refereq{\eqref{DAlipschitz}}{\leq}&
\sum_{\nbk=1}^{\kk} \oneoverstrongly  \dualnorm{ \stepsizen{\itr} \estgradkn{\nbk}{\itr-1} }
\\
&\refereq{\globaldualnorm{}}{=}&
\overstrongly{\kk  \stepsizen{\itr}} \dualnorm{ \estgradn{\itr-1}}
.
\end{array}
\end{equation}
 
}%

\begin{lemma}
\label{lemma:norm}
%
If \ac{MXL+} is implemented under \eqref{deltacondition} and \eqref{generalfirstcondition}-\eqref{generalsecondcondition}, then $\dualnorm{\estgradkn{\nbk}{\itr}}$ is uniformly bounded 
for $\nbk=1,\dots,\kk$.

In particular, if $\vect{\stepsizen{\itr},\deltan{\itr}}=\vect{\stepsize \, \itr^{-\alpha},\delta\, \itr^{-\beta}}$ with 
\begin{equation}\textstyle 
\label{oldcondition}
\textup{(a)}\quad 0 \leq \beta \leq \alpha,
\qquad\quad
\textup{(b)} \quad \maxdimension \Lipschitz\kk \, \stepsize < 2 \strongly \, \delta ,
\end{equation}
then there is $\normbound{\alpha}{\beta}{\stepsize}{\delta}
{\,<\,}\infty$  such that  
$
\dualnorm{\estgradkn{\nbk}{\itr}}
{\,\leq\,} 
\frac{\dimensionk{\nbk}}{2}
\,\normbound{\alpha}{\beta}{\stepsize}{\delta}
$ 
holds for all $\itr$ and for $\nbk=1,\dots,\kk$ and, when $\beta=0$, \begin{equation} \label{normboundtwo}
\textstyle
\normbound{\alpha}{0}{\stepsize}{\delta}
=
\Big(  \frac{  \twotwitchFlipschitz \coef{\alpha}    \strongly}{\sqrt{\maxdimension}  } \Big) \left( \frac{2\strongly}{\maxdimension \Lipschitz\kk  } -  \frac{\stepsize  }{\delta} \right)^{-1}
.
\end{equation}
\end{lemma}
\begin{proof}[Proof of \cref{lemma:norm}]
With the convention $\offsetn{0}=  \rate{\matn{1}}$, we have, for $\nbk=1,\dots,\kk$,
\begin{equation}
\label{firstsecondmomentcomputation}
\nocolsep
\begin{array}{c}
\dualnorm{\estgradkn{\nbk}{1}}
 \refereq{\!\!\!\!\eqref{estgradratekMO}\!\!\!\!}{\leq}  
\frac{\dimensionk{\nbk}}{\deltan{1}} \modulus{  \rate{ \testmatn{1}} - \rate{\matn{1}} } \dualnorm{\zmatkn{\nbk}{1}}
\qquad\qquad\qquad\\\hfill
  \refereq{\eqref{lipschitz}}{\leq} 
\frac{\dimensionk{\nbk} \Lipschitz}{2\deltan{1}} \norm{  {\testmatn{1}}  - {\matn{1}}} 
 \refereq{\eqref{perturbatorbound}}{\leq} 
\mutehalftwitchFlipschitz \dimensionk{\nbk} \Lipschitz   \kk\sqrt{\maxdimension}  
,
\end{array}
\end{equation}
and it follows from \globaldualnorm{} that
$
\dualnorm{\estgradn{1}}
\leq \mutehalftwitchFlipschitz\maxdimension \Lipschitz \kk\sqrt{\maxdimension}   
$.
%
For $\itr\geq 2$,
\begin{equation}
\label{earlysecondmomentcomputation}
\nocolsep
\begin{array}{l}
\dualnorm{\estgradkn{\nbk}{\itr}}
 \refereq{\eqref{estgradratekMO}}{\leq}  
\frac{\dimensionk{\nbk}}{\deltan{1}} \modulus{  \rate{ \testmatn{\itr}} - \rate{ \testmatn{\itr-1}} } \dualnorm{\zmatkn{\nbk}{\itr}}
\\ \qquad \ 
\obsolete{
\leq
\frac{\dimensionk{\nbk}}{2\deltan{\itr}} \Big| [  \rate{ \testmatn{\itr}} - \rate{ \matn{\itr}} ]
\\ \hfill
+ [\rate{ \matn{\itr}} -\rate{ \matn{\itr-1}}] - [\rate{ \testmatn{\itr-1}}- \rate{ \matn{\itr-1}} ]\Big|
\\ \qquad \ 
}
\leq   
\frac{\dimensionk{\nbk}}{2\deltan{\itr}} \big[ \modulus{ \rate{ \testmatn{\itr}} - \rate{ \matn{\itr}} } 
+ \modulus{\rate{ \matn{\itr}} -\rate{ \matn{\itr-1}} } 
\\  \hfill
+ \modulus{\rate{ \testmatn{\itr-1}}- \rate{ \matn{\itr-1}} }\big]
\\ \qquad \ 
  \refereq{\eqref{lipschitz}}{\leq} 
\frac{\dimensionk{\nbk}}{2\deltan{\itr}} \big[\Lipschitz \norm{\testmatn{\itr} -  \matn{\itr} } +\Lipschitz \norm{\matn{\itr}-\matn{\itr-1}}
\\ \hfill
 +\Lipschitz \norm{ \testmatn{\itr-1}-  \matn{\itr-1} }\big]
\\ \qquad \ 
  \refereq{\eqref{perturbatorbound}}{\leq} 
\frac{\dimensionk{\nbk}\Lipschitz}{2\deltan{\itr}} \big[ \twitchFlipschitz \kk\sqrt{\maxdimension} (\deltan{\itr}+\deltan{\itr-1})+  \sum_{\nbk=1}^{\kk} \norm{\matkn{\nbk}{\itr}-\matkn{\nbk}{\itr-1}}  \big]
\\ \qquad \ 
 \refereq{\eqref{DAlipschitz}}{\leq} 
\frac{\dimensionk{\nbk}\Lipschitz}{2\deltan{\itr}} \big[\twitchFlipschitz \kk\sqrt{\maxdimension}  (\deltan{\itr}+\deltan{\itr-1})+ \overstrongly{\kk \stepsizen{\itr}} \dualnorm{ \estgradn{\itr-1} } \big]
\obsolete{
\\ \
 \refereq{\globaldualnorm{}}{\leq} 
\dimensionk{\nbk}\frac{1}{2}\Lipschitz\kk\sqrt{\maxdimension}  \twitchFlipschitz\left(1+\frac{\deltan{\itr-1}}{\deltan{\itr}}\right) 
\\ 
+ 
\frac{\dimensionk{\nbk}\Lipschitz\kk  \stepsizen{\itr}}{2\strongly\deltan{\itr}} \max_{\altnbk\in\{1,\dots,\kk\}}\dualnorm{ \estgradkn{\altnbk}{\itr-1} } 
}
,
\end{array}
\end{equation}
so that
$\textstyle
\dualnorm{\estgradn{\itr}}
\leq
\frac{\maxdimension \Lipschitz}{2\deltan{\itr}} [\twitchFlipschitz \kk\sqrt{\maxdimension}  (\deltan{\itr}{+}\deltan{\itr{-}1}) {+} \overstrongly{\kk \stepsizen{\itr}} \dualnorm{ \estgradn{\itr{-}1} } ] 
.$
With the convention $\deltan{0}=0$, we find, by induction on $\itr$,
\begin{equation}
\label{onesecondmomentcomputation}
\begin{array}{l}
\dualnorm{\estgradkn{\nbk}{\itr}}
 {\,\leq\,}
\mutehalftwitchFlipschitz \Lipschitz\kk\dimensionk{\nbk}\!\sqrt{\maxdimension}  
 \sum\nolimits_{\altitr=1}^{\itr}
\big( \frac{\maxdimension \Lipschitz\kk  }{2\strongly} \big)^{\itr{-}\altitr}  \prod\nolimits_{\altaltitr=\altitr+1}^{\itr}  \frac{\stepsizen{\altaltitr-1}}{\deltan{\altaltitr}}  \big(1{+}\frac{\deltan{\altitr-1}}{\deltan{\altitr}}\big)  
.
\end{array}
\end{equation}
%

Condition \eqref{generalsecondcondition} tells us that ${\deltan{\itr-1}}/{\deltan{\itr}}$ is uniformly bounded by a finite constant, say, $c<\infty$, while \eqref{generalfirstcondition} rewrites as
\begin{equation}\textstyle
q \defeq \frac{{\maxdimension} \Lipschitz\kk }{2\strongly} \, \big( \sup_{\itr\geq 2} \frac{\stepsizen{\itr-1}}{\deltan{\itr} }  \big) <1 .
\end{equation}
Using $\frac{\stepsizen{\itr-1}}{\deltan{\itr}}\leq \frac{ 2 \strongly q}{\maxdimension \Lipschitz\kk}$ and $\frac{\deltan{\itr-1}}{\deltan{\itr}}\leq c$ in \eqref{onesecondmomentcomputation}, we find, for $\itr\geq 2$,
\begin{equation}
\label{newsecondmomentcomputation}
\begin{array}{l}
\frac{\dualnorm{\estgradkn{\nbk}{\itr}}}{\mutehalftwitchFlipschitz \Lipschitz\kk\dimensionk{\nbk}\sqrt{\maxdimension} }
\leq 
 \sum_{\altitr=1}^{\itr}q^{\itr-\altitr} \left(1+c\right)  
= 
(1+c)\frac{1-q^{\itr}}{1-q}
\leq
\frac{1+c}{1-q}
.
%
\end{array}
\end{equation}
%
%
Under the policies $\stepsizen{\itr}=\stepsize \, \itr^{-\alpha}$  and $\deltan{\itr}=\delta\, \itr^{-\beta}$, \eqref{onesecondmomentcomputation} becomes
\begin{equation} \label{twosecondmomentcomputation}
\begin{array}{l}
\dualnorm{\estgradkn{\nbk}{\itr}}
%
 \leq
\mutehalftwitchFlipschitz\coef{\alpha}\dimensionk{\nbk} \Lipschitz 
\kk\sqrt{\maxdimension}  \sum\nolimits_{\altitr=1}^{\itr}\left( \frac{\stepsize\maxdimension \Lipschitz\kk  }{2\delta\strongly} \right)^{\itr-\altitr} \left( \frac{\itr}{\altitr}\right)^\beta \left[    \frac{ \factorial{(\itr-1)}}{ \factorial{(\altitr-1)}} \right]^{\beta-\alpha} 
,
\end{array}
\end{equation}
where $\coef{\alpha}=1+2^\alpha$. Under  Condition \oldsecondcondition{} 
the last factor is no larger than $1$, and we 
obtain the uniform bound with
\begin{equation}\label{normbound}  
\textstyle 
\normbound{\alpha}{\beta}{\stepsize}{\delta}
=
 \twitchFlipschitz\Lipschitz (1+2^\alpha)\kk\sqrt{\maxdimension}   \sum_{\altitr=1}^{\infty}\big[ \frac{\stepsize\maxdimension \Lipschitz\kk  }{2\delta\strongly} \big]^{\itr-\altitr} \big( \frac{\itr}{\altitr}\big)^\beta,
\end{equation} 
which is finite on condition that \oldfirstcondition{} holds. 
For $\beta=0$,
 \eqref{normbound} 
reduces to a geometric series and \eqref{normboundtwo} follows directly.
\end{proof}

We are now able to show \cref{thm:lasttrMXLMO} and \cref{timecomplexity:trMXLMO}. Again, the Lyapunov function \eqref{lyapunov} and \cref{lemma:globaldescentargument} are used.

\showMXLzerolastiterate{
\begin{proof}[Proof of \cref{thm:lasttrMXLMO}]
Proceed as in the proof of \cref{thm:lasttrMXL}, now with assumptions \trconditionstepsize{}, \trconditionsquared{} and \eqref{trconditionbias} in place of \reflasttrMXLi{}, \reflasttrMXLii{}, \reflasttrMXLiii{}.
Because the conditions of \cref{lemma:norm} are met, there exists $ \cstnormbound <\infty$ such that $  \dualnorm{\estgradkn{\nbk}{\itr}}< \cstnormbound$ for all $k$, so that \eqref{adabsurdumlyapunovdescentbound} and \eqref{nottightlyapunovdescentbound} respectively become, for some~$\optmat\in\optMat$,
\begin{equation} \label{adabsurdumlyapunovdescentboundplus} 
\begin{array}{l}
\lyapunov{\optmat}{\dualmatn{\itr+1}}
 \leq
\lyapunov{\optmat}{\dualmatn{1}} - \sum_{\altitr=1}^{\itr} \stepsizen{\altitr} \left[\optrate-\rate{\matn{\altitr}}\right]
\quad
\\
\hfill 
+  \sum_{\altitr=1}^{\itr}\MDSn{\altitr}  + 4\kk^2 \meanlipschitz \sum_{\altitr=1}^{\itr} \stepsizen{\altitr}\deltan{\altitr}  +\frac{\kk   \cstnormbound^2}{2\strongly} \sum_{\altitr=1}^{\itr} \stepsizen{\altitr}^2 
,
\end{array}
\end{equation}
with
$
\sum_{\altitr=1}^{\infty}  \Expectation{\MDSn{\altitr}^2\cond \filtrationn{\altitr-1}  } 
\leq   
\left(4\kk  \cstnormbound\right)^2  \sum_{\altitr=1}^{\infty}   \stepsizen{\altitr}^2   
 <
 \infty 
$, 
and
\begin{equation}\label{nottightlyapunovdescentboundplus}
\begin{array}{l}
\lyapunov{\optmat}{\dualmatn{\itr+1}}
 \leq
\lyapunov{\optmat}{\dualmatn{\itr}}+  \MDSn{\itr} 
+ 4\kk^2\meanlipschitz\stepsizen{\itr}\deltan{\itr}   + \frac{\kk   \cstnormbound^2}{2\strongly} \stepsizen{\itr}^2
.
\end{array}
\end{equation}
Thus, $\supermartn{\itr}=\lyapunov{\optmat}{\dualmatn{\itr+1}}  +  \sum_{\altitr=\itr+1}^{\infty} [  4\kk^2\meanlipschitz\stepsizen{\altitr}\deltan{\altitr}   + \oneoverstrongly\kk   \cstnormbound^2\stepsizen{\altitr}^2/2 ]$ now defines the supermartingale with respect to $\{ \filtrationn{\itr}\}$.
\end{proof}
}{
\begin{proof}[Proof of \cref{thm:lasttrMXLMO}]
Following the line of thought of the proof of \cite[Theorem 5.1]{BLM18}, we first show that
\begin{equation} \label{liminfassumption}
\textstyle
\liminf_{\itr\to\infty} \Lyapunov{\optmat}{\dualmatn{\itr}} =0
\quad
\as \, .
\end{equation}
Next, we see that $\{\lyapunov{\optmat}{\dualmatn{\itr}}\}$ converges almost surely (\as{}) towards a finite quantity which, in view of \eqref{liminfassumption}, can only be $0$. \As{} convergence of $\{\matn{\itr}\}$ can then be inferred from \cref{lemma:fenchel}-\eqref{fenchelstronglyconvex}.
%
The assumption of  non-increasing $\{\deltan{\itr}\}$, together with \trconditionstepsize{} and \eqref{trconditionbias}, implies $\deltan{\itr}\downarrow 0$ which, in view of \eqref{perturbatorbound}, secures \as{} convergence of $\{\testmatn{\itr}\}$ as well.

First observe that \eqref{liminfassumption} holds if, almost surely, there exists a subsequence of $\{\matn{\itr}\}$ that converges towards $\optmat$.
Suppose this condition not to hold, and let $\limitpoints$ denote the set of the limit points of all  subsequences of $\{\matn{\itr}\}$. 
%
Then, almost surely, we have $\optmat\notin \limitpoints$ and, since $\limitpoints$ is closed 
by construction and $\ratefunction$ is continuous and strictly convex, 
$ \varrho \defeq \rate{\optmat} - \max_{\mat\in \limitpoints}{\rate{\mat}}>0$.
\obsolete{
we can write
\begin{equation}\label{compactset}\begin{array}{c}
\max_{\mat\in \limitpoints}{\rate{\mat}} = \rate{\optmat} - \varrho, \textup{ with }\varrho>0 \ \as{} \, .
\end{array}\end{equation}
}%
\obsolete{
\footnotetext{Indeed, suppose there is  a sequence $\{\compactmatn{\altitr}\}$ in $\limitpoints$ converging to $\compactmatn{\infty}\notin \limitpoints$. For any be any sequence of positive numbers $\{\epsilon_\altitr\}$ with $\epsilon_\altitr\downarrow 0$, one can find $\{\matn{\itr_\altitr}\}$, subsequence of $\{\matn{\itr}\}$, such that $\norm{\matn{\itr_\altitr}-\compactmatn{\altitr}}<\epsilon_\altitr$ for all $\altitr$. Hence, $\matn{\itr_\altitr}\to\compactmatn{\infty}$ and we reach the contradiction $\compactmatn{\infty} \in \limitpoints$.}%
}%

Because the conditions of \cref{lemma:norm} are met, there exists $ \cstnormbound <\infty$ such that $  \dualnorm{\estgradkn{\nbk}{\itr}}< \cstnormbound$ for all $k$.
By telescoping \eqref{descentargumenttwo}  in  \cref{lemma:globaldescentargument}\eqref{lemma:globaldescentargumenti}, we find
\begin{equation} \label{adabsurdumlyapunovdescentboundplus} 
\begin{array}{l}
\lyapunov{\optmat}{\dualmatn{\itr+1}}
 \leq
\lyapunov{\optmat}{\dualmatn{1}} - \sum_{\altitr=1}^{\itr} \stepsizen{\altitr} \left[\optrate-\rate{\matn{\altitr}}\right]
\quad\
\\
\hfill 
+  \sum_{\altitr=1}^{\itr}\MDSn{\altitr}  + 4\kk^2 \meanlipschitz \sum_{\altitr=1}^{\itr} \stepsizen{\altitr}\deltan{\altitr}  +\frac{\kk   \cstnormbound^2}{2\strongly} \sum_{\altitr=1}^{\itr} \stepsizen{\altitr}^2 
,
\end{array}
\end{equation}
where $\{\MDSn{\itr}\}$ is the difference sequence of a martingale  with respect to the filtration $\{\filtrationn{\itr}\}$.
\obsolete{

If we decompose $\estgradkn{\nbk}{\itr}$ as in \eqref{estimatordecomposition} and define $\MDSn{\itr}$ as in \eqref{MDSn}
, then it follows from \eqref{earlydeviationbound} that
\begin{equation} \label{modulusMDSn}
\textstyle
\modulus{\MDSn{\itr}} 
\leq 
2\stepsizen{\itr} \sum_{\nbk=1}^{\kk}  \dualnorm{  \devmatkn{\nbk}{\itr}} 
\leq 
4\kk   \cstnormbound \, \stepsizen{\itr}
.
\end{equation}
By telescoping \eqref{descentargumenttwo}, we find
\begin{equation}\label{adabsurdumlyapunovdescentboundplus}
\begin{array}{l}
\lyapunov{\optmat}{\dualmatn{\itr+1}}
\leq
\lyapunov{\optmat}{\dualmatn{1}} - \sum_{\altitr=1}^{\itr} \stepsizen{\altitr} \left[\optrate-\rate{\matn{\altitr}}\right]
\qquad
\\
\hfill 
+  \sum_{\altitr=1}^{\itr}  \MDSn{\altitr} 
+ 4\kk^2\meanlipschitz \sum_{\altitr=1}^{\itr} \stepsizen{\altitr}\deltan{\altitr}   +   \frac{ \kk \cstnormbound^2}{2\strongly}\sum_{\altitr=1}^{\itr} \stepsizen{\altitr}^2 
,
\end{array}
\end{equation}  
where $\MDSn{\altitr}= \stepsizen{\altitr} \sum_{\nbk=1}^{\kk}   \inner{ \devmatkn{\nbk}{\altitr}}{\matkn{\nbk}{\altitr}-\optmatk{\nbk}} $
satisfies \martingaledifference{} 
and \eqref{modulusMDSn}, so that $\{\MDSn{\itr}\}$ is the difference sequence of a martingale  with respect to the filtration $\{\filtrationn{\itr}\}$. 

}%
In view of assumptions \trconditionsquared{} and \eqref{trconditionbias},  the last two terms in the second member of \eqref{adabsurdumlyapunovdescentboundplus} converge as $\itr\to\infty$. As for the third term, since
\begin{equation}\label{sumMDSboundplus}
\textstyle
\sum_{\altitr=1}^{\infty}  \Expectation{\MDSn{\altitr}^2\cond \filtrationn{\altitr-1}  } 
\leq   
\left(4\kk  \cstnormbound\right)^2  \sum_{\altitr=1}^{\infty}   \stepsizen{\altitr}^2   
 \refereq{\trconditionsquared}{<} \infty ,
\end{equation}
\cite[Theorem 2.18]{hall80} applies with parameter $p=2$, and it follows that $\sum_{\altitr=1}^{\itr} \MDSn{\altitr}$ converges \as{} as $\itr\to\infty$. 
Finally, one can find a subsequence $\{\matn{\itr_\altitr}\}$ that converges to a point of $\limitpoints$ and thus satisfies $  \rate{\optmat} - \rate{\matn{\itr_\altitr}} >  \varrho/2 $ for $\altitr$ large enough. It follows from \trconditionstepsize{} that the second term  $\sum_{\altitr=1}^{\infty} \stepsizen{\itr} [   \rate{\matn{\itr}} - \rate{\optmat} ] \to-\infty$.
All in all we find that $\Lyapunov{\optmat}{\dualmatn{\itr}}\to-\infty$ \as{}, which is in contradiction with the nonnegativity of $\lyapunovfunction$. We infer that \eqref{liminfassumption} is true.
%

It remains to show that $\{\lyapunov{\optmat}{\dualmatn{\itr}}\}$ is almost surely convergent.
%
%
To do so we rely on
Doob's  convergence theorem for supermartingales \cite[Theorem 2.5]{hall80}.
Recalling \eqref{descentargumenttwo}, and using $  \dualnorm{\estgradkn{\nbk}{\itr}}^2< \cstnormbound$ and $\rate{\matn{\itr}} - \rate{\optmat}  \leq 0$, we find
\begin{equation}\label{nottightlyapunovdescentboundplus}
\begin{array}{l}
\lyapunov{\optmat}{\dualmatn{\itr+1}}
 \leq
\lyapunov{\optmat}{\dualmatn{\itr}}+  \MDSn{\itr} 
+ 4\kk^2\meanlipschitz\stepsizen{\itr}\deltan{\itr}   + \frac{\kk   \cstnormbound^2}{2\strongly} \stepsizen{\itr}^2
.
\end{array}
\end{equation}
Let $\supermartn{\itr}=\lyapunov{\optmat}{\dualmatn{\itr+1}}  +  \sum_{\altitr=\itr+1}^{\infty} [  4\kk^2\meanlipschitz\stepsizen{\altitr}\deltan{\altitr}   + \oneoverstrongly\kk   \cstnormbound^2\stepsizen{\altitr}^2/2 ]$. 
Note that, under assumptions \trconditionsquared{} and \eqref{trconditionbias}, $\supermartn{0}$ is  finite by construction.
We infer from \\martingaledifference{} 
and \eqref{nottightlyapunovdescentboundplus} that $\expectation{\supermartn{\itr}\cond \filtrationn{\itr-1} } \leq \supermartn{\itr-1}$ for     $\itr\geq 1$,
\obsolete{
\begin{equation} \label{supermartingale}
\expectation{\supermartn{\itr}\cond \filtrationn{\itr-1} } \leq \supermartn{\itr-1} \quad( \itr\geq 1)
\end{equation}
}%
and $\{\supermartn{\itr}\}$ is a supermartingale with respect to $\{ \filtrationn{\itr}\}$, thus satisfying $\Expectation{\supermartn{\itr} } \leq \supermartn{0} <\infty$.
\obsolete{
Under assumptions \trconditionsquared{} and \eqref{trconditionbias}, we find, for $\itr\geq2$,
\begin{equation}
\nocolsep
\begin{array}{rcl}
\Expectation{\supermartn{\itr} } 
&=&
 \Expectation{\cdots\Expectation{\Expectation{\supermartn{\itr}\cond \filtrationn{\itr-1} }\cond \filtrationn{n-2} }\cdots\cond \filtrationn{0} } 
\\&\leq& 
 \supermartn{0} 
\\ &=&
\lyapunov{\optmat}{\dualmatn{1}}  + 4\kk^2\meanlipschitz \sum_{\altitr=1}^{\infty} \stepsizen{\altitr}\deltan{\altitr}   + \frac{\kk   \cstnormbound^2}{2\strongly} \sum_{\altitr=1}^{\infty}  \stepsizen{\altitr}^2
\\ &<&
\infty
.
\end{array}
\end{equation}
}%
Hence, $\{\supermartn{\itr}\}$ is uniformly $L^1$-bounded and Doob's theorem applies.  It follows that $\{\supermartn{\itr}\}$, and consequently $\{\lyapunov{\optmat}{\dualmatn{\itr}}\}$, are almost surely convergent, which completes the proof.
\end{proof}
}%
\begin{proof}[Proof of \cref{timecomplexity:trMXLMO}]
\eqref{thm:trMXLMOi}
\obsolete{
Since  under our assumptions \cref{lemma:norm} applies, \eqref{threesecondmomentcomputation} gives us the bound
\begin{equation}\label{totalnormbound}
\textstyle
\sum_{\nbk=1}^{\kk}   \dualnorm{\estgradkn{\nbk}{\itr}}^2
\refereq{\eqref{threesecondmomentcomputation}}{\leq}
\frac{\kk}{4} \maxdimension^2 \, 
 [\normbound{\alpha}{\beta}{\stepsize}{\delta}]^2
 ,
\end{equation}
where $\normbound{\alpha}{\beta}{\stepsize}{\delta}$ is given by \eqref{normbound}.
By combining \eqref{totalnormbound} with \eqref{expectationyergodicaverageidentity} in \cref{lemma:globaldescentargument}\eqref{lemma:globaldescentargumentii}, we find, for $\vect{\stepsizen{\itr},\deltan{\itr}}=\vect{\cststepsize\,\itr^{-\alpha},\cstdelta\,\itr^{-\beta}}$,
}%
%
By combining 
the uniform bound in \cref{lemma:norm}
with \eqref{expectationyergodicaverageidentity} in \cref{lemma:globaldescentargument}\eqref{lemma:globaldescentargumentii} and using $\lyapunov{\optmat}{\dualmatn{1}} \leq \kk\log\maxmm  $, we find, for the  policy $\vect{\stepsizen{\itr},\deltan{\itr}}=\vect{\cststepsize\,\itr^{-\alpha},\cstdelta\,\itr^{-\beta}}$,
\begin{equation}\label{ergodicaverageidentitythree}
\begin{array}{l}
\optrate  -\bigexpectation{ \rate{ \meanmatn{\itr}}}
\leq
\frac{
 \kk\log\maxmm
}{
\stepsize\sum_{\altitr=1}^{\itr} \altitr^{-\alpha} 
} 
+ 4\kk^2\meanlipschitz\delta\frac{
\sum_{\altitr=1}^{\itr}\altitr^{-\alpha-\beta} 
}{
\sum_{\altitr=1}^{\itr} \altitr^{-\alpha} 
}
\qquad \quad\\\hfill
  + \frac{\kk \maxdimension^2[\normbound{\alpha}{\beta}{\stepsize}{\delta}]^2 \stepsize}{8\strongly} \frac{
 \sum_{\altitr=1}^{\itr} \altitr^{-2\alpha}  
}{
\sum_{\altitr=1}^{\itr} \altitr^{-\alpha} 
} 
,
\end{array}
\end{equation}
where $\normbound{\alpha}{\beta}{\stepsize}{\delta}$ is given by \eqref{normbound}.
The above upper bound is minimized for $\alpha=\beta=1/2$, in which case we find \eqref{asymptoticrate}.

 \eqref{thm:trMXLMOiinew}
 \renewcommand{\dimensionk}[1]{\maxdimension}%
\renewcommand{\meansqdimension}{\maxdimension}%
\renewcommand{\meandimension}{\maxdimension}%
Using  \cref{lemma:globaldescentargument}\eqref{lemma:globaldescentargumentiii}  under $\vect{\cststepsize,\cstdelta}=\vect{\stepsize/\sqrt{T},\delta/\sqrt{T}}$ and with the 
bounds  
$\vbound= \frac{1}{2}
\,\normbound{0}{0}{\cststepsize}{\cstdelta} $, given by \cref{lemma:norm}, 
yields
\begin{equation} \label{ergodicaverageidentityconstant} 
\textstyle
\optrate -\bigexpectation{ \rate{ \meanmatn{T}} }
 \refereq{\eqref{standardtroptrateboundmeanexpectation}}{\leq}
 \frac{\kk\log\maxmm}{\stepsize\sqrt{T} } 
+ \frac{4\kk^2\meanlipschitz\delta  }{\sqrt{T}}
+ \frac{  \kk \dimensionk{\nbk}^2  [
\,\normbound{0}{0}{\frac{\stepsize}{\sqrt{T}}}{\frac{\delta}{\sqrt{T}}}]^2\stepsize }{8\strongly\sqrt{T} }
\end{equation}
where $
\textstyle
\normbound{0}{0}{\cststepsize}{\cstdelta}
=
\big(  \frac{ \fourtwitchFlipschitz     \strongly}{\sqrt{\maxdimension}  } \big) \big( \frac{2\strongly}{\maxdimension \Lipschitz\kk  } -  \frac{\stepsize  }{\delta} \big)^{-1} 
$,
and
 \begin{equation} 
 \label{basicazumabissecond}
\textstyle
\Bigproba{\frac{\sum_{\itr=1}^{T}  \MDSn{\itr}  }{ \sqrt{T}\,\stepsize  }  \leq \zmargin }
\refereq{\eqref{standardazumabis}}{\geq}
1-\Bigexponential{-\frac{T\zmargin^2}{8   \kk^2 \meandimension^2  [  
\,\normbound{0}{0}{\frac{\stepsize}{\sqrt{T}}}{\frac{\delta}{\sqrt{T}}}  ]^2 }}
.
\end{equation}   
We find \eqref{trsimplerateboundmeanexpectation}  after substituting $\normbound{0}{0}{\cststepsize}{\cstdelta}$ with its actual value in \eqref{ergodicaverageidentityconstant}. Then,   \eqref{trsimplerateboundmeaninterval} follows from \eqref{trsimplerateboundmeanexpectation} and \eqref{basicazumabissecond}.
\end{proof} 

The proof of Corollary \ref{convergenceratres:trMXLMO} relies on the following lemma.

\begin{lemma}\label{lemma:optimization}
Let $\Gamma = \{ (\stepsize,\delta)\in\Realplus\times\Realplus : \stepsize< h \delta \}$ and consider the function $\genericfunction:\Gamma \mapsto\Real$ defined by
\begin{equation} \label{newresidualfunction}
\textstyle
\generic{\stepsize,\delta} = \frac{a}{\stepsize} + 2 b \delta + c \stepsize \, \big( h - \frac{ \stepsize }{\delta }\big)^{-2} + d\, \big( h - \frac{ \stepsize }{\delta }\big)^{-1},
\end{equation}
\obsolete{
\begin{equation} 
\textstyle
\mom{\stepsize}{\delta} = 
\left( h - \frac{ \stepsize }{\delta }\right)^{-1}
,
\end{equation}
}%
where $a,b,c,h>0 $ and $d\geq 0$ are given parameters. 
 
\begin{enumerate}[(i),ref=\roman*,wide,labelwidth=!,labelindent=5pt]
 \item \label{lemma:optimizationi} 
At the point $(\stepsize^*,\delta^*)\in\Gamma$, where
\begin{equation} \label{earlyoptimalpolicyMO}
\textstyle
\stepsize^* {\,=\,}
{h}\, \bigg[{\sqrt{\frac c a}+ {\sqrt{{\sqrt{\frac c a}} \Big(\frac{ 2  b h  }{2\sqrt{ac}+d } \Big) }}}\bigg]^{-1}
,
\ \ \
\delta^* {\,=\,}
\sqrt{\sqrt{\frac{a}{c}}\Big(\frac{2\sqrt{ac}+d }{  2  b h    }\Big)}
,
\end{equation}
the value of $f$ is given by
\begin{equation} \label{boundfMO}
\textstyle
\generic{\stepsize^*,\delta^*}
 =
%
\frac{2\sqrt{ac}+d}{h} +2 \sqrt{ 2b\sqrt{\frac{a}{c}} \Big( \frac{2\sqrt{ac}+d}{h} \Big)  }
.
\end{equation}
Under the constraint $\delta/\sqrt{\itr}<\radius$, where $\radius>0$, \eqref{boundfMO} holds for 
\begin{equation} \label{conditionondeltaMO}
\textstyle
\itr >
 \sqrt{\frac a c}\,\Big(\frac{2\sqrt{ac}+d}{2bh}\Big)\, \radius^{-2} .
\end{equation}
%

\item \label{lemma:optimizationii}
For any $\errortolerance{\,>\,}0$, $\generic{\stepsize^*,\delta^*}/\sqrt{\itr}\leq\errortolerance$ holds for $\itr\geq T$ if $T{\,=\,} [\generic{\stepsize^*,\delta^*}/\errortolerance]^{2}$.
\obsolete{
\begin{equation} \label{optimaltimeMO}
\nocolsep
\begin{array}{c} 
T= \left[ \frac{2\sqrt{ac}+d}{h} +2 \sqrt{ 2b\sqrt{\frac{a}{c}} \left[ \frac{2\sqrt{ac}+d}{h} \right]  } \right]^2 \errortolerance^{-2}
.
\end{array}
\end{equation}.
}%
The constraint $\delta^*{/}\sqrt{T}{\,<\,}\radius$ 
then rewrites as
\begin{equation} \label{newconditionondeltaMO}
\textstyle
%
\errortolerance < \bigg(4b + \sqrt{2({b}/{h}) \sqrt{{c}/{a}} \, ( 2\sqrt{ac}+d) } \bigg) \, r
.
\end{equation}
\end{enumerate}
\end{lemma}
\begin{proof}  Verification of all the claims is straightforward.
\end{proof}

\obsolete{
\begin{proof}
\eqref{lemma:optimizationi} 
 Rewrite \eqref{newresidualfunction} as $f(\stepsize,\delta)=g(\stepsize,\delta,\mom{\stepsize}{\delta})$, where
$g(\stepsize,\delta,\mome) = {a}/{\stepsize} + 2 b \delta + c \stepsize \mome^{2}
+  d  \mome
$.
\obsolete{
\begin{equation} \label{rewrittenresidualfunction}
g(\stepsize,\delta,\mome) = \frac{a}{\stepsize} + 2 b \delta + c \stepsize \mome^{2}
+  d  \mome.
\end{equation}
}%
Function $g$ is convex and its coordinate minima with respect to direction $\stepsize$ lie at the points $(\tilde\stepsize(\mome),\delta,\mome)$ such that $\tilde\stepsize(\mome) = \mome^{-1} \sqrt{a/c}$. Under the constraint $\mome=\mom{\tilde\stepsize(\mome)}{\delta}$, which rewrites as $\delta=\tilde\delta(\mome) \defeq \tilde\stepsize(\mome)/(h-\mome^{-1}) =  (h\mome-1)^{-1}\sqrt{a/c}\, $, we find 
\begin{equation}
\textstyle
g(\tilde\stepsize(\mome),\tilde\delta(\mome),\mome) 
=
 (2\sqrt{ac}+d)\mome  + \frac{ 2 b \sqrt{a/c} }{h\mome-1} 
,
\end{equation}
which takes its minimal value for
\begin{equation}\textstyle
\mome^* =
\frac 1 h \left( 1+ \sqrt{ \sqrt{\frac{a}{c}}\frac{ 2  bh  }{2\sqrt{ac}+d } } \right)  
.
\end{equation}
\obsolete{
\begin{equation}
g(\tilde\stepsize(\mome^*),\tilde\delta(\mome^*),\mome^*) 
= 
  \frac{2\sqrt{ac}+d}{h} +2 \sqrt{ 2b\sqrt{\frac{a}{c}} \left[ \frac{2\sqrt{ac}+d}{h} \right]  }    
\end{equation}
}
Setting $\stepsize^*=\tilde\stepsize(\mome^*)$ and $\delta^*=\tilde\delta(\mome^*)$ yields \eqref{earlyoptimalpolicyMO} and \eqref{boundfMO}.
The additional constraint $\delta/\sqrt{\itr}<\radius$ is satisfied by \eqref{earlyoptimalpolicyMO} for as long as \eqref{conditionondeltaMO}
holds.
\obsolete{
\begin{equation} 
\itr^{-\frac 1 2} =  \frac{e \generic{\stepsize^*,\delta^*}}{ \sqrt{2\sqrt{ac}+d} \left[\sqrt{2\sqrt{ac}+d}  + 2 \sqrt{ 2b\sqrt{\frac a c} } \right] }  
.
\end{equation}
\begin{equation} 
\begin{array}{l}
\stepsize^* =  
 \frac{e^2  \sqrt{\frac a c} \generic{\stepsize^*,\delta^*}}{ 
2\sqrt{ac}+d
+
 4b\sqrt{e} 
+
\sqrt{2\sqrt{ac}+d}   \left[  \sqrt{  2  b e \sqrt{\frac c a}  }  + 2 \sqrt{ 2b\sqrt{\frac a c} }\right]
}  
,
\\
\delta^* =  \frac{\sqrt{e}}{4b +\sqrt{ 2  b  \sqrt{\frac c a}(2\sqrt{ac}+d)}  }  \,  \generic{\stepsize^*,\delta^*}
,
\end{array}  
\end{equation}
\begin{equation} 
\begin{array}{l}
\stepsize^* =
 \frac{ e }{  \sqrt{ \sqrt{\frac c a} \left(\frac{2\sqrt{ac}+d }{ 2  b e  } \right) }+1 }
\,\delta^*,
\\
\delta^*  \left[ \sqrt{\frac c a} \left(\frac{ 2  b e  }{2\sqrt{ac}+d } \right) \right]^{\frac 1 2} = \itr^{-\frac 1 2} 
,
\end{array}  
\end{equation}
}
Claim \eqref{lemma:optimizationii} follows from \eqref{earlyoptimalpolicyMO} and \eqref{boundfMO} after setting  $\generic{\stepsize^*,\delta^*}=\sqrt{T}\, \errortolerance$.
\obsolete{
Set $\generic{\stepsize^*,\delta^*}=\sqrt{T}\, \errortolerance$ in \eqref{boundfMO}. 
We get \eqref{optimaltimeMO} by solving the result for $T$ as a function of $\errortolerance$, and the claim follows.
Besides, by combining \eqref{earlyoptimalpolicyMO} with \eqref{optimaltimeMO}, we find
\begin{equation} \label{optimalpolicyMO}
\begin{array}{c}
\frac{\delta^*}{\sqrt{T}} = 
\frac{ \sqrt[4]{{a}/{c}} }{\sqrt{2b(2\sqrt{ac}+d)/h} +4b \sqrt[4]{{a}/{c}} } \,  \errortolerance
,
\end{array}
\end{equation}
and the constraint $\delta^*/\sqrt{T}<\radius$ rewrites as \eqref{newconditionondeltaMO}.
}%
\end{proof}

}%

\begin{proof}[Proof of Corollary \ref{convergenceratres:trMXLMO}]
\eqref{cr:trMXLMOi}
To derive $\stepsize$ and $\delta$  in \eqref{cr:trMXLMOi} it suffices to apply \cref{lemma:optimization}\eqref{lemma:optimizationi} to the expression for $\Bcoefficient{\stepsize}{\delta}$ given in  \cref{timecomplexity:trMXLMO}\eqref{thm:trMXLMOiifirst}. The convergence rate of $\expectation{ \rate{ \meanmatn{T}} }$ follows from \eqref{troptrateboundmeanexpectation} and \eqref{boundfMO}, while the condition on $T$ is a translation of \eqref{conditionondeltaMO} into the present setting, where the restriction $\delta/\sqrt{\altitr}<\radiusk{\nbk}$ for all $\nbk$ applies, with $\radiusk{\nbk}$ given by \eqref{radius}. 

\eqref{cr:trMXLMOii}
Recall \cref{timecomplexity:trMXLMO}\eqref{thm:trMXLMOiisecond}.  The second part of \eqref{troptrateboundmeaninterval} rewrites as $1-\errorproba$ for
$ \textstyle
 \zmargin  
 = 
16 \halftwitchFlipschitz  \strongly    \left[ \frac{2\strongly}{ \Lipschitz\kk (\maxmm^2-1) } -  \frac{\stepsize  }{\delta} \right]^{-1}  \sqrt{2\,\logarithm{\frac 1 \errorproba}\kk^2(\maxmm^2-1) / T}      
 .   
$
Observe that $\Bcoefficient{\stepsize}{\delta}+\zmargin\sqrt{T}$ is an instance of the function $\generic{\stepsize,\delta}$ defined in \eqref{newresidualfunction}. \cref{lemma:optimization}\eqref{lemma:optimizationii} gives us a condition on $T$ for $\Bcoefficient{\stepsize}{\delta}/\sqrt{\ceil T}+\zmargin\leq \errortolerance$ to be true which, in view of \eqref{troptrateboundmeaninterval}, is also sufficient for \eqref{standardtroptrateboundmeaninterval} to hold.
After computations we find 
the value of $T$ in Table \ref{table:convergenceratres:trMXLMO}\ref{cr:trMXLMOii} 
with the restriction on $\errortolerance$:
\begin{multline*}
\txs
\errortolerance 
\refereq{\!\eqref{newconditionondeltaMO}\!}{{\,<\,}}
%
%
%
%
2^{\frac{9}{4}}\meanlipschitz 
	\bracks*{\renrmlzdcsteps{\errorproba}
	\sqrt{\halftwitchFlipschitz \Lipschitz / \meanlipschitz}  {+} \frac{2^{\frac 3 4}  }{ \sqrt[4]{\log(1/\errorproba)}  [\maxmm^2-1]} }
	\\
	\times \sqrt{\frac{ \sqrt{\log(1/\errorproba)} \kk^4 [\maxmm{+}1][\maxmm^2{-}1]}{\maxmm}}.
\qedhere
\end{multline*}
\end{proof}

\subsection{Analysis of the \randomMXLzeroplus{} algorithm}
\label{appendix:randomMXLzeroplus}

Lemmas \ref{lemma:bias} and \ref{lemma:norm} still apply in the asynchronous setting. 
Instead of \eqref{lyapunov} we use the Lyapunov function
\begin{equation} \label{randomlyapunov}
\textstyle
\lyapunovprob{\prob}{\optmat}{\dualmat} = \sum_{\nbk=1}^{\kk} \frac{1}{\probk{\nbk}}\fenchel{\optmatk{\nbk}}{\dualmatk{\nbk} } ,
\end{equation}
where $\optmat\in\optMat$ is a solution.
Proceeding as for the derivation of \eqref{earlydescentargumenttwo}    in  \cref{lemma:globaldescentargument}, we find, for the algorithm \eqref{randomMXL},
\begin{equation} \label{randomdescentargumenttwo} 
\begin{array}{l}
 \lyapunovprob{\prob}{\optmat}{\dualmatn{\itr+1}}
 \leq
\lyapunovprob{\prob}{\optmat}{\dualmatn{\itr}} 
-\stepsizen{\itr} \left[\optrate-\rate{\matn{\itr}} \right] 
\qquad  \quad 
\\
\hfill 
+ 4\kk^2\meanlipschitz \stepsizen{\itr}\deltan{\itr} 
+  \randomMDSn{\itr}
+\sum_{\nbk=1}^{\kk} \frac{\stepsizen{\itr}^2}{2\strongly} \dualnorm{\estgradkn{\nbk}{\itr}}^2
,
\end{array}
\end{equation}
with the random sequence $\{\randomMDSn{\itr} \}$   now given by
\begin{equation}\begin{array}{c} \label{randomMDSn}
\randomMDSn{\itr} 
= 
 \stepsizen{\itr}\sum_{\nbk\in\Updatingn{\itr}} \inner{ \devmatkn{\nbk}{\itr} }{\matkn{\nbk}{\itr}-\optmatk{\nbk}}  \qquad \qquad \qquad \qquad \quad
\\
\hfill 
+ \sum_{\nbk=1}^{\kk}  \frac{\indicator{\Updatingn{\itr}}{\nbk}-\probk{\nbk}}{\probk{\nbk}} \, \big[ \stepsizen{\itr} \inner{  \estgradkn{\nbk}{\itr} }{\matkn{\nbk}{\itr}-\optmatk{\nbk}} + \frac{\stepsizen{\itr}^2}{2} \dualnorm{\estgradkn{\nbk}{\itr}}^2  \big]
.
\end{array}
\end{equation}
It is easily seen that $\expectation{\randomMDSn{\itr}   \cond \filtrationn{n-1}   }= 0 $, and
\begin{equation} \label{randommodulusMDSn}
\begin{array}{l}
\nocolsep
\modulus{\randomMDSn{\itr}} 
\leq  
 2\stepsizen{\itr}\sum_{\nbk=1}^{\kk}  \dualnorm{ \devmatkn{\nbk}{\itr} }
\qquad\qquad\qquad\qquad\qquad\qquad\\ \hfill 
+ \sum_{\nbk=1}^{\kk} \max\big(1,\frac{1}{\probk{\nbk}}-1\big) \, \big[ 2\stepsizen{\itr} \dualnorm{  \estgradkn{\nbk}{\itr} } + \frac{\stepsizen{\itr}^2}{2} \dualnorm{\estgradkn{\nbk}{\itr}}^2  \big].
\end{array} 
\end{equation}
Compare \eqref{randomdescentargumenttwo},\eqref{randommodulusMDSn} with \eqref{descentargumenttwo},\modulusmartingaledifference{}. 
By reproducing the rationale behind the proof of \cref{lemma:globaldescentargument}, we obtain an asynchronous 
counterpart to \cref{lemma:globaldescentargument}, where \eqref{expectationyergodicaverageidentity} and \eqref{standardtroptrateboundmeanexpectation} now hold with $\lyapunovprobfunction{\prob} $ in place of $ \lyapunovfunction $, and \eqref{standardazumabis} becomes
\obsolete{
\begin{equation}
\begin{array}{l}
\nocolsep
\modulus{\randomMDSn{\itr}} 
\leq
2\cststepsize\sum_{\nbk=1}^{\kk}  \dualnorm{ \devmatkn{\nbk}{\itr} }
\hfill
\\
\quad
+ \sum_{\nbk=1}^{\kk} \max\left(1,\frac{1}{\probk{\nbk}}-1\right) \left[ 2\cststepsize \dualnorm{  \estgradkn{\nbk}{\itr} } + \frac{\cststepsize^2}{2} \dualnorm{\estgradkn{\nbk}{\itr}}^2  \right].
\\
\leq
4\vbound\cststepsize\kk\meandimension
\hfill
\\
\quad
+  \sum_{\nbk=1}^{\kk} \max\left(1,\frac{1}{\probk{\nbk}}-1\right) \left[ 2\vbound\cststepsize \dimensionk{\nbk} + \frac{\vbound^2\cststepsize^2}{2} \dimensionk{\nbk}^2  \right].
\\
=
2 (2+ \probdist{\prob}) \cststepsize\kk\maxdimension\vbound
+  \frac{\cststepsize^2\probdist{\prob}}{2} \kk\maxdimension^2 \vbound^2 .
\\
=
2 [(2+ \probdist{\prob}) \cststepsize\kk\maxdimension\vbound
+  \frac{\cststepsize^2\probdist{\prob}}{4} \kk\maxdimension^2 \vbound^2 ].
\end{array}
\end{equation}
}%
\begin{equation} \label{randomstandardazuma}
\begin{array}{c}
\Bigproba{\frac{\sum_{\itr=1}^{T}  \randomMDSn{\itr}  }{ T\cststepsize  }   \geq \zeta   }
\leq 
%
%
\Bigexponential{-\frac{T\zeta^2}{    8\kk^2\maxdimension^2     \left[ (1+\frac{\probdist{\prob}}2)    \normbound{0}{0}{\cststepsize}{\cstdelta} + \frac{\probdist{\prob}}{8} \cststepsize      \maxdimension \, [\normbound{0}{0}{\cststepsize}{\cstdelta}]^2  \right]^2  }}
 ,
\end{array}
\end{equation}
where $\probdist{\prob}$ is defined as in \cref{timecomplexity:randomtrMXLMO}. 
\begin{proof}[Proof of \cref{timecomplexity:randomtrMXLMO}]
Proceed as in the proof of \cref{timecomplexity:trMXLMO}.
\end{proof}
%
%
\obsolete{
$\normbound{0}{0}{\cststepsize}{\cstdelta} =   \frac{  \fourtwitchFlipschitz    \strongly}{\sqrt{\maxdimension}  }  \left( \frac{2\strongly}{\maxdimension \Lipschitz\kk  } -  \frac{\cststepsize  }{\cstdelta} \right)^{-1}$

\begin{equation} \label{randomstandardazuma}
\begin{array}{c}
\Proba{\frac{\sum_{\itr=1}^{T}  \randomMDSn{\itr}  }{ T\cststepsize  }   \geq \zeta   }
 \hspace{55mm}
 \\ 
 \hfill
\leq 
\Exponential{-  {\scriptstyle  \left[   \frac{ \halftwitchFlipschitz(1+\frac{\probdist{\prob}}{2}) \strongly   }{\big( \frac{2\strongly}{(\maxmm^2-1) \Lipschitz\kk  } -  \frac{\cststepsize  }{\cstdelta} \big)}  +     \frac{  \squaredhalftwitchFlipschitz     \strongly \probdist{\prob}      \sqrt{\maxmm^2-1}  \cststepsize }{\big( \frac{2\strongly}{(\maxmm^2-1) \Lipschitz\kk  } -  \frac{\cststepsize  }{\cstdelta} \big)^{2}}  \right]^{-2}  } \frac{T\zeta^2}{    2^9 \kk^2 (\maxmm^2-1)      }  }
 \\ 
 \hfill
\leq 
\Exponential{-  {\scriptstyle \frac{\big( \frac{2\strongly}{(\maxmm^2-1) \Lipschitz\kk  } -  \frac{\cststepsize  }{\cstdelta} \big)^{2}}{ \left[ \halftwitchFlipschitz    \forcestrongly+\halftwitchFlipschitz \frac{\probdist{\prob}}{2} \strongly    +     \frac{  \squaredhalftwitchFlipschitz     \strongly \probdist{\prob}      \sqrt{\maxmm^2-1}  \cststepsize }{\big( \frac{2\strongly}{(\maxmm^2-1) \Lipschitz\kk  } -  \frac{\cststepsize  }{\cstdelta} \big)}  \right]^{2}  } } \frac{T\zeta^2}{    2^9 \kk^2 (\maxmm^2-1)      }  }
 ,
\end{array}
\end{equation}
where $\probdist{\prob}=\frac 1 \kk\sum_{\nbk=1}^{\kk}   \max(1,\frac{1}{\probk{\nbk}}-1) $.
}%
%

%
%
\begin{proof}[Proof of Corollary \ref{convergenceratres:randomtrMXLMO}]
First observe that  we have $\probdist{\prob}=\kk-1$ if $\kk\geq2$. The rest of the proof bases on the conclusions of \cref{timecomplexity:randomtrMXLMO} 
and follows the exact lines of the proof of Corollary \ref{convergenceratres:trMXLMO}, now using $\Bcoefficientprob{\prob}{\stepsize}{\delta}$ and $\Ccoefficientprob{\prob}{\stepsize}{\delta}$. Note that Corollary \ref{convergenceratres:randomtrMXLMO}\ref{cr:randomtrMXLMOii}
holds with the following restriction on $\errortolerance$:
\begin{multline*}
\errortolerance 
{\refereq{\eqref{newconditionondeltaMO}}{{\,<\,}}}
%
%
%
%
2^{3} \meanlipschitz
	\bracks*{ \renrmlzdcstepsprob{\prob}{\errorproba}
	\sqrt{\renrmlzdstuff{\errorproba}  \sqrthalftwitchFlipschitz  \sqrtstrongly     \squaredhalftwitchFlipschitz \strongly    \Lipschitz / \meanlipschitz } + \big[{\powlogarithm{{1/ \errorproba}}{\frac 3 8} \kk^{\frac{3}{4}}\dimension}\big]^{-1} }
	\\
	\times \sqrt{ \frac{\powlogarithm{1/\errorproba}{\frac 3 4}  \kk^{\frac{11}{4}}[\maxmm+1][\maxmm^2{-}1] }{\maxmm} }
.
\qedhere
\end{multline*}
\end{proof}


\ifCLASSOPTIONcaptionsoff
  \newpage
\fi

  

\bibliographystyle{IEEEtran}
\bibliography{bibtex/IEEEabrv,bibtex/obinria,bibtex/Bibliography-PM,bibtex/verobiblio}

 
\begin{biographies}
 
\begin{IEEEbiographynophoto}
{Olivier Bilenne}
received the M.Sc. (Ingénieur Civil \'{E}lectricien) degree from the University of Liège, Belgium, in 2002, and the  Dr.-Ing. degree from the Technical University of Berlin, Germany, in 2015. 

Since 2018 he has been a CNRS Postdoctoral Researcher with the Inria-POLARIS Team at the Laboratoire d’Informatique de Grenoble, France. He previously completed postdoctoral stays at Aalto University, Espoo, Finland (2016--17) and at the KTH Royal Institute of Technology in Stockholm, Sweden (2017--18). From 2003 to 2009, he was with Multitel Research Centre, Mons, Belgium. Starting fall 2020, he will be affiliated with the Department of Data Science and Knowledge Engineering at Maastricht University, Netherlands. His research interests are in optimization, learning, game theory, networks, signal processing, and queueing systems.
\end{IEEEbiographynophoto}

\begin{IEEEbiographynophoto}
{Panayotis Mertikopoulos} (M'11) graduated valedictorian from the Physics Department of the University of Athens in 2003. He received his M.Sc. and M.Phil. degrees in mathematics from Brown University in 2005 and 2006, and his Ph.D. degree from the University of Athens in 2010. During 2010–2011, he was a post-doctoral researcher at the \'{E}cole Polytechnique, Paris, France. Since 2011, he has been a CNRS Researcher at the Laboratoire d’Informatique de Grenoble, Grenoble, France, and has held visiting positions at UC Berkeley and EPFL. Since 2020, he holds a joint affiliation with Criteo AI Lab as a principal researcher.

P. Mertikopoulos was an Embeirikeion Foundation Fellow between 2003 and 2006, and received the best paper award in NetGCoop '12, and spotlight awards at ICLR 2020 and NeurIPS 2020. He is serving on the editorial board and program committees of several journals and conferences on learning and optimization (such as NeurIPS, ICML and ICLR). His main research interests lie in learning, optimization, game theory, and their applications to networks, machine learning, and signal processing.
\end{IEEEbiographynophoto}

\begin{IEEEbiographynophoto}
{E.~Veronica Belmega}(S'08-M'10-SM'20) has been an Associate Professor (MCF HDR) with ENSEA graduate school since 2011 and Deputy Director of ETIS laboratory since 2020, Cergy, France.

She received the M.Sc. (eng. dipl.) degree from the University Politehnica of Bucharest, Romania, in 2007, and the M.Sc. and Ph.D. degrees both from the University Paris-Sud 11, Orsay, France, in 2007 and 2010. She received the HDR habilitation degree from the University of Cergy-Pontoise in 2019. From 2010 to 2011, she was a Post-doctoral researcher in a joint project between Princeton University, N.J., USA  and Sup\'elec, France. In 2015-2017, she was a visiting researcher at Inria, Grenoble, France.

Her research interests lie in convex optimization, game theory and machine learning applied to distributed networks. Dr. E. Veronica Belmega received the L'Or\'eal - UNESCO - French Academy of Science national fellowship in 2009. She served as Executive Editor of Trans. on Emerging Telecommun. Technologies (ETT) in 2016-2020; distinguished among the Top Editors 2016-2017. From 2018 until 2022, she receives the Doctoral Supervision and Research Bonus by the French National Council of Universities.
\end{IEEEbiographynophoto}

%
%
%

\end{biographies}

\end{document}